\documentclass{PapER}
\usepackage{mathtools} 
\usetikzlibrary{fit,shapes}

\newcommand{\dhat}[1]{\underline{\hat\delta}^{(#1)}(\epsilon)}
\newcommand{\esgj}{\epsilon^{\tau_L,\tau_R}_{SG}}
\newcommand{\esg}[2]{\epsilon^{#1,#2}_{SG}}
\newcommand{\epslr}{\epsilon^*_{l,r}}
\newcommand{\udi}{\underline\delta_I}
\newcommand{\udo}{\underline\delta_O}
\newcommand{\udv}{\underline d_V}
\newcommand{\udc}{\underline d_C}
\newcommand{\udel}{\underline\delta}
\newcommand{\Delf}[1]{\Delta(#1,\underline\delta)}

\newcommand{\DelfL}[1]{\Delta_L(#1,\underline\delta)}
\newcommand{\DelfR}[1]{\Delta_R(#1,\underline\delta)}
\newcommand{\DelfLk}[2]{\Delta^{(#1)}_L(#2,\underline\delta)}
\newcommand{\DelfRk}[2]{\Delta^{(#1)}_R(#2,\underline\delta)}

\newcommand{\Bleft}{B_\mathrm{left}}
\newcommand{\Bright}{B_\mathrm{right}}

\begin{document}
\bstctlcite{IEEEexample:BSTcontrol}

\title{On the Decoding Performance of Spatially Coupled LDPC Codes with Sub-block Access}
\author{
\vspace*{0.8cm} Eshed Ram \qquad Yuval Cassuto\\
Andrew and Erna Viterbi Department of Electrical Engineering \\
Technion -- Israel Institute of Technology, Haifa 32000, Israel\\
E-mails: \{s6eshedr@campus, ycassuto@ee\}.technion.ac.il
}

\maketitle

\begin{abstract}

We study spatially coupled LDPC codes that allow access to sub-blocks much smaller than the full code block. Sub-block access is realized by a semi-global decoder that decodes a chosen target sub-block by only accessing the target, plus a prescribed number of helper sub-blocks adjacent in the code chain. This paper develops a theoretical methodology for analyzing the semi-global decoding performance of spatially coupled LDPC codes constructed from protographs. The main result shows that semi-global decoding thresholds can be derived from certain thresholds we define for the single-sub-block graph. These characterizing thresholds are also used for deriving lower bounds on the decoder's performance over channels with variability across sub-blocks, which are motivated by applications in data-storage.

\end{abstract}


\section{Introduction}
\label{Sec:Intro}

Spatially coupled low-density parity-check (SC-LDPC) codes \cite{FelstromZigangirov99} have been shown to be an attractive class of graph codes, thanks to their many desired properties. Most of their good properties stem from the convolution-like structure imposed on their code graphs. Their structure as a terminated chain allows achieving capacity universally with belief-propagation (BP) decoding~\cite{KudekarRichUrb13}. Thanks to the so called threshold-saturation effect, approaching capacity with BP decoding is possible with much simpler graph structures than classical LDPC codes, thus alleviating many of the notorious finite-block issues, and offering superior performance in practice. A very popular and effective construction method for SC-LDPC codes uses chaining of coupled protographs~\cite{Thorpe03}, where extremely simple protographs (e.g. regular) are often sufficient for extremely good performance \cite{MitchLent15,LentSrid10,PusSamr11,MitchDol14}.

Another advantage of the convolutional structure lies in its enabling of efficient low-latency decoders such as the window decoder \cite{IyenPapa12,IyenSiegel13,Lentmaier11}, in which at each step the decoder spans a subset of the code graph and runs BP on it, and then shifts the window to a new position. More recent work further enhances protograph-based SC-LDPC codes and their decoding, through several novel ideas such as optimizing edge spreading  \cite{EsfHar19,MoChen20,LevEsf20}, connecting sub-chains \cite{TruMitch19}, and doping \cite{ZhuMitch20}. In this paper, we harness the convolutional structure of SC-LDPC codes toward a new feature: allowing selective decoding of target sub-blocks within the full code block, without requiring to start the decoding from the beginning of the block. This feature is attractive for deploying SC-LDPC codes in storage applications, which require read access to small units of data at low latency.  

In a recent series of papers \cite{SCLDPCL_arXiv,RamCass18,RamCass19,EsfRam20}, a new type of SC-LDPC codes for efficient sub-block access is presented and studied. 
These codes, called SC-LDPCL codes (suffix 'L' stands for locality), can be decoded locally at the level of sub-blocks that are much smaller than the full code block, thus offering fast access to the coded information alongside the strong reliability of the global full-block decoding. Earlier work on codes with sub-block access includes multi-sub-block Reed-Solomon codes in \cite{CassHemo17}, and multi-sub-block LDPC codes (without spatial coupling) in \cite{RamCass18a}. Toward the analysis and design of SC-LDPC with sub-block access, we present in this paper a detailed characterization of a BP-based decoding mode called {\em semi-global (SG)} decoding, which allows decoding a target sub-block by accessing only the target and a few other sub-blocks around it in the code chain. As a result, SG decoding offers improved correction capabilities for the target sub-block, with access and complexity costs that are much lower compared to full global decoding of the entire block. The SG decoding mode was defined in \cite{RamCass18}, and some tools to evaluate its performance on SC-LDPCL codes are given in \cite{SCLDPCL_arXiv}. These initial results motivate the development of a theoretical methodology toward constructing SC-LDPCL codes with maximal SG-decoding performance, which we pursue in this paper. The key component in the proposed theoretical methodology is the characterization of certain thresholds for the {\em single sub-block code} that govern the decoding performance in the SG mode involving $d+1$ sub-blocks (the target and $d$ other sub-blocks called helpers). As a result, code analysis and construction are reduced to calculation and maximization of these single-sub-block thresholds, in comparison to the unwieldy analysis and search over large graphs encompassing all the sub-blocks participating in SG decoding. The main idea driving the new methodology is to break down the {\em density-evolution} analysis of the full SG decoder to density evolution on single sub-blocks, while formalizing and accounting for the information transfer between subsequent sub-blocks in the decoding process. For simplicity and clarity, we present the results assuming density evolution over the binary erasure channel (BEC), but the results can be extended to other channels using known extensions of the density-evolution method \cite{RichUrb,TenBrink04,LivaChiani07}.

The paper is organized as follows. 
In Section~\ref{Sec:Pre}, we give the necessary background needed for the results presented in the paper. This includes a general construction of SC-LDPCL codes based on binary-regular protographs, a definition of the SG decoding mode, and some preliminaries on Markov chains needed for the codes' analysis on channels with memory in Section~\ref{Sec:SBMV}.
In Section~\ref{Sec:SingleSB}, we study SG decoding from the perspective of a single sub-block (SB) decoded in the process. We define various {\em SB thresholds}, discuss their operational meanings, and show how to calculate them. The derived results in this section are used in later sections to analyze the full SG process. 
In Section~\ref{Sec:Thrsholds}, we define and characterize thresholds for SG decoding ({\em SG thresholds}). In particular, we focus on memoryless channels and consider a limiting case, where the number of accessed SBs is large. Our results connect between SG thresholds and SB thresholds. 
In Section~\ref{Sec:SBMV}, we study the performance of SG decoding over a practically motivated data-storage model in which variability is introduced to the channel quality (as motivated by recent empirical studies \cite{TarUchiSieg16,ShaAl20}). 
While \cite{SCLDPCL_arXiv} considered a SB varying i.i.d. model, in this paper we allow spatial memory in the SBs' channel parameters, for which we consider a model based on Markov chains. 

The useful outcome from the results of this paper is that by examining only the local structure of the single-SB graph, we can tell a lot about the code's performance under different decoding modes. The values of the SB thresholds frame which channel parameters can be handled with local decoding (target only), which require SG decoding and with how many helpers, which require SG decoding to start from a termination sub-block, and which may only be handled by the classical global full-block decoder. To demonstrate this ability, in Section~\ref{Sub:code} we examine all possible unit-memory SC-LDPCL codes (up to symmetries) constructed from two regular protographs, and find the most attractive choice in different scenarios of channel parameters, channel memory, and the number of helpers accessed by the SG decoder.

\section{Background}
\label{Sec:Pre}

\subsection{Notations}\label{Sub:Notations}
The set of natural numbers is denoted by $ \mathbb{N} $. We use calligraphic letters (e.g., $ \mathcal{V,C,E,S} $) or curly brackets to mark discrete sets. 

Scalars are commonly denoted by lowercase Latin or Greek letters, for example, $ i,j,k,l $ or $ \epsilon,\delta,\psi,\phi$. Random variables are denoted by uppercase letters (e.g, $ X,Y,Z $).
For vector notations we add an underline (e.g., $ \underline d,\udel $). For a $ k $-dimensional vector $ \underline d=(d_1,\ldots,d_k)$, we mark by $ \|\underline d \|_1 $ its $ l_1 $-norm, i.e., $ \|\underline d \|_1 = \sum_{i=1}^k |d_i|$.
For two $ k $-dimensional vectors $ \underline \delta=(\delta_1,\ldots,\delta_k)$ and $ \underline \nu=(\nu_1,\ldots,\nu_k) $ we write $ \underline \delta \preceq \underline \nu $ if for every $ i\in\{1,\ldots,k\} $, we have $ \delta_i \leq \nu_i $. If in addition,  $ \delta_i < \nu_i $ for some $ i\in\{1,\ldots,k\} $, then we write $ \underline \delta \prec \underline \nu $.

Matrices are denoted by uppercase letters (e.g., $ B,P,Q $). The element in the $ i $-th row and $ j $-th column of a matrix $ A $ is marked by square brackets: $ [A]_{i,j} $.
If all of the elements of a $ k\times l $ matrix $ A $ equal to some scalar $ a$, then we write $ A= a^{k\times l}$.
For two matrices $ A_1,A_2 $, $ (A_1;A_2 )$ and $ (A_1\;A_2) $ stand for vertical and horizontal concatenation, respectively, of the matrices $ A_1 $ and $ A_2 $.
The Kronecker product of matrices is denoted by $ \otimes $.

\subsection{SC-LDPC Codes with Sub-Block Locality}\label{Sub:SCLDPCL}
An LDPC protograph is a (small) bipartite graph $G=\left(\mathcal{V}\cup\mathcal{C},\mathcal{E}\right)$, where $\mathcal{V},\mathcal{C},$ and $\mathcal{E}$ are the sets of variable nodes (VNs), check nodes (CNs), and edges, respectively.
A protograph $\mathcal G=\left(\mathcal{V}\cup\mathcal{C},\mathcal{E}\right)$ is frequently represented through a bi-adjacency matrix $B$ (called protomatrix), where the VNs in $\mathcal{V}$ are indexed by the columns of $B$, the CNs in $\mathcal{C}$ by the rows, and an element in $B$ represents the number of edges connecting the corresponding VN and CN.
A Tanner graph is generated from a protograph by a lifting operation specified by some lifting parameter (see \cite{MitchLent15}). The design rate of the derived LDPC code is independent of the lifting parameter and given by $1-\big| \mathcal{C} \big|/\big| \mathcal{V} \big| $.
In the limit of large lifting parameters, we can analyze the performance of the BP decoder on the resulting ensemble of Tanner graphs via density evolution (DE) on the protograph. 
The BP decoding threshold of an LDPC protograph is defined as the largest (worst) channel parameter such that DE on the protograph converges to error-free distributions for the protograph's VNs. In this paper, we 
pursue such asymptotic protograph analysis for the binary erasure channel BEC($\epsilon$), but the same analysis can be extended to other channels, e.g., through the EXIT method \cite{TenBrink04}.
Specifically, we write $ \epsilon^*(B) $ for the asymptotic threshold of the protograph $ B $.

An $ (l,r) $-regular SC-LDPC protograph\footnote{the term regular refers to the protomatrix $ B $, while the resulting coupled graph is not regular due to termination.} is constructed as follows. Let $ B =1^{l\times r}$ be an all-ones base matrix representing an $(l,r)$-regular LDPC protograph, let $ T\geq 1 $ be an integer \emph{memory} parameter, and let $ \{B_\tau\}_{\tau=0}^T $ be binary matrices such that $ B=\sum_{\tau=0}^T B_\tau$ (in this paper we consider only binary $ B $ matrices). Coupling $ M >1$ copies of $ B $ amounts to diagonally placing $ M $ copies of $\begin{pmatrix} B_0 ; B_1 ; \cdots ; B_T \end{pmatrix}$ in the coupled matrix. As an example, for $T=1  $ the coupled protomatrix is given by
\begin{align}\label{Eq:SCProto}
&\left ({\arraycolsep=5pt\begin{array}{ccccccc}
B_0	&		&		&	&		&		&\\
B_1	&		&		&	&		&		&\\
&\ddots	&		&	&		&		&\\
&		&B_0	&	&		&		&\\
&		&B_1	&B_0&		&		&\\
&		&		&B_1&B_0	&		&\\
&		&		&	&B_1	&		&\\
&		&		&	&		&\ddots	&\\
&		&		&	&		&		&B_0\\
&		&		&	&		&		&B_1
\end{array}}\right ).\\[-6mm]\notag
&\;\;\;\;\;{\small\arraycolsep=2pt\def\arraystretch{2.2}\begin{array}{ccccccc}
\underbrace{\phantom{1}}_1&\underbrace{\phantom{1}}_{\cdots}&\underbrace{\phantom{1}}_{m\!-\!1}&\underbrace{\phantom{1}}_{m}&\underbrace{\phantom{1}}_{m\!+\!1}&\underbrace{\phantom{1}}_{\cdots}&\underbrace{\phantom{1}}_{M}
\end{array}}
\end{align}
Throughout this paper, we consider $ (l,r)$-regular SC-LDPC protographs with memory $ T=1$, i.e., $ B=1^{l\times r} $ and $ B_1 =1^{l\times r}-B_0  $. We call such codes \emph{unit-memory binary-regular SC-LDPC codes}. The results can be extended to higher-memory codes with some technical modifications.

To endow SC-LDPC codes with more flexible access, we divide the codeword to $ M $ sub-blocks (SBs), where each SB corresponds to one copy of $(B_0;\cdots;B_T)$  in the coupled matrix; the SB size is $ r $ times the lifting factor. We define an $ (l,r,t) $-regular SC-LDPC code with SB locality (in short SC-LDPCL) to be an $ (l,r) $-regular SC-LDPC protograph with a partitioning that is constrained such that $B_0  $ has $ l-t\geq 2 $ all-one rows and $ t $ mixed rows (i.e., with ones and zeros). The all-one and mixed rows correspond to \emph{local checks} (LC) and \emph{coupling checks} (CC), respectively (LCs are connected to VNs only inside SBs, and CCs connect between SBs). It was proved in \cite{SCLDPCL_arXiv} that this is a general description of an $ (l,r,t) $-regular SC-LDPCL protograph that allows local SB decoding.
The resulting protograph can be visualized as a chain of $ M>1 $ coupled SBs, where each SB is an $ (l-t,r) $-regular local code, and adjacent SBs are connected via $ t $ coupling checks with connections specified by the $ t $ mixed rows in $ B_0 $. 

Let $ B_\mathrm{loc} $ be the $ (l-t)\times r $ all-ones matrix that forms the \emph{local} part of $ B_0 $, and let $ B_\mathrm{left},B_\mathrm{right} $ be the $ t\times r $ matrices that form the \emph{coupling} part of $ B_0 ,B_1 $, respectively, i.e., $ B_0=\left (B_\mathrm{left}\;;\;B_\mathrm{loc}\right),\,B_1=\left (B_\mathrm{right}\;;\;0^{(l-t)\times r}\right )$. Then, \eqref{Eq:SCProto} can be re-written as 
\begin{align}\label{Eq:SCProto2}
\begin{pmatrix*}[l]
B_\mathrm{left}	&		&		&	&		&		&\\
B_\mathrm{loc}	&		&		&	&		&		&\\
B_\mathrm{right}	&		&		&	&		&		&\\
&\ddots	&		&	&		&		&\\
&		&B_\mathrm{left}	&	&		&		&\\
&		&B_\mathrm{loc}	&	&		&		&\\
&		&B_\mathrm{right}	&B_\mathrm{left}&		&		&\\
&		&		&B_\mathrm{loc}&		&		&\\
&		&		&B_\mathrm{right}&B_\mathrm{left}	&		&\\
&		&		&	&B_\mathrm{loc}	&		&\\
&		&		&	&B_\mathrm{right} 	&		&\\
&		&		&	&		&\ddots	&\\
&		&		&	&		&		&B_\mathrm{left}\\
&		&		&	&		&		&B_\mathrm{loc}\\
&		&		&	&		&		&B_\mathrm{right}
\end{pmatrix*}.
\end{align}
In view of \eqref{Eq:SCProto2}, $ B_\mathrm{left} $ and $ B_\mathrm{right} $ connect a SB to its neighbors on the left and right, respectively. 

In \cite{RamCass18}, a special case of SC-LDPCL protographs are introduced. In particular, a unit-memory $ (l,r,t) $-regular SC-LDPCL protograph with $ t\in\{1,2,\ldots,l-2\} $ is constructed by setting  $ B_\mathrm{left} $ to have a uniform staircase structure (uniform cutting-vector partition \cite{MitchDol14}). That is, if we set $w=\left \lfloor r/(t+1) \right \rfloor  $, then
\begin{align}\label{Eq:lrtConst}
\left [B_\mathrm{left}\right ]_{i,j}=
\left \{
\begin{array}{ll}
1& 1\leq j \leq iw\\
0& \text{otherwise}
\end{array}
\right .\qquad
\left [B_\mathrm{right}\right ]_{i,j}=
\left \{
\begin{array}{ll}
1& iw+1\leq j \leq r\\
0& \text{otherwise}
\end{array}
\right .\;.
\end{align}

\begin{example}\label{Ex:361_SCLDPCL}
	For $ l=3,r=6,t=1 $, the construction in \cite{RamCass18} yields
	\begin{align}\label{Eq:361}
	\left (
	\begin{array}{c}
	\phantom{
		\begin{array}{c}
		11
		\end{array}
	}
	B_\mathrm{left} 
	\phantom{
		\begin{array}{c}
		11
		\end{array}
	}
	\\
	\hline
	\phantom{
		\begin{array}{c}
		1\\1
		\end{array}
	}
	B_\mathrm{loc} 
	\phantom{
		\begin{array}{c}
		1\\1
		\end{array}
	} \\
	\hline
	B_\mathrm{right}
	\end{array}\right )=
	\left ({\arraycolsep=3pt\begin{array}{cccccc}
	1&1&1&0&0&0\\
	\hline
	1&1&1&1&1&1\\
	1&1&1&1&1&1\\
	\hline
	0&0&0&1&1&1
	\end{array}}\right ).
	\end{align}
	Figure~\ref{Fig:361} illustrates a single SB in the $ (l=3,r=6,t=1) $ SC-LDPCL protograph from \cite{RamCass18}. This example will serve as a running example throughout the paper.
\end{example}

\begin{example}\label{Ex:l4r16Const}
	For $ l=4,r=16,t=2 $, the following partition induces a different SC-LDPCL protograph (not covered by the construction of \cite{RamCass18})
	\begin{align*}
	\left (
	\begin{array}{c}
	\phantom{
		\begin{array}{cc}
		11111&11111\\
		1&1
		\end{array}
	}
	B_\mathrm{left} 
	\phantom{
		\begin{array}{cc}
		1111&11111
		\end{array}
	}
	\\
	\hline
	\phantom{
		\begin{array}{c}
		1\\1
		\end{array}
	}
	B_\mathrm{loc} 
	\phantom{
		\begin{array}{c}
		1\\1
		\end{array}
	} \\
	\hline
	\phantom{
		\begin{array}{cc}
		1\\	1
		\end{array}
	}
	B_\mathrm{right}
	\phantom{
		\begin{array}{cc}
		1\\	1
		\end{array}
	}
	\end{array}\right )=
	\left ({\arraycolsep=3pt\begin{array}{cccccccccccccccc}
	1&1&1&1&1&1&1&1&0&0&0&0&0&0&1&1\\
	0&0&0&1&1&1&1&1&1&1&1&1&0&1&0&1\\
	\hline
	1&1&1&1&1&1&1&1&1&1&1&1&1&1&1&1\\
	1&1&1&1&1&1&1&1&1&1&1&1&1&1&1&1\\
	\hline
	0&0&0&0&0&0&0&0&1&1&1&1&1&1&0&0\\
	1&1&1&0&0&0&0&0&0&0&0&0&1&0&1&0
	\end{array}}\right ).
	\end{align*}
\end{example}

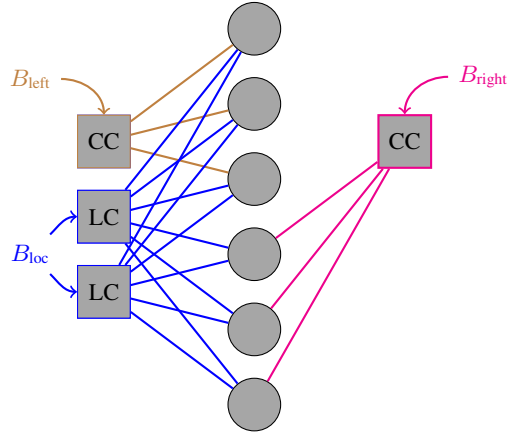
\begin{figure}
	\begin{center}
	\begin{tikzpicture}\label{Tikz:36 SCLDPCL}
	\tikzstyle{cnode}=[rectangle,draw,fill=gray!70!white,minimum size=7mm]
	\tikzstyle{vnode}=[circle,draw,fill=gray!70!white,minimum size=7mm]
	\tikzstyle{edge}=[thick,blue]
	\pgfmathsetmacro{\x}{4}
	\pgfmathsetmacro{\y}{1}
	
	\foreach \m in {2}
	{
		
		\foreach \c in {1,2,3}
		{
			\node[cnode,draw=blue] (c\c\m) at (\m*\x-\x,\c*\y-\y) {};	
		}
		\foreach \v in {1,2,3,4,5,6}
		{
			\node[vnode] (v\v\m) at (\m*\x-0.5*\x,\v*\y-2.5*\y) {};	
		}	
	}
	\node[cnode,draw=blue] (c22) at (c22) {\footnotesize LC};
	\node[cnode,draw=blue] (c12) at (c12) {\footnotesize LC};
	\node[cnode,draw=brown] (c32) at (c32) {\textcolor{black}{\footnotesize CC}};
	\node[cnode,thick,draw=magenta] (c33) at (3*\x-\x,3*\y-\y) {\textcolor{black}{\footnotesize CC}};	
	
	\node (Bl) [above left=3mm of c32,align=center,brown] {\footnotesize $B_{\text{left}}$}; 
	\node (Br) [above right=3mm of c33,align=center,magenta] {\footnotesize $B_{\text{right}}$};
	\node (x) at ($(c22)!0.5!(c12)$){};
	\node (Bloc) at (Bl|-x)[align=center,blue] {\footnotesize $B_{\text{loc}}$}; 
	\draw[->,thick,brown] (Bl)to[in= 90, out = 0] (c32.north);
	\draw[->,thick,magenta] (Br)to[in= 90, out = 180] (c33.north);
	\draw[->,edge] (Bloc)to[in= 180, out = 45] (c22);
	\draw[->,edge] (Bloc)to[in= 180, out = 315] (c12);
	\foreach \m in {2}
	{	
		\pgfmathtruncatemacro{\k}{\m + 1}
		\draw[edge] (v1\m)--(c1\m) ;
		\draw[edge] (v1\m)--(c2\m) ;
		\draw[thick,magenta] (v1\m)--(c3\k) ;	
		\draw[edge] (v2\m)--(c1\m) ;
		\draw[edge] (v2\m)--(c2\m) ;
		\draw[thick,magenta] (v2\m)--(c3\k) ;
		
		\draw[edge] (v3\m)--(c1\m) ;
		\draw[edge] (v3\m)--(c2\m) ;
		\draw[thick,magenta] (v3\m)--(c3\k) ;
		\draw[edge] (v4\m)--(c1\m) ;
		\draw[edge] (v4\m)--(c2\m) ;
		\draw[thick,brown] (v4\m)--(c3\m) ;
		
		\draw[edge] (v5\m)--(c1\m) ;
		\draw[edge] (v5\m)--(c2\m) ;
		\draw[thick,brown] (v5\m)--(c3\m) ;
		\draw[edge] (v6\m)--(c1\m) ;
		\draw[edge] (v6\m)--(c2\m) ;
		\draw[thick,brown] (v6\m)--(c3\m) ;			
	}	
	\end{tikzpicture}
	\caption{\label{Fig:361}Illustration of the $ (l=3,r=6,t=1) $ SC-LDPCL SB protograph in \eqref{Eq:361}. The connections corresponding to $B_{\text{left}}  $,$B_{\text{loc}}  $, and $B_{\text{right}}  $, are drawn in brown, blue, and magenta, respectively.}		
	\end{center}
	
\end{figure}

\subsection{Semi-Global Decoding}

SC-LDPCL codes can be decoded locally for fast read access (as done in \cite{RamCass18a} for non-SC LDPC codes) and globally for increased data reliability. 
Semi-global (SG) decoding \cite{RamCass18,RamCass19} is another decoding mode for SC-LDPCL codes that offers a middle way between local and global decoding. 

\subsubsection{Decoding-Mode Description}

In SG decoding, the user is interested in SB $ m\in\{1,\ldots,M\}$, which is called the \emph{target} SB, and the decoder decodes it with the help of additional $d$ neighbor SBs called \emph{helper} SBs. $d$ is a parameter that bounds the number of additional SBs read for decoding one SB; hence, the smaller $d$ is, the faster access the code offers for single SBs.
In semi-global decoding there are two phases: the \emph{helper phase}, and the \emph{target phase}. In the former, helper SBs are decoded locally, incorporating information from other previously decoded helper SBs. In the latter, the target SB is decoded while incorporating information from its neighboring helper SBs. 
Figure~\ref{Fig:SG decoding} exemplifies semi-global decoding with $ d=4 $ helper SBs. In this example, the helper phase consists of decoding helper SBs $ m - 2 $ and $ m + 2 $ locally, and decoding helper SBs $ m - 1 $ and $ m + 1 $ using the information from helper SBs $ m - 2 $ and $ m + 2 $, respectively. In the target phase, SB $ m $ is decoded using information from both SB $ m - 1$ and $ m + 1$. 
\begin{figure}
	\begin{center}
		\begin{tikzpicture}[>=latex]\label{Tikz:SG decoding}
		
			\tikzstyle{SB}=[rectangle,very thick,draw, rounded corners, minimum width=1.1cm,minimum height=0.5cm,fill=white]
			\pgfmathsetmacro{\x}{1.1}
			\pgfmathsetmacro{\y}{1.2}
			\node (SB0) [SB] at (0*\x,0) {\footnotesize $m$};
			\node (SB1) [SB] at (1*\x,0) {\footnotesize $m + 1$};
			\node (SB2) [SB,fill=gray!30!white] at (2*\x,0) {\footnotesize $m + 2$};
			\node (SB-1) [SB] at (-1*\x,0) {\footnotesize $m - 1$};
			\node (SB-2) [SB,fill=gray!30!white] at (-2*\x,0) {\footnotesize $m - 2$};
			\node (rdots1) [right=0.3mm of SB2] {\footnotesize $ \dots $};
			\node (ldots) [left =0.3mm of SB-2] {\footnotesize step 1 $ \dots $};

			\node (SB0) [SB] at (0*\x,-\y) {\footnotesize$m$};
			\node (SB1) [SB,fill=gray!30!white] at (1*\x,-\y) {\footnotesize$m + 1$};
			\draw [->] (SB2)--(SB1) ;
			\node (SB2) [SB] at (2*\x,-\y) {\footnotesize$m + 2$};
			\node (SB-1) [SB,fill=gray!30!white] at (-1*\x,-\y) {\footnotesize$m - 1$};
			\draw [->] (SB-2)--(SB-1) ;
			\node (SB-2) [SB] at (-2*\x,-\y) {\footnotesize$m - 2$};
			\node (rdots2) [right=0.3mm of SB2] {\footnotesize$ \dots $};
			\node (ldots) [left =0.3mm of SB-2] {\footnotesize step 2 $ \dots $};
			\draw [thick,decorate,decoration={brace,amplitude=5pt,raise=6pt}] (rdots1.north)--(rdots2.south) node [black,midway,xshift=9mm,text width=1cm,align=center] {\footnotesize helper phase};
			
			\node (SB0) [SB,fill=gray!30!white] at (0*\x,-2*\y) {\footnotesize$m$};
			\draw [->] (SB1)--(SB0) ;
			\draw [->] (SB-1)--(SB0) ;
			\node (SB1) [SB] at (1*\x,-2*\y) {\footnotesize$m + 1$};
			
			\node (SB2) [SB] at (2*\x,-2*\y) {\footnotesize$m + 2$};
			\node (SB-1) [SB] at (-1*\x,-2*\y) {\footnotesize$m - 1$};
			
			\node (SB-2) [SB] at (-2*\x,-2*\y) {\footnotesize$m - 2$};
			\node (rdots3) [right=0.3mm of SB2] {\footnotesize$ \dots $};
			\node (ldots) [left =0.3mm of SB-2] {\footnotesize step 3 $ \dots $};
			\draw [thick,decorate,decoration={brace,amplitude=3pt,raise=33pt}] (SB2.north)--(SB2.south) node [black,midway,xshift=17mm,text width=1cm,align=center] {\footnotesize target phase};
			
		\end{tikzpicture}
	\end{center}
	\caption{\label{Fig:SG decoding}Example of semi-global decoding with target SB $ m\in[1:M] $, and $ d=4 $; the steps are shown from top to bottom. The gray SBs are those that are decoded in a given step, and the arrows represent information passed between SBs.}
\end{figure}
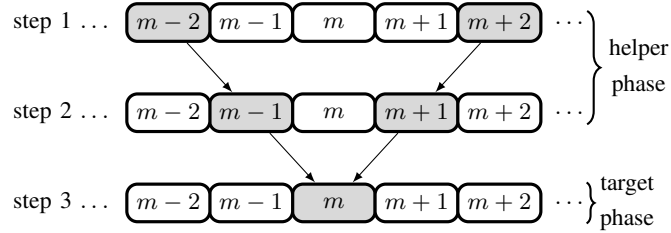
Note that semi-global decoding resembles window decoding of SC-LDPC codes (see \cite{IyenPapa12,IyenSiegel13,Lentmaier11}) but differs in: 1) for a given target, there is no overlap between two window positions, which decreases latency and complexity, and 2) decoding can start close to the target SB (i.e., not necessarily at the first or last SBs), allowing low-latency access to SBs anywhere in the block. 

Semi-global decoding is highly motivated by the locality property of SBs in SC-LDPCL codes (SBs can be decoded locally), the spatial coupling of SBs (SBs can help their neighbor SBs), and by practical channels in storage devices, i.e., channels with variability \cite{TarUchiSieg16}. 

\subsubsection{Semi-Global Density-Evolution}\label{Sub:SG_DE}
When analyzing semi-global decoding in the asymptotic regime (i.e., performing DE), we can exploit the sequential nature of the algorithm to reduce the analysis over $d$ SBs to analysis of individual SBs with information carried between adjacent SBs. This information is represented in the analysis through fixed DE values input to the SB's DE calculation. 
We develop such a method in this paper, where for simplicity the DE is performed on the BEC such that DE values are simply erasure probabilities \cite{RichUrb}.

For a left (resp. right) helper SBs, the DE values from the previously decoded helper are directed to the $ t $ CCs corresponding to the $ t $ rows of $ \Bleft $ (resp. $ \Bright $) and remain fixed during the SB's DE iterations. We call these CCs, \emph{incoming} CCs and we denote by $ \underline \delta_I \in[0,1]^t$ the incoming DE values; $ \udi $ affects the DE equations for incoming CCs. 
Note that for termination helper SBs, i.e., the endpoint SBs in the coupled chain, we have $  \underline \delta_I=\underline 0  $.
When the SB's DE iterations stop (i.e., reach a DE fixed-point or any other stopping criteria), the \emph{outgoing} DE values, which are denoted by $ \udo\in[0,1]^t $ are calculated and sent to the next SB via the $ t $ \emph{outgoing} CCs corresponding to the $ t $ rows of $ \Bright $, for left helper SBs, or $ \Bright $, for right helper SBs. These outgoing DE values turn into incoming DE values for the next helper SB, which runs a similar process and outputs DE values to the next SB, etc. We call the DE values passing between SBs \emph{inter-SB DE values}.

Similarly, the target SB receives incoming DE values from both of its neighbor helper SBs, and these DE values are kept fixed during the target's DE iterations. We denote the incoming DE values to the target from the left (resp. right) neighbor by $ \underline \delta_L  \in[0,1]^t$ (resp. $ \underline \delta_R\in[0,1]^t $); $ \udel_L $ (resp. $ \udel_R $) is directed to the $ t $ left (resp. right) CCs corresponding to the $ t $ rows of $ \Bleft $ (resp. $ \Bright $). Same as for the helper SB, DE iterations at the target are executed until some stopping criteria is met. If there is some non-trivial DE fixed point (i.e., not at the origin), then decoding will fail; we call such a fixed-point \emph{SG-DE fixed-point}.

\begin{example}\label{Ex:SGGraph}
Figure~\ref{Fig:SGgraph} illustrates the above notations for the SB in Example~\ref{Ex:361_SCLDPCL}. Note that since $ t=1 $, then $ \delta_I,\delta_O,\delta_R,\delta_L $ are scalars.
\end{example}   

\begin{figure}
	\begin{center}
		\begin{tikzpicture}\label{Tikz:SG}
		\tikzstyle{cnode}=[rectangle,draw,fill=gray!70!white,minimum size=7mm]
		\tikzstyle{vnode}=[circle,draw,fill=gray!70!white,minimum size=7mm]
		\tikzstyle{edge}=[thick]
		\pgfmathsetmacro{\x}{4}
		\pgfmathsetmacro{\y}{1}
		
		\foreach \m in {2}
		{
			
			\foreach \c in {1,2,3}
			{
				\node[cnode] (c\c\m) at (\m*\x-\x,\c*\y-\y) {};	
			}
			\foreach \v in {1,2,3,4,5,6}
			{
				\node[vnode] (v\v\m) at (\m*\x-0.5*\x,\v*\y-2.5*\y) {};	
			}	
		}
		\node[cnode] (c22) at (c22) {\footnotesize LC};
		\node[cnode] (c12) at (c12) {\footnotesize LC};
		\node[cnode] (c32) at (c32) {\footnotesize CC};
		\node[cnode] (c33) at (3*\x-\x,3*\y-\y) {\footnotesize CC};	
		\node(dl) [left of=c32] {$ \delta_L $};
		\node(dr) [right of =c33] {$ \delta_R $};
		\draw [edge,->,>=latex] (dl)--(c32);\draw [edge,->,>=latex] (dr)--(c33);
		\foreach \m in {2}
		{	
			\pgfmathtruncatemacro{\k}{\m + 1}
			\draw[edge] (v1\m)--(c1\m) ;
			\draw[edge] (v1\m)--(c2\m) ;
			\draw[edge] (v1\m)--(c3\k) ;	
			\draw[edge] (v2\m)--(c1\m) ;
			\draw[edge] (v2\m)--(c2\m) ;
			\draw[edge] (v2\m)--(c3\k) ;
			
			\draw[edge] (v3\m)--(c1\m) ;
			\draw[edge] (v3\m)--(c2\m) ;
			\draw[edge] (v3\m)--(c3\k) ;
			\draw[edge] (v4\m)--(c1\m) ;
			\draw[edge] (v4\m)--(c2\m) ;
			\draw[edge] (v4\m)--(c3\m) ;
			
			\draw[edge] (v5\m)--(c1\m) ;
			\draw[edge] (v5\m)--(c2\m) ;
			\draw[edge] (v5\m)--(c3\m) ;
			\draw[edge] (v6\m)--(c1\m) ;
			\draw[edge] (v6\m)--(c2\m) ;
			\draw[edge] (v6\m)--(c3\m) ;			
		}
		
		\node (A) [below = 5mm of v12] {(a): Target};
		
		\begin{scope}[xshift=2*\x cm]
			\foreach \m in {2}
		{
			
			\foreach \c in {1,2,3}
			{
				\node[cnode] (c\c\m) at (\m*\x-\x,\c*\y-\y) {};	
			}
			\foreach \v in {1,2,3,4,5,6}
			{
				\node[vnode] (v\v\m) at (\m*\x-0.5*\x,\v*\y-2.5*\y) {};	
			}	
		}
		\node[cnode] (c22) at (c22) {\footnotesize LC};
		\node[cnode] (c12) at (c12) {\footnotesize LC};
		\node[cnode] (c32) at (c32) {\footnotesize CC};
		\node[cnode] (c33) at (3*\x-\x,3*\y-\y) {\footnotesize CC};	
		\node (in) [above =6mm of c32,text width=35pt,align=center] {\footnotesize incoming CC};
		\node (ix) at ($(in)!0.65!(c32)$) {};
		\draw [->] (in)--(ix);
		\node (out) [above =6mm of c33,text width=35pt,align=center] {\footnotesize outgoing CC};
		\node (ox) at ($(out)!0.65!(c33)$) {};
		\draw [->] (out)--(ox);
		\node(dl) [left of=c32] {$ \delta_I $};
		\node(dr) [right of =c33] {$ \delta_O $};
		\draw [edge,->,>=latex] (dl)--(c32);\draw [edge,<-,>=latex] (dr)--(c33);
		
		\foreach \m in {2}
		{	
			\pgfmathtruncatemacro{\k}{\m + 1}
			\draw[edge] (v1\m)--(c1\m) ;
			\draw[edge] (v1\m)--(c2\m) ;
			\draw[edge,dashed] (v1\m)--(c3\k) ;	
			\draw[edge] (v2\m)--(c1\m) ;
			\draw[edge] (v2\m)--(c2\m) ;
			\draw[edge,dashed] (v2\m)--(c3\k) ;
			
			\draw[edge] (v3\m)--(c1\m) ;
			\draw[edge] (v3\m)--(c2\m) ;
			\draw[edge,dashed] (v3\m)--(c3\k) ;
			\draw[edge] (v4\m)--(c1\m) ;
			\draw[edge] (v4\m)--(c2\m) ;
			\draw[edge] (v4\m)--(c3\m) ;
			
			\draw[edge] (v5\m)--(c1\m) ;
			\draw[edge] (v5\m)--(c2\m) ;
			\draw[edge] (v5\m)--(c3\m) ;
			\draw[edge] (v6\m)--(c1\m) ;
			\draw[edge] (v6\m)--(c2\m) ;
			\draw[edge] (v6\m)--(c3\m) ;			
		}
		\end{scope}
		\node (B) at(A-|v62) {(b): Helper};

		\end{tikzpicture}
		\caption{\label{Fig:SGgraph}The graphs corresponding to the target (a) and helper (b) SBs during semi-global decoding for the SB in Example~\ref{Ex:361_SCLDPCL}. Dashed edges do not participate during the SB decoding, except in sending DE values $\delta_O$ at the end.}		
	\end{center}
	
\end{figure}
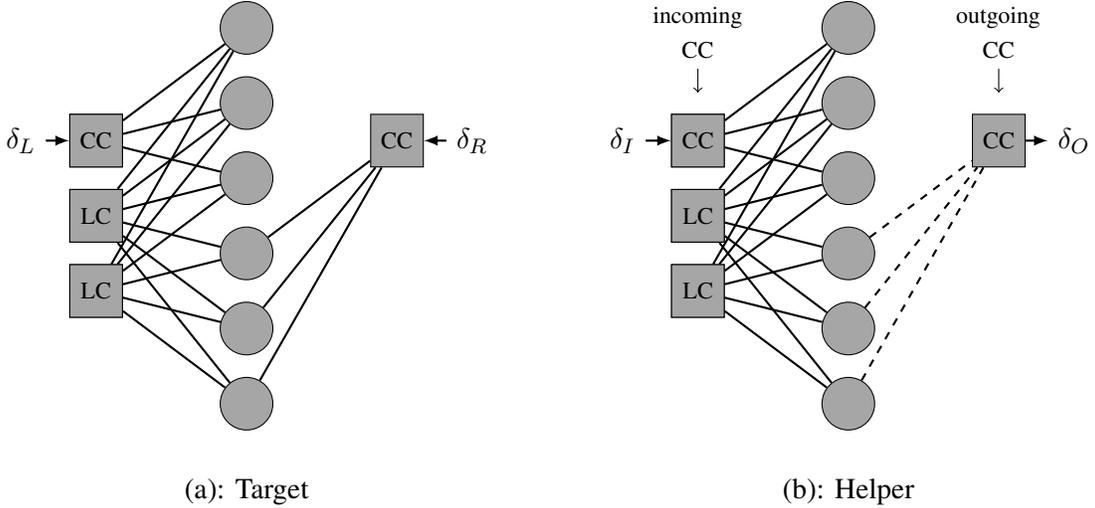
In \cite{RamCass19,SCLDPCL_arXiv}, DE equations for SG decoding of cutting-vector-based SC-LDPCL codes are derived. In this paper, we extend the results to general constructions, and give a detailed characterization of the decoder's performance (i.e., thresholds). 

\subsection{Markov Chains}

In Section~\ref{Sec:SBMV}, we study the performance of semi-global decoding over a channel model which is based on Markov chains. We now give a brief introduction to the basic notations and properties of Markov chains from \cite{Ross}. 
We will consider discrete Markov chains, i.e., the time index and state space are discrete. In our context, the time index is a SB index, and the state space represents the channel quality.
Let $ \mathcal S=\{s_1,s_2,\ldots,s_{|\mathcal S|}\} $ be the state space. The random process $ \{X_k\}_{k=1}^\infty $ is said to be a Markov chain if 
\begin{align*}
\Pr\left (X_k=s_j\big| X_{k-1}=s_{i},X_{k-2}=s_{i_{k-2}},\ldots ,X_{1}=s_{i_1}\right )= \Pr\left (X_k=s_j\big| X_{k-1}=s_{i}\right ),\quad \forall 2\leq k\;, \forall 1\leq i,j\leq |\mathcal{S}|.
\end{align*}
For every $ i,j\in\{1,2,\ldots,|\mathcal S|\} $, we write $ \Pr\left (X_k=s_j\big| X_{k-1}=s_{i}\right )=P_{i,j} $, where $ P $ is a $ |\mathcal{S}|\times |\mathcal S| $ transition-probability matrix, i.e., a matrix with non-negative entries and rows summing-up to $ 1 $.

The Chapman-Kolmogorov equations state that for every $ i,j\in\{1,2,\ldots,|\mathcal S|\} $ and $ k,d\in\mathbb N $,
\begin{align}\label{Eq:CahpKol}
\Pr \left ( X_{k+d}=s_j|X_{k}=s_i\right ) =\left [P^d\right ]_{i,j} ,
\end{align}
where $ P^d $ is the $ d $-th power of $ P $.

A distribution $ \underline \nu=\left (\nu_1,\nu_2,\ldots,\nu_{|\mathcal S|}\right ) $ such that for every $ j\in\{1,2\ldots,
|\mathcal S|\} $, $ \nu_j=\Pr(X_1=s_j) $, is a \emph{stationary distribution} if for every $ k\in \mathbb N $ and every $ j\in\{1,2,\ldots,
|\mathcal S|\} $, $ \nu_j=\Pr(X_k=s_j) $. This happens if and only if $ \underline \nu $ is a left eigenvector of $ P $ with eigenvalue $ 1 $. 

Let $\{ X_k\}_{k=1}^\infty $ be a Markov chain with a transition matrix $ P $, and let $ \underline \nu $ be a stationary distribution for $\{ X_k\}_{k=1}^\infty $. Let $ m $ be a large integer. If we consider the reverse process $ Y_k = X_{m-k} $, then $ Y_k $ is a Markov chain with a transition matrix $ \hat P $ given by
\begin{align}\label{Eq:ReverseMarkov}
\hat P_{i,j} = \frac{\nu_j}{\nu_i}P_{j,i},\quad \forall 1\leq i,j\leq |\mathcal S|,\, 
\end{align}
and with the same stationary distribution $ \underline \nu$.
Finally, let $ Z_k=\left (X_{m+k},X_{m-k}\right ) $ be a two-dimensional process taking values in $ \mathcal{S}\times \mathcal S $. Then, $ Z_k $ is a Markov chain with a transition matrix $ P\otimes\hat P $, where $ \otimes  $ stands for the Kronecker product:
\begin{align*}
\Pr (&Z_k=\left (s_{j_1},s_{j_2}\right ) \big| Z_{k-1}=\left (s_{i_1},s_{i_2}\right ),Z_{k-2}=\left (s_{i_3},s_{i_4}\right ),\ldots)
\\
=&\Pr\left (X_{m+k}=s_{j_1},X_{m-k}=s_{j_2} \big| X_{m+k-1}=s_{i_1},X_{m-k+1}=s_{i_2},X_{m+k-2}=s_{i_3},X_{m-k+2}=s_{i_4}\ldots\right )\\
=&\Pr\left (X_{m+k}=s_{j_1} \big| X_{m+k-1}=s_{i_1},X_{m-k+1}=s_{i_2},X_{m+k-2}=s_{i_3},X_{m-k+2}=s_{i_4}\ldots\right )\\
\cdot&\Pr\left (X_{m-k}=s_{j_2} \big| X_{m+k}=s_{j_1},X_{m+k-1}=s_{i_1},X_{m-k+1}=s_{i_2},X_{m+k-2}=s_{i_3},X_{m-k+2}=s_{i_4}\ldots\right )\\
=&\Pr\left (X_{m+k}=s_{j_1} \big| X_{m+k-1}=s_{i_1}\right )\cdot
\Pr\left (X_{m-k}=s_{j_2} \big| X_{m-k+1}=s_{i_2}\right )\\
=&P_{i_1,j_1}\cdot \hat P_{i_2,j_2}.
\end{align*}

\section{Single-Sub-Block Analysis}\label{Sec:SingleSB}
In this section, we study the semi-global decoding procedure from the perspective of a single SB decoded in the process (with inputs from adjacent SBs) and derive results that will be used in later sections when we characterize and analyze the full process of semi-global decoding. 

\subsection{Degree Profile}

\begin{definition}\label{Def:DegProfile}
	Let $ a,b\in\mathbb{N} $ such that $ a<b $, and let $ A\in\{0,1\}^{a\times b} $ be a protomatrix. The \emph{degree profile} of $ A $ consists of two vectors $ \udc(A)\in \mathbb{N}^{a},\udv(A)\in \mathbb{N}^{b} $ given by
		\begin{align*}
		&\left [\udc(A)\right ]_i = \sum_{j=1}^b A_{i,j},\quad 1\leq i \leq a,\\
		&\left [\udv(A)\right ]_j = \sum_{i=1}^a A_{i,j},\quad 1\leq j \leq b.
		\end{align*}	
	
\end{definition}
Note that not every pair of vectors $ \udc\in \mathbb N^{a},\udv\in \mathbb N^{b} $ is a \emph{realizable} degree profile of some protomatrix (i.e., a protomatrix with  $ \udc\in \mathbb N^{a}$ and $\udv\in \mathbb N^{b} $ as its degree profile exists). For every protomatrix $ A $, we must have $ \|\udc(A)\|_1 = \sum_{i=1}^a\sum_{j=1}^bA_{i,j} = \|\udv(A)\|_1 $. If we allow non-binary protomatrices, then this condition suffices for a realizable degree profile. For binary protomatrices, this condition is \emph{not} sufficient. For an example of a \emph{non-realizable} degree profile, consider $ a=3,b=4 $ and $  \udc = (4,4,1),\udv = (1,2,3,3) $. Since the first two rows are full, then the VN degrees should be all at least $ 2 $. 

Consider a single SB $ \left (B_{\text{left}} ;B_{\text{loc}} ;B_{\text{right}}\right ) $ in the coupled protograph. Since by construction $ B_\mathrm{left}+B_\mathrm{right} = {1}^{t\times r}, B_\mathrm{loc}=  {1}^{(l-t)\times r}$, then 
\begin{subequations}
\begin{align}
\label{Eq:dBU=dBD1}
&\udc(B_\mathrm{left})+\udc(B_\mathrm{right}) = r^{1\times t},\\
\label{Eq:dBU=dBD11}
&\udc(B_\mathrm{loc})= r^{1\times (l-t)},\\
\label{Eq:dBU=dBD2}
&\udv(B_\mathrm{left})+\udv(B_\mathrm{right}) = t^{1\times r}\\
\label{Eq:dBU=dBD3}
&\udv(B_\mathrm{loc})=(l-t)^{1\times r}.
\end{align}
\end{subequations}

\subsection{Erasure Transfer in Binary-Regular Spatially Coupled Protographs}

We now move to the characterization and analysis of a helper SB during semi-global decoding (see Figure~\ref{Fig:SGgraph}(b) for an example). Our derivations assume that the underlying base (uncoupled) protograph is binary and regular (see Section~\ref{Sub:SCLDPCL}). To analyze the performance of SG decoders in the asymptotic regime, we now develop a DE framework that captures the transfer of erasures between adjacent SBs through the vectors $\udi,\udo,\udel_R,\udel_L$. This analysis is an extension of \cite{SCLDPCL_arXiv,RamCass19} which focused on cutting-vector constructions.
Decoding of sub-graphs with incoming and outgoing DE values was considered in \cite{TruMitch19}, for the purpose of inter-connecting sub-chains of SC-LDPC codes. In this work, smaller units (SBs) are inter-connected, and the performance objective is ability to decode a single target SB with an efficient SG decoder.

The propositions in this and the next sub-sections will be used in Sections~\ref{Sec:Thrsholds} and~\ref{Sec:SBMV} to prove results concerning the performance of the semi-global decoder on general $(l,r,t)$ SC-LDPCL codes constructed from a binary base protomatrix $ B $. 
In what follows, we distinguish between helper SBs that have lower SB indices than the target SB (e.g., SBs $ m-2,m-1 $ in Figure~\ref{Fig:SG decoding}) and helper SBs that have higher SB indices (SBs $ m+1,m+2 $ in Figure~\ref{Fig:SG decoding}). We call the former \emph{left helper SBs} and the latter \emph{right helper SBs}.

\begin{definition}\label{Def:Delta}
	Consider a helper SB during semi-global decoding.
	 Let $\epsilon\in[0,1]$ be a SB erasure rate, and let $ \udi\in[0,1]^t $ be the \emph{incoming} DE values from a previously decoded helper.
	The \emph{erasure-transfer function} calculates the DE values outgoing towards the next SB (i.e., $ \udel_O \in[0,1]^t$) given $ \epsilon $ and $ \udel_I $.
	We denote by $ \Delta_L\left (\epsilon,\udi\right ) $ and $ \Delta_R\left (\epsilon,\udi\right )  $ the erasure-transfer functions corresponding to right and left helper SBs, respectively (see Figure~\ref{Fig:Delta}).
\end{definition}
Note that $ \Delta_L\left (\cdot,\cdot\right ) $ and $ \Delta_R\left (\cdot,\cdot\right )  $ depend on the connectivity imposed by $ \left (B_\mathrm{left}\;;\;B_\mathrm{loc}\;;\;B_\mathrm{right}\right)$. In particular, for a left (respectively right) helper SB, the incoming and outgoing CCs are represented by the rows of $B_\mathrm{left}  $ and $B_\mathrm{right}  $ (respectively $B_\mathrm{right}  $ and $B_\mathrm{left}  $), respectively. 

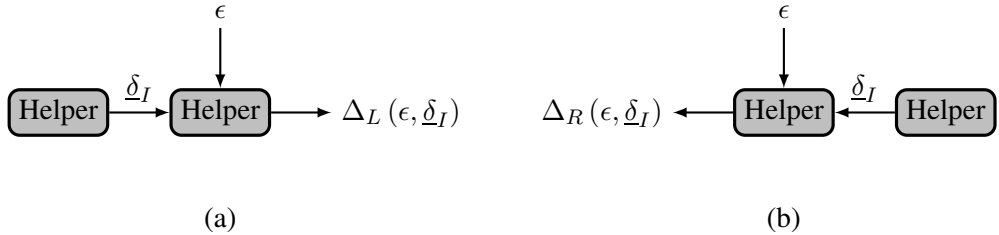
\begin{figure}[!h]
	\begin{center}
		\begin{tikzpicture}
		\tikzstyle{SB}=[fill=gray!50!white,rectangle,very thick,draw, rounded corners, minimum width=1.1cm,minimum height=0.5cm]
		\pgfmathsetmacro{\x}{8}
		
		\node(SBh)[SB]{Helper};
		\node(SB)[SB,right=\x mm of SBh]{Helper};
		\node (eps) [above=\x mm of SB] {$ \epsilon $};
		\node (a) [below=\x mm of SB] {(a)};
		
		\node (dout) [right=\x mm of SB] {$ \Delta_L\left (\epsilon,\udi\right )$};
		\draw [thick,->,>=latex] (SBh)--(SB) node [pos=0.5,above] {$\udi$};
		\draw [thick,->,>=latex] (eps)--(SB);\draw [thick,->,>=latex] (SB)--(dout);
		
		\node (dout) [right=\x mm of dout] {$ \Delta_R\left (\epsilon,\udi\right )$};
		\node(SB)[SB,right=\x mm of dout]{Helper};
		\node(SBh)[SB,right=\x mm of SB]{Helper};
		
		\node (eps) [above=\x mm of SB] {$ \epsilon $};
		\node (b) at (a-|SB) {(b)};

		\draw [thick,->,>=latex] (SBh)--(SB) node [pos=0.5,above] {$\udi$};
		\draw [thick,->,>=latex] (eps)--(SB);\draw [thick,->,>=latex] (SB)--(dout);
		\end{tikzpicture}
	\end{center}
	\caption{\label{Fig:Delta}Illustration of the operational meaning of $ \Delta_L\left (\epsilon,\udi\right ) $ and $ \Delta_R\left (\epsilon,\udi\right ) $: (a) a left helper SB; (b) a right helper SB.}
\end{figure}

Let $ \left (B_\mathrm{left};B_\mathrm{loc};B_\mathrm{right}\right )\in \{0,1\}^{(t+l)\times r} $ be a protograph of a left helper SB in an SC-LDPCL protograph (the right-helper case follows similarly). Let $ \sigma_1,\ldots,\sigma_r $ be the bit-erasure rates at the VNs of the SB when the helper-DE process stops (with some stopping criteria). Then, the outgoing DE values from this SB (to the next helper) $ \udel_O = \left (\delta_{O,1},\ldots,\delta_{O,t}\right ) $ are given by 
\begin{align}\label{Eq:udelO}
\delta_{O,i} = 1-\prod_{j=1}^r \left (1-\sigma_j\right )^{\left [B_\mathrm{right}\right ]_{i,j}},\quad \forall 1\leq i \leq t.
\end{align}
For right helper SBs, the above holds with the substitution of $ B_\mathrm{right} $ with $ B_\mathrm{left} $. Note that the values $ \sigma_1,\ldots,\sigma_r $ depend on  the channel parameter and the incoming DE values to the helper.

\begin{proposition}\label{Prop:del>0Gen}
	Consider a helper SB in a unit-memory binary-regular SC-LDPCL protograph. Then, the bit-erasure rates of all of the VNs in the SB are zero if and only if  $ \udel_O = \underline 0$. 
\end{proposition}

\begin{proof}
	In view of \eqref{Eq:udelO}, if for every $ j\in\{1,\ldots,r\} $, $ \sigma_j=0 $, then for every $ i\in\{1,\ldots,t\} $, $ \delta_{O,i} =0$.
	For the converse, assume to the contrary that for every $ i\in\{1,\ldots,t\} $, $ \delta_{O,i} =0$, and there exists a VN $ j\in\{1,2,\ldots,r\} $ such that $ \sigma_j>0 $. We divide the proof into two cases: 1) VN $ j $ is connected to all of the incoming CCs, 2) VN $ j $ is not connected to all of the incoming CCs.
	\begin{enumerate}
		\item Since VN $ j $ is connected to all of the incoming CCs and all the LCs (every VN is connected to all of the LCs), and $ \sigma_j>0 $, then all of the CNs during the SB's DE process receive positive (and bounded away from zero) values from VN $ j $. Consequently, all of the VNs in the SB receive positive (and bounded away from zero) values from all the CNs that participate in the SB's decoding. We conclude that at the end of the DE process, $\sigma_j>0$ for every $j\in\{1,\ldots,r\}$, contradicting \eqref{Eq:udelO}.
		\item Since VN $ j $ is not connected to all of the incoming CCs, and the SC protograph is based on binary-regular base protographs, then VN $ j $ is connected to some outgoing CC $ i\in\{1,\ldots,t\} $, i.e., $ \left [B_\mathrm{right}\right ]_{i,j}=1 $. In view of \eqref{Eq:udelO}, we get $ \delta_{O,i}>1-\left (1-\sigma_j\right )=\sigma_j>0 $, a contradiction.
	\end{enumerate}
	
\end{proof}

The meaning of Proposition~\ref{Prop:del>0Gen} is that, for unit-memory binary-regular SC-LDPC codes, providing full outgoing information by a helper is equivalent to the helper SB's full decoding. This fact simplifies the analysis and construction pursued later in the paper.

For the following derivations, consider the DE equation for the $ (l,r) $ LDPC ensemble \cite{RichUrb}, i.e., 
\begin{align}\label{Eq:x_ell}
x_\ell(\epsilon) = \epsilon\cdot\left (1-\left (1-x_{\ell-1}(\epsilon)\right )^{r-1}\right) ^{l-1},\quad x_{-1}=1.
\end{align}
Let $ \epslr $ be the asymptotic threshold of the $ (l,r) $ LDPC ensemble, and for every $ \epsilon\in[0,1] $ let $ x(\epsilon) = \lim_{\ell\to\infty}x_{\ell}(\epsilon) $ such that
\begin{align}\label{Eq:xs}
x(\epsilon) = \epsilon\cdot\left (1-\left (1-x(\epsilon)\right )^{r-1}\right) ^{l-1}.
\end{align}
The value $ x(\epsilon) $ is commonly referred to as a \emph{fixed point} for the DE equation \eqref{Eq:x_ell}. It is known that $ x(\epsilon)=0 $ if and only if $ \epsilon<\epsilon^*_{l,r} $. 
  
The following Proposition provides a tool we use in the sequel for bounding the SG decoding performance using properties of the single-SB code. 
\begin{proposition}\label{Prop:delta(eps)}
	Let $ \udc\left (B_\mathrm{left}\right )=\left (d_1,\ldots,d_t\right ) $ be the check degree profile of $ \Bleft $. For every $ \epsilon\in[0,1] $ define
	$ \underline\phi(\epsilon)=\left (\phi_{1}(\epsilon),\phi_{2}(\epsilon),\ldots,\phi_{t}(\epsilon)\right )$ and $ \underline\psi(\epsilon)=\left (\psi_{1}(\epsilon),\psi_{2}(\epsilon),\ldots,\psi_{t}(\epsilon)\right )$ by
	\begin{align}\label{Eq:delUdelD}
	\begin{array}{lll}
	\phi_{i}(\epsilon) &\triangleq &1-(1-x(\epsilon))^{r-d_i}\\
	\psi_{i}(\epsilon) &\triangleq &1-(1-x(\epsilon))^{d_i}
	\end{array},\quad 1\leq i\leq t.
	\end{align}
	Then, for every $ \epsilon\in [0,1] $
	\begin{subequations}
		\begin{align}
		\label{Eq:DeltLMono}
		&\Delta_L\left (\epsilon,\underline\phi(\epsilon) \right )\succeq\underline\phi(\epsilon), \\
		\label{Eq:DeltRMono}
		&\Delta_R\left (\epsilon,\underline\psi(\epsilon) \right )\succeq\underline\psi(\epsilon).
		\end{align}
	\end{subequations}
\end{proposition}

\begin{proof}
	For $ \epsilon<\epslr $ we have $ x(\epsilon)=0 $ and thus $ \underline\phi(\epsilon)=\underline\psi(\epsilon)=\underline 0 $. In this case,  \eqref{Eq:DeltLMono} and \eqref{Eq:DeltRMono} follow trivially.
	For $ \epsilon\geq\epslr $ we will first show that if  $ \underline\delta_L=\underline\phi(\epsilon) $ and $ \underline\delta_R=\underline\psi(\epsilon) $, then $ x(\epsilon)>0 $ is a SG-DE fixed-point (see Section~\ref{Sub:SG_DE}), i.e., we will show that if all VN-to-CN values in some DE iteration $ \ell $ are equal to $x_\ell(\epsilon) = x(\epsilon) $, then after one iteration, all of the the VN-to-CN values in iteration $ \ell+1 $ remain $ x(\epsilon) $.
	
	Consider row $ i\in\{1,2,\ldots,t\} $ in $ \Bleft $. This row represents a coupling check node with an incoming DE value $\phi_{i}(\epsilon)  $ given in \eqref{Eq:delUdelD}, and $ d_i $ edges between it and the VNs of the SB. By our assumption, $x_\ell(\epsilon)=x(\epsilon)$. Thus, if we mark the outgoing CN-to-VN message from a left coupling check node $ i $ to any VN in the SB in iteration $ \ell $ by $u^{(i)}_\ell(\epsilon)$, then for every $ i\in\{1,2,\ldots,t\} $ we have 
	\begin{align}\label{Eq:LCC2VN}
	\begin{split}
	u^{(i)}_\ell(\epsilon) 
	&= 1-(1-\phi_{i}(\epsilon))(1-x(\epsilon))^{d_i-1}\\
	&= 1-(1-x(\epsilon))^{r-d_i}(1-x(\epsilon))^{d_i-1}\\
	&= 1-(1-x(\epsilon))^{r-1}
	\end{split}
	\end{align}
	(see Figure~\ref{Fig:SG_DE_Proof}(a) for a graphical representation of \eqref{Eq:LCC2VN}).
	Similarly, let $w^{(i)}_\ell(\epsilon)$ denote the outgoing CN-to-VN message from a right coupling check node $ i\in\{1,2,\ldots,t\}  $ to any VN in the SB. In view of \eqref{Eq:dBU=dBD1}--\eqref{Eq:dBU=dBD3}, a right coupling check node with index $ i\in\{1,\ldots,t\} $ is connected to VNs in the SB via $ r-d_i $ edges. Hence,
	\begin{align}\label{Eq:RCC2VN}
	\begin{split}
	w^{(i)}_\ell(\epsilon) 
	&= 1-(1-\psi_{i}(\epsilon))(1-x(\epsilon))^{r-d_i-1}\\
	&= 1-(1-x(\epsilon))^{d_i}(1-x(\epsilon))^{r-d_i-1}\\
	&= 1-(1-x(\epsilon))^{r-1},
	\end{split}
	\end{align}
	(see Figure~\ref{Fig:SG_DE_Proof}(b)). 
	Moreover, let $y_\ell(\epsilon)$ denote the outgoing CN-to-VN message from a local check node to any VN in the SB. Since every local check node is of degree $ r $, then
	\begin{align}
	\label{Eq:LC2VN}
	y_\ell(\epsilon) = 1-(1-x(\epsilon))^{r-1}.
	\end{align}
	In view of \eqref{Eq:LCC2VN}--\eqref{Eq:LC2VN}, given that in iteration $ \ell $ all of the DE values emanating from every VN equal $ x_\ell(\epsilon) = x(\epsilon)$, and in addition $ \underline\delta_L=\underline\phi(\epsilon) $ and $ \underline\delta_R=\underline\psi(\epsilon) $, then all of the CN-to-VN DE values equal to $1-(1-x(\epsilon))^{r-1}  $. 
	Furthermore, every VN in the SB has degree $ l $, thus the DE values from any VN to any CN in iteration $ \ell+1 $ are given by
	$x_{\ell+1}(\epsilon) = \epsilon\cdot  (1-(1-x\left (\epsilon)\right )^{r-1} )^{l-1} $. 
	In view of \eqref{Eq:xs}, this implies that $x_{\ell+1}(\epsilon)=x(\epsilon)= x_{\ell}(\epsilon) $, and thus $ x(\epsilon) $ is a fixed point when $ \underline\delta_L=\underline\phi(\epsilon) $ and $ \underline\delta_R=\underline\psi(\epsilon) $. 
	
	In order to calculate the erasure transfer of a left helper SB $ \Delta_L\left (\epsilon,\underline \phi(\epsilon)\right ) $, we should run the SG DE equations with $ \epsilon $ as the channel parameter, and with $ \udel_L=\underline \phi(\epsilon) $ and $ \udel_R=\underline 1 $ as incoming DE values from the left and right, respectively. From the above result and from monotonicity of DE, 
	each outgoing (right) CC receives DE values from the SB that are lower bounded by $ x(\epsilon) $. Since the degree in the SB of an outgoing CC $ i\in\{1,2,\ldots,t\} $ is $ r-d_i $, then the outgoing DE values $ \Delta_L\left (\epsilon,\underline \phi(\epsilon)\right ) $ are element-wise lower bounded by $ \left ((1-x(\epsilon))^{r-d_1},\ldots,(1-x(\epsilon))^{r-d_t} \right ) $, which in view \eqref{Eq:delUdelD} of proves \eqref{Eq:DeltLMono}. Similarly, \eqref{Eq:DeltRMono} holds for right helper SBs.

\end{proof}

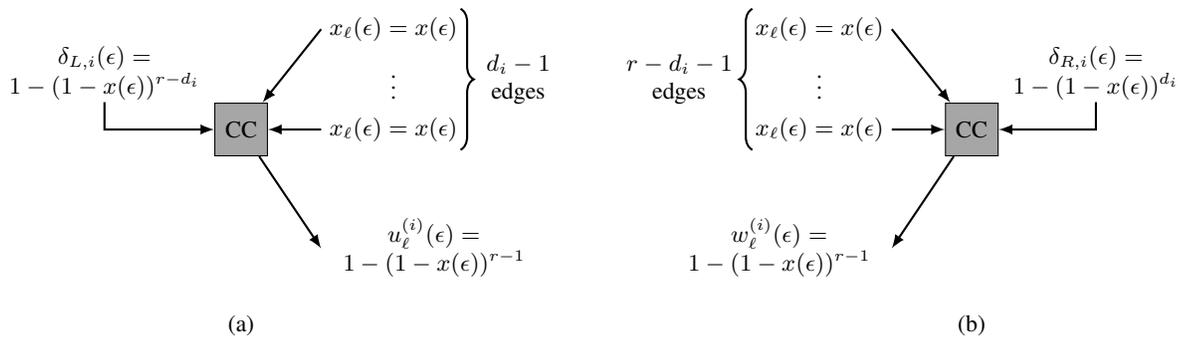
\begin{figure}
	\begin{tikzpicture}
	\tikzstyle{cnode}=[rectangle,draw,fill=gray!70!white,minimum size=7mm]
	\tikzstyle{edge}=[thick,->,>=latex]
	\pgfmathsetmacro{\x}{90}
	\pgfmathsetmacro{\y}{10}
	\footnotesize
	\node (ccl) [cnode] {CC}; \node (dl) [above left=0.01*\y mm of ccl,inner sep=0pt] {$\begin{array}{c}
		 \delta_{L,i}(\epsilon)=\\1-(1-x(\epsilon))^{r-d_i}
		\end{array}$}; \node (x) at (dl|-ccl) {};
	\node (e1) [above right=\y mm of ccl] {$ x_\ell(\epsilon)=x(\epsilon) $};
	\node (e2) at (ccl-|e1) {$ x_\ell(\epsilon)=x(\epsilon) $};
	\node (e3) [below right=\y mm of ccl] {$ \begin{array}{c}
		u^{(i)}_\ell(\epsilon)=\\1-(1-x(\epsilon))^{r-1}
		\end{array} $};
	\node at ($(e1)!0.5!(e2)$) {$ \vdots $ };
	
	\draw [edge] (dl)--(x.center)--(ccl);\draw [edge] (e1.west)--(ccl);\draw [edge] (e2.west)--(ccl);\draw [edge] (ccl)--(e3.west);
	\draw [thick,decorate,decoration={brace,amplitude=5pt,raise=-2pt}] (e1.north east)--(e2.south east) node [black,midway,xshift=7mm,text width=10mm,align=center] {\footnotesize $ d_i-1 $ edges};
	
	\node (a) [below = 2*\y mm of ccl] {(a)};
	
	
	\node (ccr) [cnode,right = \x mm of ccl] {CC}; \node (dr) [above right=0.01*\y mm of ccr,inner sep=0pt] {$\begin{array}{c}
		\delta_{R,i}(\epsilon)=\\1-(1-x(\epsilon))^{d_i}
		\end{array}$};\node (x) at (ccr-|dr) {};
	\node (e1) [above left=\y mm of ccr] {$ x_\ell(\epsilon)=x(\epsilon) $};
	\node (e2) at (e1|-ccr) {$ x_\ell(\epsilon)=x(\epsilon) $};
	\node (e3) [below left=\y mm of ccr] {$ \begin{array}{c}
		w^{(i)}_\ell(\epsilon)=\\1-(1-x(\epsilon))^{r-1}
		\end{array} $};
	\node at ($(e1)!0.5!(e2)$) {$ \vdots $ };
	
	\draw [edge] (dr)--(x.center)--(ccr);\draw [edge] (e1.east)--(ccr);\draw [edge] (e2.east)--(ccr);\draw [edge] (ccr)--(e3.east);
	\draw [thick,decorate,decoration={brace,amplitude=5pt,raise=-2pt}] (e2.south west)--(e1.north west) node [black,midway,xshift=-9mm,text width=15mm,align=center] {\footnotesize $ r-d_i-1 $ edges};
	
	\node (b) at (a-|ccr) {(b)};

	\end{tikzpicture}
	\caption{\label{Fig:SG_DE_Proof}Graphical representations of the coupling-check-nodes updates in the proof of Proposition~\ref{Prop:delta(eps)}: (a) left CC, row $ i $ in $ \Bleft $; (b) right CC, row $ i $ in $ \Bright $.}
\end{figure}

Protographs with the following definition of symmetry simplify the analysis.

\begin{definition}\label{Def:SymSB}
	Let $\left (B_\mathrm{left}\;;\;B_\mathrm{loc}\;;\;B_\mathrm{right}\right )$ be a SB protomatrix. We say that the SB is \emph{symmetric} if using only row and column permutations one can transform $B_\mathrm{left}$ into $B_\mathrm{right}$.
\end{definition}

\begin{example}\label{Ex:SymSB}
	The SB in Example~\ref{Ex:361_SCLDPCL} is symmetric while the SB in Example~\ref{Ex:l4r16Const} is not symmetric.
	Moreover, it can be shown that a SB from the cutting-vector $ (l,r,t) $ SC-LDPCL construction in \cite{RamCass18} is symmetric if and only if $ r $ is divisible by $ t+1 $.
\end{example}
For symmetric SBs we have
\begin{align*}
\Delta_L\left (\epsilon,\udi\right )= \Delta_R \left (\epsilon,\udi\right ),\quad\forall\epsilon\in[0,1],\quad \forall \udi \in [0,1]^t.
\end{align*}

\begin{definition}\label{Def:Deltak}
	Let $ k,t\in\mathbb N $, $ \epsilon\in[0,1]$, and $\udel\in[0,1]^t$. 
	We denote by $ \DelfLk{k}{\epsilon} $ (respectively $ \DelfRk{k}{\epsilon} $), the function that computes $ \DelfL{\epsilon} $ (respectively $ \DelfR{\epsilon} $) recursively $ k $ times, i.e.,
	\begin{align*}
	&\DelfLk{1}{\epsilon}=\DelfL{\epsilon},\\
	&\DelfLk{k}{\epsilon}=\Delta\left (\epsilon,\DelfLk{k-1}{\epsilon}\right ),\quad k>1.
	\end{align*}
\end{definition}

When considering an erasure-transfer function $\Delta(\epsilon, \udel)$ (right or left), we identify three important channel-parameter thresholds\footnote{These thresholds are properties of the protograph used, however, for ease of reading we make this implicit in the notations.} $ \epsilon^*_1, \epsilon^*_2 ,\epsilon^*_3  $. We call them the \emph{SB thresholds}. They play an important role in the analysis and code design that we perform in Sections~\ref{Sec:SBMV} and~\ref{Sec:Thrsholds}. The first threshold $ \epsilon^*_1$ is the largest channel parameter such that local decoding of the SB succeeds, regardless of the input DE values, i.e.,
\begin{align}\label{Eq:s1}
\Delf{\epsilon}=0  ,\quad \forall \epsilon<\epsilon^*_1,\quad\forall \udel \in [0,1]^t.
\end{align}
This threshold can be easily calculated by the observation that $ \epsilon^*_1=\epsilon^*\left (B_{\text{loc}}\right ) $, where $ \epsilon^*\left (B_{\text{loc}}\right ) $ is the threshold induced by the protograph $ B_{\text{loc}} $. 
The second threshold $ \epsilon^*_2$ is the largest channel parameter such that the outgoing DE values are element-wise smaller (with strict inequality in at least one element) than the incoming DE values, i.e., 
\begin{align}\label{Eq:s2}
\Delf{\epsilon}\prec\udel  ,\quad \forall \epsilon<\epsilon^*_2,\quad \forall \udel \in [0,1]^t.
\end{align}
This means that a sequence of inter-SB DE values passing between SBs with channel parameters $ \epsilon<\epsilon^*_2 $, will gradually decrease towards zero. 
We calculate $ \epsilon^*_2 $ by computing $ \Delf{\epsilon} $ for all $ \udel \in [0,1]^t $ and with increasing values of $ \epsilon $ until \eqref{Eq:s2} is violated. 
In view of Definition~\ref{Def:Deltak} and \eqref{Eq:s2}, for every $ \epsilon<\epsilon^*_2 $, we define
\begin{align}\label{Eq:q}
q(\epsilon)=\min\{k\in\mathbb N\colon \Delta^{(k)}\left (\epsilon,\underline 1\right )=\underline 0\}. 
\end{align} 
The value of $ q(\epsilon) $ will play an important role later in the evaluation of semi-global decoding (Section~\ref{Sub:SBMVPerf}).  

The third threshold $ \epsilon^*_3$ is the largest channel parameter such that zero incoming DE values are preserved, i.e., 
\begin{align}\label{Eq:s3}
\Delta(\epsilon,\underline 0) =\underline 0,\quad \forall \epsilon<\epsilon^*_3.
\end{align}

\begin{example}\label{Ex:thresholds}
Consider the $ (l=3,r=6,t=1) $ SC-LDPCL protograph from Example~\ref{Ex:361_SCLDPCL}. 
Since $ t+1=2 $ divides $ r=6 $, then the SBs in this protograph are symmetric (see a remark in Example~\ref{Ex:SymSB}), so we omit the $ L $ and $ R $ sub-scripts and mark  $  \Delta(\cdot, \cdot)  \triangleq \Delta_L(\cdot, \cdot)  = \Delta_R(\cdot, \cdot)$, and $ \delta \triangleq \delta_L=\delta_R $ (note that for $t=1$, $\delta$ is a scalar).
The SB thresholds of this protograph are given by $ \epsilon^*_1=0.2,\epsilon^*_2=0.3719,\epsilon^*_3=0.4297. $ 
Figure~\ref{Fig:delta_O} illustrates erasure-transfer functions for a SB in this protograph and for channel parameters $\epsilon\in\{0.18,0.3547,0.4239,0.5438\}$ (note that each of these channel parameters falls into a different threshold interval, i.e.,  $0< 0.18<\epsilon^*_1<0.3547<\epsilon^*_2<0.4239<\epsilon^*_3<0.5438<1 $). Also shown is the value $ q(0.3547) $ in \eqref{Eq:q}; the dotted red line shows the inter-SB DE values, starting from $ 1 $ (all erasure) down to $ 0 $ (no erasures) after $ q(0.3547)= 4$ SBs. 

Note that changing the value of the channel parameters could change the plots, however, different channel parameters within the same threshold interval will produce plots with the same behavior and properties as described in \eqref{Eq:s1}--\eqref{Eq:s3}.
\begin{figure}
	\begin{center}
		\begin{tikzpicture}

\begin{axis}[%
width=6in,
height=3in,
at={(0,0)},
scale only axis,
xmin=0,
xmax=1,
xlabel style={font=\color{black}},
xlabel={$\delta$},
ymin=0,
ymax=1,
ylabel style={font=\color{black}},
ylabel={$\Delta(\epsilon,\udel)$},
xmajorgrids,
ymajorgrids,
legend style={
	at={(1,0.1)}, 
	anchor=south east, 
	legend cell align=left, 
	align=left, draw=black
}
]

\addplot [color=orange,line width=1pt]
table[row sep=crcr]{%
	0	0.743408934278332\\
	0.01	0.745504811766463\\
	0.02	0.747583001509365\\
	0.03	0.749643720467772\\
	0.04	0.75168718092428\\
	0.05	0.753713590634806\\
	0.06	0.755723152973665\\
	0.07	0.757716067072591\\
	0.08	0.75969252795401\\
	0.09	0.761652726658861\\
	0.1	0.763596850369228\\
	0.11	0.765525082526046\\
	0.12	0.767437602942124\\
	0.13	0.769334587910702\\
	0.14	0.771216210309771\\
	0.15	0.773082639702349\\
	0.16	0.774934042432908\\
	0.17	0.776770581720128\\
	0.18	0.778592417746156\\
	0.19	0.780399707742524\\
	0.2	0.782192606072885\\
	0.21	0.783971264312701\\
	0.22	0.785735831326037\\
	0.23	0.787486453339571\\
	0.24	0.789223274013953\\
	0.25	0.790946434512633\\
	0.26	0.792656073568254\\
	0.27	0.794352327546736\\
	0.28	0.796035330509121\\
	0.29	0.797705214271309\\
	0.3	0.799362108461738\\
	0.31	0.801006140577126\\
	0.32	0.802637436036327\\
	0.33	0.804256118232404\\
	0.34	0.805862308582971\\
	0.35	0.807456126578887\\
	0.36	0.809037689831364\\
	0.37	0.810607114117552\\
	0.38	0.812164513424659\\
	0.39	0.813709999992671\\
	0.4	0.81524368435572\\
	0.41	0.816765675382152\\
	0.42	0.818276080313348\\
	0.43	0.819775004801353\\
	0.44	0.821262552945335\\
	0.45	0.822738827326945\\
	0.46	0.824203929044606\\
	0.47	0.825657957746764\\
	0.48	0.82710101166416\\
	0.49	0.828533187641133\\
	0.5	0.829954581166011\\
	0.51	0.831365286400611\\
	0.52	0.832765396208886\\
	0.53	0.834155002184742\\
	0.54	0.835534194679062\\
	0.55	0.836903062825964\\
	0.56	0.838261694568313\\
	0.57	0.839610176682519\\
	0.58	0.840948594802648\\
	0.59	0.842277033443861\\
	0.6	0.843595576025209\\
	0.61	0.84490430489181\\
	0.62	0.84620330133642\\
	0.63	0.847492645620419\\
	0.64	0.848772416994247\\
	0.65	0.850042693717281\\
	0.66	0.851303553077197\\
	0.67	0.852555071408819\\
	0.68	0.853797324112475\\
	0.69	0.855030385671881\\
	0.7	0.856254329671556\\
	0.71	0.857469228813798\\
	0.72	0.858675154935222\\
	0.73	0.859872179022885\\
	0.74	0.861060371229999\\
	0.75	0.862239800891255\\
	0.76	0.863410536537764\\
	0.77	0.864572645911631\\
	0.78	0.865726195980165\\
	0.79	0.866871252949752\\
	0.8	0.86800788227938\\
	0.81	0.869136148693847\\
	0.82	0.870256116196645\\
	0.83	0.871367848082543\\
	0.84	0.872471406949862\\
	0.85	0.873566854712473\\
	0.86	0.874654252611505\\
	0.87	0.875733661226779\\
	0.88	0.876805140487987\\
	0.89	0.8778687496856\\
	0.9	0.878924547481537\\
	0.91	0.879972591919584\\
	0.92	0.881012940435581\\
	0.93	0.882045649867377\\
	0.94	0.883070776464561\\
	0.95	0.884088375897981\\
	0.96	0.885098503269047\\
	0.97	0.886101213118831\\
	0.98	0.887096559436965\\
	0.99	0.888084595670352\\
	1	0.889065374731678\\
};
\addlegendentry{$\epsilon=0.5438\in(\epsilon^*_3,1)$}

\addplot [color=cyan,line width=1pt]
table[row sep=crcr]{%
	0	0\\
	0.01	6.66133814775094e-16\\
	0.02	9.99200722162641e-16\\
	0.03	0.109005425639857\\
	0.04	0.161641740149829\\
	0.05	0.186290044993605\\
	0.06	0.205772449850714\\
	0.07	0.222545969383147\\
	0.08	0.237564522644154\\
	0.09	0.251319304467609\\
	0.1	0.264103441282026\\
	0.11	0.276108841970643\\
	0.12	0.287469467014207\\
	0.13	0.298283342575036\\
	0.14	0.308624845937651\\
	0.15	0.318552055733137\\
	0.16	0.32811139546395\\
	0.17	0.337340698479743\\
	0.18	0.346271305369791\\
	0.19	0.35492954332804\\
	0.2	0.363337796827614\\
	0.21	0.371515299903803\\
	0.22	0.379478733877978\\
	0.23	0.387242686030213\\
	0.24	0.39482000691018\\
	0.25	0.402222092452746\\
	0.26	0.409459109427732\\
	0.27	0.416540177580215\\
	0.28	0.423473518244064\\
	0.29	0.430266576698543\\
	0.3	0.436926123742249\\
	0.31	0.443458340656588\\
	0.32	0.449868890774107\\
	0.33	0.456162980154989\\
	0.34	0.462345409339171\\
	0.35	0.468420617733997\\
	0.36	0.4743927218843\\
	0.37	0.480265548629181\\
	0.38	0.486042663960119\\
	0.39	0.491727398245589\\
	0.4	0.497322868368768\\
	0.41	0.502831997230047\\
	0.42	0.50825753098979\\
	0.43	0.513602054364969\\
	0.44	0.518868004242988\\
	0.45	0.524057681834811\\
	0.46	0.529173263555525\\
	0.47	0.534216810792471\\
	0.48	0.539190278697638\\
	0.49	0.544095524121617\\
	0.5	0.548934312789972\\
	0.51	0.553708325809195\\
	0.52	0.558419165577707\\
	0.53	0.563068361167528\\
	0.54	0.567657373233817\\
	0.55	0.57218759850227\\
	0.56	0.576660373878241\\
	0.57	0.581076980216108\\
	0.58	0.585438645782855\\
	0.59	0.58974654944589\\
	0.6	0.594001823611652\\
	0.61	0.598205556938614\\
	0.62	0.60235879684566\\
	0.63	0.606462551834536\\
	0.64	0.610517793643089\\
	0.65	0.614525459244247\\
	0.66	0.618486452704153\\
	0.67	0.6224016469115\\
	0.68	0.626271885188901\\
	0.69	0.63009798279609\\
	0.7	0.633880728333757\\
	0.71	0.637620885056023\\
	0.72	0.641319192098774\\
	0.73	0.644976365630419\\
	0.74	0.648593099931053\\
	0.75	0.652170068405415\\
	0.76	0.655707924534627\\
	0.77	0.659207302771194\\
	0.78	0.662668819381401\\
	0.79	0.666093073238875\\
	0.8	0.669480646572768\\
	0.81	0.672832105673717\\
	0.82	0.676148001560509\\
	0.83	0.679428870610095\\
	0.84	0.682675235153444\\
	0.85	0.685887604039475\\
	0.86	0.689066473169174\\
	0.87	0.692212326001814\\
	0.88	0.695325634035062\\
	0.89	0.698406857260628\\
	0.9	0.701456444596966\\
	0.91	0.704474834300461\\
	0.92	0.707462454356395\\
	0.93	0.710419722850918\\
	0.94	0.713347048325162\\
	0.95	0.716244830112526\\
	0.96	0.719113458660144\\
	0.97	0.721953315835418\\
	0.98	0.724764775218485\\
	0.99	0.727548202381402\\
	1	0.730303955154792\\
};
\addlegendentry{$\epsilon=0.4239\in(\epsilon^*_2,\epsilon^*_3) $}

\addplot [color=red,line width=1pt]
  table[row sep=crcr]{%
0	0\\
0.01	6.66133814775094e-16\\
0.02	1.33226762955019e-15\\
0.03	2.33146835171283e-15\\
0.04	2.66453525910038e-15\\
0.05	3.99680288865056e-15\\
0.06	3.99680288865056e-15\\
0.07	5.32907051820075e-15\\
0.08	5.99520433297585e-15\\
0.09	6.32827124036339e-15\\
0.1	6.99440505513849e-15\\
0.11	8.65973959207622e-15\\
0.12	9.99200722162641e-15\\
0.13	9.65894031423886e-15\\
0.14	9.99200722162641e-15\\
0.15	1.19904086659517e-14\\
0.16	1.19904086659517e-14\\
0.17	1.29896093881143e-14\\
0.18	1.43218770176645e-14\\
0.19	1.59872115546023e-14\\
0.2	1.76325620770967e-12\\
0.21	3.83372279477712e-08\\
0.22	3.06613109425369e-05\\
0.23	0.00043885766684626\\
0.24	0.0016338903097507\\
0.25	0.0040209865060612\\
0.26	0.00810983455036318\\
0.27	0.0143276257563153\\
0.28	0.0227176403005931\\
0.29	0.0328448123349178\\
0.3	0.0440662287994253\\
0.31	0.0558313002451578\\
0.32	0.0677762780926909\\
0.33	0.0796902129014074\\
0.34	0.0914592177473008\\
0.35	0.103025445636137\\
0.36	0.114362289312337\\
0.37	0.125460330455003\\
0.38	0.13631951319735\\
0.39	0.146944778731009\\
0.4	0.157343613790698\\
0.41	0.167524666227893\\
0.42	0.177496963298401\\
0.43	0.187269474902781\\
0.44	0.196850876436266\\
0.45	0.206249427956348\\
0.46	0.215472921216235\\
0.47	0.22452866601795\\
0.48	0.233423498901652\\
0.49	0.242163804006271\\
0.5	0.250755540010068\\
0.51	0.259204269519744\\
0.52	0.267515188774178\\
0.53	0.275693156445011\\
0.54	0.283742720877265\\
0.55	0.291668145455368\\
0.56	0.299473431986384\\
0.57	0.307162342114144\\
0.58	0.314738416847184\\
0.59	0.32220499431984\\
0.6	0.329565225922113\\
0.61	0.336822090937985\\
0.62	0.343978409828671\\
0.63	0.351036856290108\\
0.64	0.357999968204744\\
0.65	0.364870157597603\\
0.66	0.371649719696545\\
0.67	0.378340841186795\\
0.68	0.38494560774077\\
0.69	0.391466010895728\\
0.7	0.397903954344144\\
0.71	0.404261259694788\\
0.72	0.41053967175627\\
0.73	0.416740863389285\\
0.74	0.422866439968834\\
0.75	0.428917943493333\\
0.76	0.434896856373598\\
0.77	0.440804604931231\\
0.78	0.446642562632898\\
0.79	0.4524120530842\\
0.8	0.458114352804472\\
0.81	0.463750693801667\\
0.82	0.469322265964569\\
0.83	0.474830219287884\\
0.84	0.480275665944236\\
0.85	0.485659682215741\\
0.86	0.49098331029664\\
0.87	0.49624755997735\\
0.88	0.501453410219388\\
0.89	0.506601810629702\\
0.9	0.5116936828422\\
0.91	0.516729921813551\\
0.92	0.521711397039719\\
0.93	0.526638953699123\\
0.94	0.531513413727796\\
0.95	0.536335576831473\\
0.96	0.541106221439116\\
0.97	0.545826105602002\\
0.98	0.550495967842166\\
0.99	0.555116527953698\\
1	0.559688487760067\\
};
\addlegendentry{$\epsilon= 0.3547 \in(\epsilon^*_1,\epsilon^*_2)$}

\addplot [color=blue,line width=1pt]
table[row sep=crcr]{%
	0	1e-4\\
	1	1e-4\\	
};
\addlegendentry{$\epsilon=0.1800\in(0,\epsilon^*_1) $}

\addplot [color=black, line width=1pt]
table[row sep=crcr]{%
	0	0\\
	1	1\\
};

\addplot [color=red,line width=0.75pt,dotted]
table[row sep=crcr]{%
	0		0\\
	0.04	0\\
	0.04	0.04\\
	0.3	0.0440662287994253\\
	0.3	0.3\\
	0.56	0.299473431986384\\
	0.56	0.56\\
	1	0.559688487760067\\
};

\end{axis}

\end{tikzpicture}%
		\caption{\label{Fig:delta_O} $  \Delta(\epsilon, \delta)  $ for the $ (l=3,r=6,t=1) $ protograph and for four possible channel parameters, each in a different interval defined by the thresholds in \eqref{Eq:s1}--\eqref{Eq:s3}: $0.18$ (blue), $ 0.3547$ (red), $0.4239$ (cyan), $0.5438$ (orange); black represents the identity function $\Delta(\epsilon,\delta)=\delta$. For $ \epsilon=0.3547 $, four consecutive SBs suffice to decrease the inter-SB DE values from $ 1 $ to $ 0 $ (dotted red); i.e., $ q(0.3547 )=4 $.}
	\end{center}
\end{figure}
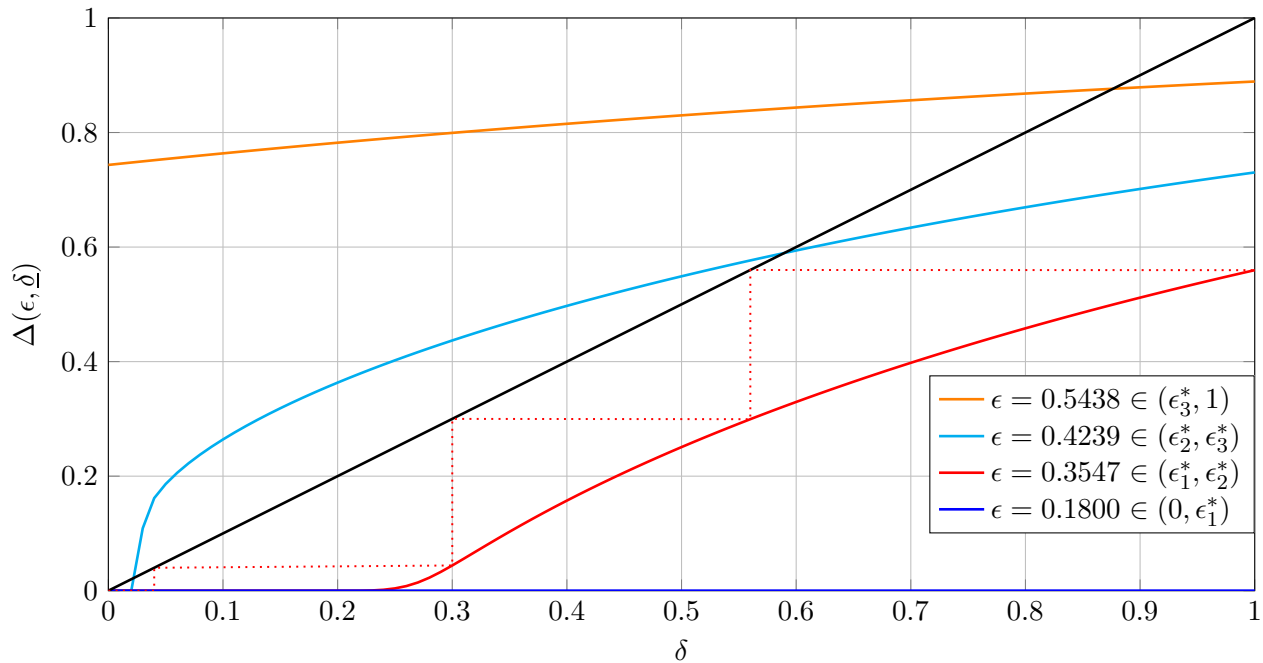
\end{example}

\subsection{Target Threshold}\label{Sub:TargetTh}

We define the target threshold $ \epsilon^*(\udel_L, \udel_R) $ as the maximum erasure rate such that the target is successfully decoded \emph{locally} given incoming DE values $ \udel_L $ and $ \udel_R $. For example, if $ \udel_L=\udel_R=\underline 1 $, then the target cannot use any side information, and its threshold is the local threshold, i.e., $ \epsilon^*(\underline 1,\underline 1)=\epsilon^*_1  $. 
Furthermore, we have the following result.
\begin{proposition}\label{Prop:s3}
	 For left and right helper SBs, $ \epsilon^*(\underline 0,\underline 1)=\epsilon^*_3  $, and $ \epsilon^*(\underline 1,\underline 0)=\epsilon^*_3  $, respectively.
\end{proposition}
\begin{proof}
	Recall that $ \epsilon^*_3 $ is largest channel parameter such that a helper SB that receives zero DE values from the previous decoded helper SB outputs zero DE values to the next SB. In view of Proposition~\ref{Prop:del>0Gen}, we get that $ \epsilon^*_3 $  is the largest channel parameter such that a helper SB is successfully decoded with zero incoming DE values. Since helper SBs receive information from one side only (they output information to the other side), then $  \epsilon^*_3 =  \epsilon^*(\underline 0,\underline 1)$ for left helper SBs and $  \epsilon^*_3 =  \epsilon^*(\underline 1,\underline 0)$ for right helper SBs.
\end{proof}
Proposition~\ref{Prop:s3} provides a simple way to calculate $ \epsilon^*_3 $. Consider a SB given by $ [B_{\text{left}};B_{\text{loc}};B_{\text{right}}] $. Zero incoming DE values means that all incoming CCs are fully available, i.e., they can be treated as LCs. In addition, outgoing CCs are not valid.
We get that for left and right helper SBs
\begin{subequations}
\begin{align}\label{Eq:s3leftCalc}
\epsilon^*_3 =\epsilon^*([B_{\text{left}};B_{\text{loc}}]),
\end{align}
and
\begin{align}\label{Eq:s3rightCalc}
\epsilon^*_3 =\epsilon^*([B_{\text{loc}};B_{\text{right}}]),
\end{align}
respectively.	
\end{subequations}

We add another SB threshold that will be used in the sequel
\begin{align}\label{Eq:s4}
\epsilon^*_4=\epsilon^*(\underline 0,\underline 0).
\end{align}
It can be verified that
\begin{align}\label{Eq:order_eps*}
\epsilon^*_1\leq \epsilon^*_2\leq \epsilon^*_3\leq \epsilon^*_4.
\end{align}

The following result is the parallel of Proposition~\ref{Prop:delta(eps)} applying to the target SB.
\begin{proposition}\label{Prop:targetTh}
	Let $ \underline\phi(\epsilon),\underline\psi(\epsilon)$ be defined as in Proposition~\ref{Prop:delta(eps)}. Then, for every $ \epsilon\in[0,1] $
	\begin{align}\label{Eq:TargetTh}
	\epsilon^*\left (\underline\phi(\epsilon),\underline\psi(\epsilon)\right )<\epsilon^*_{l,r}.
	\end{align}
\end{proposition}

\begin{proof}
	Using the same arguments in the proof of Proposition~\ref{Prop:delta(eps)}, we get that if the left and right incoming DE values are $ \underline\phi(\epsilon)$ and $\underline\psi(\epsilon)$, respectively, then $ x(\epsilon) $ is a semi-global fixed point. Since $ x(\epsilon)>0 $ for every $ \epsilon>\epslr $, then  target decoding fails. Hence, $ \epsilon>  \epsilon^*\left (\underline\phi(\epsilon),\underline\psi(\epsilon)\right )$. Since this is true for all $ \epsilon>\epslr $, then \eqref{Eq:TargetTh} holds.	
\end{proof}

\begin{remark}
	For symmetric SBs (Definition~\ref{Def:SymSB}) we have $ \epsilon^*(\udel_L, \udel_R) = \epsilon^*(\udel_R, \udel_L) $, for every $ \udel_L,\udel_R\in [0,1]^t $.
\end{remark}
%
%

\section{Memoryless Channels: Threshold Analysis} \label{Sec:Thrsholds}
In this section, we characterize and calculate the semi-global thresholds of SC-LDPCL protographs. We focus in this section on memoryless channels, before generalizing to channels with memory in the next section.
Formally, the semi-global threshold $ \epsilon^*_{SG}(m,d) $ is defined as the maximal erasure rate such that SG decoding successfully decodes SB $ m $ with $ d $ helper SBs. In general, semi-global thresholds depend on the protograph (i.e., local graph, coupling checks and total number of SBs), the position of the target SB $ m $, and on the number of helper SBs accessed during semi-global decoding $ d $. However, in this section, we study the case of many helper SBs (in the limit where $ d\to \infty $) and show that in this case semi-global thresholds depend solely on the single-SB structure and on the question whether the first accessed helper SBs are termination SBs (i.e., at the endpoint of the coupled chain of SBs) or not. In what follows, we use a binary parameter $ \tau\in\{0,1\} $ to mark the first-accessed-helper type: termination ($\tau= 0 $) or non-termination ($\tau= 1 $). While needing to start from a termination SB is undesirable in practical SG decoding of a random target SB, analyzing this mode is useful for ``pseudo-termination" SBs defined in the next section for channels with variability. The results in this section apply to unit-memory binary-regular SC-LDPC codes, but for terseness we keep this implicit in most of the result statements.
	
Our ultimate goal in this section is to characterize the SG thresholds defined by
\begin{definition}\label{Def:esgj}
	Let $ \tau_L,\tau_R\in\{0,1\} $ indicate termination SBs ($ \tau=0 $) or not ($ \tau=1 $) from left and right of the target SB, respectively. We define $ \esgj $ as the largest channel parameter such that SG successfully decodes the target SB in the limit of $ d\to\infty $ helper SBs.
\end{definition}

For simplicity of the derivations, we assume that the SBs in the code are symmetric (see Definition~\ref{Def:SymSB}). Consequently, we remove the subscripts that were used to indicate left and right helper SBs (i.e., use $ \Delta  $ to mark both $ \Delta_L $ and $ \Delta_R $). In addition, in view of \eqref{Eq:delUdelD}, we have that for every $ \epsilon\in[0,1] $, $ \underline\phi(\epsilon)=\underline\psi(\epsilon) $.

For $\tau\in\{0,1\} $, let $ \underline\delta_{0}^{(\tau)}(\epsilon),\underline\delta_{1}^{(\tau)}(\epsilon) ,\underline\delta_{2}^{(\tau)}(\epsilon) ,\ldots $ be the sequence of inter-SB DE values between helper SBs during semi-global decoding, with $ \tau=0$ if decoding starts in a termination SB, and $\tau= 1 $ otherwise, i.e.,
\begin{align}\label{Eq:Delta}
\begin{split}
&\udel_{i+1}^{(\tau)}(\epsilon)=\Delta\left (\epsilon, \udel_{i}^{(\tau)}\right ) ,\quad i\geq 0,\,\tau=0,1,\\
&\udel_{0}^{(0)} = \underline 0,\quad \udel_{0}^{(1)} =\underline  1,
\end{split}
\end{align}
where $ \Delf{\epsilon} $ is the erasure-transfer function given in Section~\ref{Sec:SingleSB}.  
From the monotonicity of the DE equations, for every $ \epsilon $, $  \Delf{\epsilon} $ is monotonically non-decreasing in $ \udel $. Consequently, for $ \tau=0,1 $ the sequences $ \{ \udel_{i}^{(\tau)}(\epsilon)\}_{i\geq 0} $ (which are bounded by $[0,1] $) converge to some limit value.
\begin{definition}\label{Def:del_hat}
	For $ \tau=0,1 $, let $  {\hat{\udel}}^{(\tau)} (\epsilon)=\lim_{i\to \infty} \udel_i^{(\tau)}(\epsilon)$.
\end{definition}
In view of \eqref{Eq:Delta} and Definition~\ref{Def:del_hat}, for every $ \epsilon\in [0,1] $, $\Delta\left (\epsilon,\dhat{\tau}\right ) = \dhat{\tau}$. 

\begin{lemma}\label{Lemma:delhat}
	For every $ \epsilon\in[0,1] $, $ \dhat{1}\succ \underline\phi(\epsilon) $, where $\underline\phi(\epsilon) $ is defined in \eqref{Eq:delUdelD}.
\end{lemma}

\begin{proof}
We will prove by a mathematical induction on $i $ that 
for every 
$ i\geq0 $, $ \udel_{i}^{(1)}\succeq \underline\phi(\epsilon) $. Consequently, this will prove the Lemma.
For $ i=0 $ 
we have $ \udel_{0}^{(1)} =  \underline 1 \succeq \underline \phi(\epsilon)  $. 
Now, assume correctness for some $ i\geq0 $. Then, 
\begin{align*}
\udel_{i+1}^{(1)} 
&\mathop{=}\Delta(\epsilon,\udel_{i}^{(1)}) \\
&\mathop{\succeq}\limits^{(1)} \Delta(\epsilon, \underline\phi(\epsilon)) \\
&\mathop{\succeq}\limits^{(2)} \phi(\epsilon),
\end{align*}
where $ (1) $ is the induction assumption combined with monotonicity of DE, and $ (2) $ follows from Proposition~\ref{Prop:delta(eps)}.
\end{proof}

The following proposition characterizes the SG thresholds defined in Definition~\ref{Def:esgj}.
\begin{proposition}\label{Prop:esgj}
		Let $ \tau_L,\tau_R\in\{0,1\} $ indicate termination SBs from left and right of the target SB, respectively. Then, the semi-global thresholds are given by
		\begin{align}\label{Eq:eps_SG}
		\esgj = \max\left \{\epsilon\colon\epsilon<\epsilon^*\left ( \dhat{\tau_L},\dhat{\tau_R}\right ) \right \},
		\end{align}
		where $ \epsilon^*(\cdot,\cdot) $ and $ \dhat{\tau} $ are the target threshold defined in Section~\ref{Sub:TargetTh}, and the limit incoming DE value to the target from Definition~\ref{Def:del_hat}, respectively.
\end{proposition}
\begin{proof}
Since we consider the limit of many helper SBs, then for every channel parameter $ \epsilon\in[0,1] $, the target SB "sees" incoming  DE values given by $  \dhat{\tau_L} $ and $  \dhat{\tau_R} $ from the left and right, respectively. Consequently, the threshold at the target is $ \epsilon^*\left ( \dhat{\tau_L},\dhat{\tau_R}\right ) $. By definition, the target is decoded successfully if and only if the channel parameter is below the target threshold, i.e., $ \epsilon < \epsilon^*\left ( \dhat{\tau_L},\dhat{\tau_R}\right )$. Finally, the SG thresholds are defined as the largest such channel parameter.
\end{proof}

Clearly, starting from a termination SB -- as opposed to starting from a non-terminating SB -- can only help. Thus, from the symmetry-of-SBs assumption,
\begin{align}\label{Eq:order_SG}
\esg{1}{1}\leq\esg{0}{1}=\esg{1}{0}\leq\esg{0}{0} .
\end{align}

We now turn to evaluate $\esgj  $. Let us first address the non-termination case in which from both sides of the target the first helper is \emph{not} a termination SB, i.e., $ \udel_0=\underline 1 $.

\begin{lemma}\label{Lemma:esg11}
	$ \esg{1}{1}\geq\epsilon^*_2  $, where $ \epsilon^*_2 $ is given in \eqref{Eq:s2}. 
\end{lemma}

\begin{proof}
	In view of \eqref{Eq:s2} and Definition~\ref{Def:del_hat}, 
	\begin{align*}
	\dhat{\tau}=\underline 0,\quad \forall \epsilon\in[0,\epsilon^*_2), \, \tau=0,1.
	\end{align*}
	Thus, for every $ \epsilon<\epsilon^*_2 $, the target receives full side information, i.e., $ \udel_L=\udel_R=\underline 0 $. From \eqref{Eq:s4}--\eqref{Eq:order_eps*} we have
	\begin{align*}
	\epsilon< \epsilon^*_4=\epsilon^*(0,0),
	\end{align*}
	so the target is successfully decoded; hence $ \epsilon<\esg{1}{1} $. Since this is true for all $ \epsilon<\epsilon^*_2 $, we conclude that $ \epsilon^*_2\leq \esg{1}{1} $.
\end{proof}

\begin{lemma} \label{Lemma:hat}
	$  \dhat{0} =\underline 0$ if and only if $ \epsilon<\epsilon^*_3$.
\end{lemma}
\begin{proof}
	Let $ \epsilon<\epsilon^*_3 $, and let $  \udel^{(0)}_0(\epsilon) =\underline 0$. In view of \eqref{Eq:s3}, the first helper outputs $ \udel^{(0)}_1=\underline 0 $. Applying this argument repeatedly, we get that $ \udel_i^{(0)}(\epsilon)=\underline 0$ for all $ i\geq 0 $. Consequently, $  \dhat{0} =\underline 0$.
	
	For the other direction, let $ \epsilon>\epsilon^*_3 $. By the definition of $ \epsilon^*_3 $ we have $ \udel_1^{(0)}(\epsilon)\succ \underline 0 = \udel_0^{(0)}(\epsilon)$. Since $ \Delta(\epsilon,\cdot) $ is monotonically non-decreasing, then 
	\begin{align*}
	\udel_2^{(0)}(\epsilon) 
	&= \Delta\left (\epsilon,\udel_1^{(0)}\right )\\
	&\succ \Delta\left (\epsilon,\udel_0^{(0)}\right )\\
	&= \udel_1^{(0)}.
	\end{align*}
	Repeating this argument for successive $ i>2  $ gives $ \udel_i^{(0)}(\epsilon)\succ\udel_1^{(0)}(\epsilon)\succ \underline 0$, for all $ i \geq 2$, and consequently$  \dhat{0} \succ \underline 0$.
\end{proof}

\begin{lemma}\label{Lemma:eps00}
	$\esg{0}{0} \leq  \epsilon^*_4$.
\end{lemma}
\begin{proof}
	In view of \eqref{Eq:s4}, for every $ \epsilon> \epsilon^*_4 $ the target will fail to decode even if $ \udel_L=\udel_R=0 $. Thus, semi-global decoding will fail and $ \epsilon>\esg{0}{0} $. This is true for all $ \epsilon>  \epsilon^*_4 $, thus $ \esg{0}{0}\leq  \epsilon^*_4$.
\end{proof}

\begin{theorem}\label{Th:eps*SG}
	For every unit-memory binary-regular SC-LDPC protograph, 
	\begin{subequations}
		\begin{align}\label{Eq:epsSG1}
		 \esg{0}{1}\geq\epsilon^*_3 ,
		\end{align}
		and if $ \epsilon^*_{l,r}\leq \epsilon^*_3 $ and for every $ \epsilon>\epsilon^*_3  $, $ \dhat{0}=\dhat{1}$, then 
		\begin{align}\label{Eq:epsSG2}
		\esg{0}{0}=\esg{0}{1}=\epsilon^*_3. 
		\end{align}
	\end{subequations}
\end{theorem}

\begin{remark}
	The conditions $ \epsilon^*_{l,r}\leq \epsilon^*_3 $ and $ \dhat{0}=\dhat{1}$ for $\epsilon>\epsilon^*_3$ hold for many constructions. 
	$ \epsilon^*_{l,r}$ is the threshold of the $ l\times r $ all-ones matrix, while in view of \eqref{Eq:s3leftCalc}--\eqref{Eq:s3rightCalc}, for left helper SBs, $ \epsilon^*_3  $ equals to the threshold of the $ l\times r $ protomatrix $[B_{\text{left}};B_{\text{loc}}]  $. For many assignments of $ B_{\text{left}} $, we will get $ \epsilon^*_{l,r}\leq \epsilon^*_3 $ (see Table~\ref{Tbl:Bleft} in Section~\ref{Sub:code}). 
	In addition, the condition $ \dhat{0}=\dhat{1}$ holds whenever $ \Delf{\epsilon}=\udel $ for a unique $ \udel\in[0,1]^t $. In all of the DE enumerations that we have done, this was the case; see the orange curve in Figure~\ref{Fig:delta_O}
	for an example.
\end{remark}

\begin{proof}
	
	Let $ \epsilon<\epsilon^*_3 $ and assume semi-global decoding with exactly one side starting from termination; w.l.o.g assume termination on the left, i.e., $ \tau_L=0,\tau_R=1 $. If we mark the incoming DE values to the target from the left and right by $\udel_L  $ and $ \udel_R $, respectively, then $ \udel_R=\dhat{1}\preceq\underline 1 $, and in view of Lemma~\ref{Lemma:hat}, $ \udel_L=\dhat{0}=\underline 0 $.
	From monotonicity and from Proposition~\ref{Prop:s3},  
	\begin{align*}
	\epsilon^*\left (\dhat{0},\dhat{1}\right) >\epsilon^*(\underline 0,\underline 1)= \epsilon^*_3 >\epsilon,
	\end{align*}
	hence semi-global decoding succeeds. In view of \eqref{Eq:eps_SG}, we have
	$\epsilon<\esg{0}{1}$. Since this is true for all $ \epsilon<\epsilon^*_3 $, then \eqref{Eq:epsSG1} holds.
	
	To prove \eqref{Eq:epsSG2}, let $ \epsilon>\epsilon^*_3>\epslr $. From the added condition and from Lemma~\ref{Lemma:delhat}, $ \dhat{0}=\dhat{1}\succeq\underline \phi(\epsilon) $. Thus, if semi-global decoding starts from termination (on both sides), the target receives incoming DE values $ \udel_L=\udel_R\succeq\underline \phi(\epsilon) $. 
	From monotonicity and from Proposition~\ref{Prop:targetTh}, the target threshold is 
	\begin{align*}
	\epsilon^*\left (\udel_L,\udel_R\right ) 
	&\leq \epsilon^*\left (\underline \phi(\epsilon),\underline \phi(\epsilon)\right ) \\
	&\leq \epslr \\
	&< \epsilon.
	\end{align*}
	Thus, for every $ \epsilon>\epsilon^*_3$,  $ \epsilon>\esg{0}{0} $; hence, $ \esg{0}{0}\leq  \epsilon^*_3$. In view of \eqref{Eq:order_SG} and \eqref{Eq:epsSG1}, we complete the proof for \eqref{Eq:epsSG2}.
\end{proof}
The meaning of \eqref{Eq:epsSG2} is that for construction that satisfy the added conditions in Theorem~\ref{Th:eps*SG}, if one side starts from a termination, it does not help to start the other side from termination.
Table~\ref{Tbl:notations} and Figure~\ref{Fig:ThSummary} summarize the notations defined and results derived in Sections~\ref{Sec:SingleSB} and~\ref{Sec:Thrsholds}. 
\begin{table}
	\caption{\label{Tbl:notations}Notations used above}
	\begin{center}
		\begin{tabular}{c||p{8cm}||c}
			Notation 						& Meaning 																								& Reference				\\
			\hline
			\hline
			$ \Delta(\epsilon,\udel) $		& Helper outgoing DE values 																			& 	Definition~\ref{Def:Delta}		 			\\
			\hline
			$ \epsilon^*_1$				& Maximal channel parameter such that local decoding succeeds 	&\eqref{Eq:s1}\\
			\hline
			$ \epsilon^*_2 $				& Maximal channel parameter such that the outgoing erasure is always smaller than the incoming erasure 	&\eqref{Eq:s2}\\
			\hline
			$ \epsilon^*_3 $				& Maximal channel parameter such that a helper preserves zero incoming DE values 	&\eqref{Eq:s3}\\
			\hline
			$ \epsilon^*(\udel_L,\udel_R)$& Maximum channel parameter such that the target is successfully decoded given incoming DE values $ \udel_L $ and $ \udel_R $&
			Section~\ref{Sub:TargetTh} \\
			\hline
			$ \epsilon^*_4 = \epsilon^*(\underline 0,\underline 0)$				& Maximal channel parameter such that a target with full incoming information is successfully decoded 	&\eqref{Eq:s4}\\
			\hline
			$ \epslr $						& Threshold of the $ (l,r) $ LDPC ensemble														& \cite{RichUrb}\\
			\hline
			$ \dhat{\tau} $					& Input DE values into the target when semi-global starts from termination ($ \tau=0 $) and non-termination ($ \tau=1 $)																		&Definition~\ref{Def:del_hat}\\
			\hline
			$ \esg{\tau_L}{\tau_R} $				& Threshold of semi-global decoding 																				&\eqref{Eq:eps_SG}
		\end{tabular}
	\end{center}
\end{table}

\begin{figure}
	\begin{center}
		\begin{tikzpicture}
		\pgfmathsetmacro{\y}{15}
		\pgfmathsetmacro{\x}{5}
		\node(zero) [anchor=west,align=left]at (0,0) {$ 0 $};
		\node(epsL) [above = \y mm of zero.west,anchor=west,align=left] {$ \epsilon^*_1$};
		\node(tilde) [above = \y mm of epsL.west,anchor=west,align=left] {$ \epsilon^*_2$};
		\node(SG11) [above = \y mm of tilde.west,anchor=west,align=left] {$ \esg{1}{1}$};
		\node(eps01) [above = \y mm of SG11.west,anchor=west,align=left] {$ \epsilon^*_3$ ($=\esg{0}{0}=\esg{0}{1}$, Th.\ref{Th:eps*SG})};
		\node(eg) [above = \y mm of eps01.west,anchor=west,align=left] {$ \epsilon^*_G$};
		\node(eps00) [above = \y mm of eg.west,anchor=west,align=left] {$ \epsilon^*_4$};
		\node(one) [above = \y mm of eps00.west,anchor=west,align=left] {$ 1$};
		
		\node (st) [left = \x mm of zero.south] { };
		\node (end) at (st|-one.north) {};
		\draw [thick,|->,>=latex] (st)--(end);
		\node (e1) at (st|-epsL) {$ - $};\node (e2) at (st|-tilde) {$ - $};
		\node (e3) at (st|-SG11) {$ - $};\node (e4) at (st|-eps01) {$ - $};
		\node (e5) at (st|-eg) {$ - $};\node (e6) at (st|-eps00) {$ - $};
		
		\node (e11) [left = 3.5*\x mm of e1 ] { };\node (e00) at(e11|-st) { };
		\node (e33) [left = 3.5*\x mm of e3 ] { };\node (e44) [left = 3.5*\x mm of e4 ] { };
		\node (e55) [left = 3.5*\x mm of e5 ] { };		
		\draw [thick,decorate,decoration={brace,amplitude=5pt,raise=6pt}] (e00)--(e11.center) node [black,midway,xshift=-12mm,text width=2cm,align=center] {\footnotesize Local decoding};
		\draw [thick,decorate,decoration={brace,amplitude=5pt,raise=6pt}] (e1.center)--(e2.center) node [black,midway,xshift=-12mm,text width=2cm,align=center] {\footnotesize $ \dhat{1}=0 $};		
		
		\draw [thick,decorate,decoration={brace,amplitude=5pt,raise=6pt}] (e11.center)--(e33.center) node [black,midway,xshift=-12mm,text width=2cm,align=center] {\footnotesize Semi-global decoding not from termination};	
		\draw [thick,decorate,decoration={brace,amplitude=5pt,raise=6pt}] (e2.center)--(e4.center) node [black,midway,xshift=-12mm,text width=2cm,align=center] {\footnotesize 
			$ \dhat{0}=0$ $\dhat{1}>0 $};	
		
		\draw [thick,decorate,decoration={brace,amplitude=5pt,raise=6pt}] (e33.center)--(e44.center) node [black,midway,xshift=-15mm,text width=2cm,align=center] {\footnotesize Semi-global decoding from termination};	
		\draw [thick,decorate,decoration={brace,amplitude=5pt,raise=6pt}] (e4.center)--(e5.center) node [black,midway,xshift=-12mm,text width=2cm,align=center] {\footnotesize 
			$ \dhat{0}= \dhat{1}$ };	
		\draw [thick,decorate,decoration={brace,amplitude=5pt,raise=6pt}] (e44)--(e55.center) node [black,midway,xshift=-12mm,text width=2cm,align=center] {\footnotesize Global decoding};				
		\end{tikzpicture}
	\end{center}
	\caption{\label{Fig:ThSummary}Summary of the asymptotic results derived in Sections~\ref{Sec:SingleSB} and~\ref{Sec:Thrsholds}.}
\end{figure}
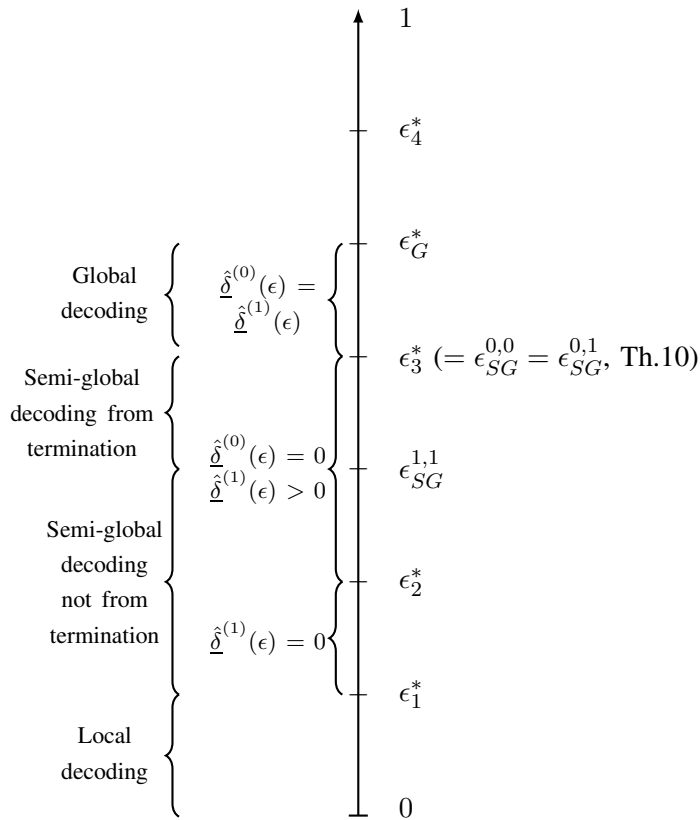

\section{Performance Over The Sub-Block Markov-Varying Channel}\label{Sec:SBMV}

As shown in Section~\ref{Sec:Thrsholds} (and previously reported in \cite{RamCass19}), starting semi-global decoding from a termination SB boosts the performance significantly, however, if the target SB is far away from termination SBs (i.e., the coupled chain is very long and the target is around the middle of it), then it is costly -- in terms of latency -- to start semi-global decoding from the endpoints of the chain.    

In some practically motivated channel models, like the SB-varying channel that we next study, we might get ``lucky" such that a non-termination helper SB is decoded successfully. Since the outgoing DE values from a successfully decoded SB equal $ 0 $ (see Proposition~\ref{Prop:del>0Gen}), then this SB acts like a termination SB. This happens, for example, if the SB erasure rate $ \epsilon $ is below the local threshold, as well as in other scenarios. 
To capture this favorable effect, we define \emph{pseudo-termination sequences}, whose qualitative definition is sequences of SBs (and their channel states) that output zero DE values in their SG density-evolution analysis.
Pseudo-termination sequences are a central component of the analysis in this section. 

To analyze SG decoding over channels with variability (and memory), in the following sub-sections we define three Markov chains: the first describes the channel (Section~\ref{Sub:MarkovChannel}), the second is simplified to fit the SB thresholds of the code ensemble (Section~\ref{Sub:BlockState}), and the third corresponds to the decoder states during semi-global decoding (Section~\ref{Sub:MarkovSG}).
\subsection{Channel Model}\label{Sub:MarkovChannel}

We consider a channel model, in which each SB suffers from a (possibly) different erasure rate. 
Furthermore, the channel state (erasure rate) between SBs forms a Markov chain. We call this channel the SB Markov-varying (SBMV) channel. 
The i.i.d. version of the SB varying channel was first introduced in \cite{McElStark84}, and was studied for SC-LDPCL codes in \cite{RamCass19}. 

Let $ \mathcal E =\{e_1,e_2,\ldots,e_{|\mathcal{E}|}\}$ be the possible channel states (erasure rates), and let $ \{E_m\}_{m=1}^M $ be a Markov chain describing the channel state of SBs $ m\in\{1,2,\ldots,M\} $, with the transition probabilities
\begin{align}\label{Eq:P}
\Pr\left (E_m =e_j \big| E_{m-1} =e_i\right ) = P_{i,j},\quad\; \forall \;2\leq m	 \leq M,
\end{align}
where $ P $ is a given $|\mathcal{E}|\times |\mathcal{E}|  $ transition matrix (non-negative elements with rows summing to $ 1 $).
We assume that this Markov chain has a unique stationary distribution $ \underline \nu = (\nu_1,\nu_2,\ldots,\nu_{|\mathcal{E}|}) $, and that all SBs of interest (i.e., target and helper SBs) are distributed according to $ \underline\nu $ (convergence to steady state irrespective of initial conditions lies on the assumption that these SBs are sufficiently far from the block boundaries). 
The \emph{expected erasure rate} of an SBMV channel is given by
\begin{align}\label{Eq:ExpEps}
\mathbb{E}\left[E_m\right ] = \sum_{i=1}^{|\mathcal{E}|}\nu_i e_i.
\end{align}
For example, the constant case, where the channel parameter is the same in all SBs and equals to some $ e\in[0,1] $, can be viewed as a SBMV channel with $ \mathcal E = \{e\} $, and trivial transition matrix and stationary distribution $ P=\nu=1 $. Another example is the SB i.i.d. channel \cite{RamCass18,McElStark84}, where each SB suffers from a channel parameter $ e_i\in\mathcal E $ with some probability $ p_i\in[0,1] $ where $ \sum_{i=0}^\mathcal{E} p_i=1 $, and the SBs' channel parameters are independent. In this case, for every $ i,j\in\{1,\ldots,\mathcal{E}\} $, $ P_{i,j} =p_j$, and $ \nu_i = p_i $.

\subsection{Sub-Block States}\label{Sub:BlockState}

When considering semi-global decoding over the SBMV channel, we assign to each SB a state according to the four intervals corresponding to the thresholds defined in \eqref{Eq:s1}--\eqref{Eq:s3}:
1) local decoding interval $ [0,\epsilon^*_1) $, for channel parameters in this interval, the SB is decodable locally; 2) error-reduction interval $ [\epsilon^*_1,\epsilon^*_2) $, where the inter-DE values between SBs decrease; 3) error-free-preservation interval  $ [\epsilon^*_2,\epsilon^*_3) $, where zero incoming DE values are preserved; 4) anti-termination interval  $ [\epsilon^*_3,1] $, where the outgoing DE values are arbitrarily high, regardless of the incoming DE values.
We therefore merge channel states that fall inside the same interval and define
\begin{align}\label{Eq:a_i}
a_i  = \{1\leq k \leq |\mathcal E| \colon e_k\in [\epsilon^*_{i-1},\epsilon^*_i) \}\;,\quad 1\leq i \leq 4,
\end{align}
where $ \epsilon_0^*=0 $ and $ \epsilon^*_4=1 $.

Let $ \{B_m\}_{m=1}^M $ be a Markov chain describing the state of SBs $ m\in\{1,2,\ldots,M\} $ with a state space $ \mathcal{S}=\{s_1,s_2,s_3,s_4\} $ that correspond to $ \{a_1,a_2,a_3,a_4\} $ in \eqref{Eq:a_i}, and a $ 4\times 4 $ transition matrix $ Q $. 
Recall that SBs in state $ s_1 $ are decodable locally, meaning that they output zero DE values, regardless of their incoming DE values. SBs in state $ s_2 $ are not decodable locally, but a sufficient number of consecutive $ s_2 $-blocks will output zero DE values. In general, the number of $s_2  $ SBs needed for this procedure, which we denote by $ q $, depends on the exact channel parameter in the interval $[\epsilon^*1,\epsilon^*2)$. However, for simplicity of the results, we take a worst-case approach and set 
\begin{align}\label{Eq:qmax}
q = \max\{q(e_k)\colon k\in a_2\},
\end{align}
where $ q(\cdot) $ is defined in \eqref{Eq:q}.

In addition, SBs in state $s_3 $ preserve zero incoming DE values, meaning that if the incoming DE values are zero, then so are the outgoing DE values. 
As another worst-case assumption, we assume that if in state $ s_3 $ the incoming DE values are not zero, then the outgoing DE values are arbitrarily high. 
Finally, SBs in state $s_4 $ output high DE values even if the incoming DE values are zero. If the semi-global decoder encounters a SB in state $ s_4 $, then all of the information gathered from previously decoded helper SBs becomes irrelevant for the rest of the decoding process.
	
Note that in general, SBs in states $ s_3 $ and $ s_4 $ could be helpful, i.e., produce smaller or equal outgoing DE values even if the incoming DE values are non-zero, however, it is not guaranteed (unlike $ s_1 $ or $ s_2 $ SBs). Thus, the real performance may actually be better, and the results we derive in the following sub-sections are lower bounds on the real performance. 

For every $ i\in \{1,2,3,4\} $, let $ \mu_i = \sum_{k\in a_i} \nu_{k} $, with the convention that an empty sum equals zero. If indeed $ a_i=\emptyset $ for some $ i\in\{1,2,3,4\} $, then we set $ Q(i,i)=1 $ and $ Q(i,j)=Q(j,i)=0 $ for every $ j\neq i $. Else, the transition probabilities are given by
\begin{align}\label{Eq:Q}
\begin{split}
Q_{i,j}
&\triangleq\Pr\left (B_m=s_j \big|B_{m-1}=s_i\right ) \\
&= \Pr\left (E_m\in a_j \big|E_{m-1}\in a_i\right ) \\
&= \sum_{i'\in a_i} \Pr\left ( E_{m-1} = e_{i'}\big|E_{m-1}\in a_i\right )\Pr\left (E_m\in a_j \big|E_{m-1} = e_{i'}\right ) \\
&= \sum_{i'\in a_i} \Pr\left ( E_{m-1} = e_{i'}\big|E_{m-1}\in a_i\right )\sum_{j'\in a_j} \Pr\left (E_m = e_{j'} \big|E_{m-1} = e_{i'}\right )\\
&= \frac{1}{\mu_i}\sum_{i'\in a_i}\sum_{j'\in a_j}  \nu_{i'} P_{i',j'},
\end{split}
\end{align}
where $ i,j\in\{1,2,3,4\} $, $s_i,s_j\in\mathcal S  $, and $ P $ is a given in \eqref{Eq:P}.
Figure~\ref{Fig:MarkocChannel}(a) and Figure~\ref{Fig:MarkocChannel}(b) illustrate the channel-state and SB-state diagrams, respectively, for the case of $ |\mathcal E|=4 $ and $ a_1=\{1,2\},a_2=\{3\},a_3=\{4\},a_4=\emptyset $.
\begin{lemma}\label{Lemma:StatB}
The stationary distribution of  $ \{B_m\}_{m=1}^M $  is given by $ \underline \mu = (\mu_1,\mu_2,\mu_3,\mu_4) $, where
$ \mu_i = \sum_{k\in a_i} \nu_{k}  $, $ a_i  $ is given in \eqref{Eq:a_i}, and $ \underline \nu $ is the stationary distribution of  $ \{E_m\}_{m=1}^M $.
\end{lemma}

\begin{proof}
	For every $ j\in\{1,2,3,4\} $
\begin{subequations}
	\begin{align}
	\label{Eq:mu_ja}
	\mu_j 
	&= \sum_{j'\in a_j} \nu_{j'}\\
	\label{Eq:mu_jb}
	&=\sum_{j'\in a_j} \sum_{i=1}^{|\mathcal{E}|} \nu_i P_{i,j'}\\
	\label{Eq:mu_jc}
	&=\sum_{i=1}^{|\mathcal{E}|}\sum_{j'\in a_j}  \nu_i   P_{i,j'}\\
	\label{Eq:mu_jd}
	&=\sum_{i=1}^{4} \sum_{i'\in a_i}   \sum_{j'\in a_j} \nu_{i'}  P_{i',j'}\\
	\label{Eq:mu_je}
	&=\sum_{i=1}^{4} \mu_i Q_{i,j},
	\end{align}
\end{subequations}	
	where \eqref{Eq:mu_ja} is the definition of $ \mu_j $, \eqref{Eq:mu_jb} follows since $ \underline \nu $ is left eigenvector of $ P $ with a unit eigenvalue, \eqref{Eq:mu_jc} is a change in summation order, \eqref{Eq:mu_jd} is summing over the partition of $ \{1,2,\ldots,|\mathcal E|\} $ into the sets $ \{a_i\}_{i=1}^4 $, \eqref{Eq:mu_je} is due to \eqref{Eq:Q}.
\end{proof}

%

\begin{figure}
	\begin{center}
		\begin{tikzpicture}[>=latex,el/.style = {inner sep=1pt}]
		\tikzstyle{state}=[circle,draw,thick,minimum size=10mm]
		\pgfmathsetmacro{\x}{28}
		\pgfmathsetmacro{\y}{10}
		\pgfmathsetmacro{\z}{13}
		
		\node (s1) [state] {\textcolor{black}{$ e_1 $}};\node (s2) [state,right = \x mm of s1] {\textcolor{black}{$ e_2 $}};
		\node (s3) [state,below = \x mm of s1] {\textcolor{black}{$ e_3 $}};\node (s4) [state,right = \x mm of s3] {\textcolor{black}{$ e_4 $}};
%
%
		\path 
		(s1) edge [loop ,min distance=\y mm,in=135,out=90,->,thick] node [above,el]{\footnotesize $ P_{1,1} $} (s1)
		(s2.north) edge [loop ,min distance=\y mm,in=45,out=90,->,thick] node[above,el]{\footnotesize $ P_{2,2} $} (s2.north east)
		(s3.south) edge [loop ,min distance=\y mm,in=225,out=270,->,thick] node [below,el] {\footnotesize $ P_{3,3} $} (s3.south west)
		(s4.south) edge [loop ,min distance=\y mm,in=315,out=270,->,thick] node [below,el] {\footnotesize $ P_{4,4} $} (s4.south east)
		
		(s1) edge [->,thick,bend left=\z] node [above,el] {\footnotesize $ P_{1,2} $} (s2)
		(s1) edge [<-,thick,bend right=\z] node [above,el] {\footnotesize $ P_{2,1} $} (s2)
		
		(s3) edge [->,thick,bend left=\z] node [below,el] {\footnotesize $ P_{3,4} $} (s4)
		(s3) edge [<-,thick,bend right=\z] node (x) [below,el] {\footnotesize $ P_{4,3} $} (s4)
		
		(s1) edge [->,thick,bend right=\z] node [left,el] {\footnotesize $ P_{1,3} $} (s3)
		(s1) edge [<-,thick,bend left=\z] node [left,el] {\footnotesize $ P_{3,1} $} (s3)
		
		(s2) edge [->,thick,bend right=\z] node [right,el] {\footnotesize $ P_{2,4} $} (s4)
		(s2) edge [<-,thick,bend left=\z] node [right,el] {\footnotesize $ P_{4,2} $} (s4)
		
		(s1) edge [->,thick,bend right=\z] node [right,el,pos=0.7] {\footnotesize $ P_{1,4} $} (s4)
		(s1) edge [<-,thick,bend left=\z] node [right,el,pos=0.2] {\footnotesize $ P_{4,1} $} (s4)
		
		(s2) edge [->,thick,bend right=\z] node [right,el,pos=0.8] {\footnotesize $ P_{2,3} $} (s3)
		(s2) edge [<-,thick,bend left=\z] node [right,el,pos=0.3] {\footnotesize $ P_{3,2} $} (s3)
		;		
		
		\node (a) [below = 1cm of x] {(a)};\node (b) [right = 3*\x mm of a] {(b)}; \node (c) at(s1-|b) {};

		\node (s1) [state,left=0.5*\x mm of c] {\textcolor{black}{$ e_1 $}};\node (s2) [state,right = \x mm of s1] {\textcolor{black}{$ e_2 $}};
		\node (s3) [state,below = \x mm of s1] {\textcolor{black}{$ e_3 $}};\node (s4) [state,right = \x mm of s3] {\textcolor{black}{$ e_4$}};	
		
		\node (s12) [thick,rectangle,rounded corners, fit=(s1)(s2),blue,draw] {}; 
		\node (x) [above =1mm of s12] {};\node  [right=1mm of x] {$ s_1 $};
		\node (s33) [thick,rectangle,rounded corners, fit=(s3),red,draw] {}; \node [left=1mm of s33] {$ s_2 $};
		\node (s44) [thick,rectangle,rounded corners, fit=(s4),cyan,draw]{}; \node [right=1mm of s44] {$ s_3 $};
			
		\path 
		(s12) edge [loop ,min distance=\y mm,in=135,out=90,->,thick] node [above,el]{\footnotesize $ Q_{1,1} $} (s12)
		(s33.south) edge [loop ,min distance=\y mm,in=225,out=270,->,thick] node [above,el]{\footnotesize $ Q_{2,2} $} (s33.south west)
		(s44.south) edge [loop ,min distance=\y mm,in=315,out=270,->,thick]  node [above,el]{\footnotesize $ Q_{3,3} $} (s44.south east)
		
		(s33) edge [<-,thick,bend right=\z] node (x) [below,el] {\footnotesize $ Q_{3,2} $} (s44)
		(s33) edge [->,thick,bend left=\z] node [below,el] {\footnotesize $ Q_{2,3} $} (s44)
		(s12) edge [->,thick,bend right=\z] node [left,el] {\footnotesize $ Q_{1,2} $} (s33)
		(s12) edge [<-,thick,bend left=\z] node [right,el] {\footnotesize $ Q_{2,1} $} (s33)
		(s12) edge [->,thick,bend right=\z] node [right,el,pos=0.7] {\footnotesize $ Q_{1,3} $} (s44)
		(s12) edge [<-,thick,bend left=\z] node [right,el,pos=0.2] {\footnotesize $ Q_{3,2} $} (s44)
		;
		\end{tikzpicture}
	\end{center}
	\caption{\label{Fig:MarkocChannel}(a): a state diagram representing an SBMV channel with state space $ \mathcal{E}=\{e_1,e_2,e_3,e_4\} $, and transition matrix $ P $; (b) a state diagram representing the block states with $ a_1 = \{1,2\},a_2 = \{3\},a_3=\{4\},a_4=\emptyset $, and transition matrix $ Q$ given in \eqref{Eq:Q}. For $ i\in\{1,2,3,4\} $, each state $s_i$ is marked in a color corresponding to the curve in Figure~\ref{Fig:delta_O} that depicts its erasure-transfer behavior.}
\end{figure}
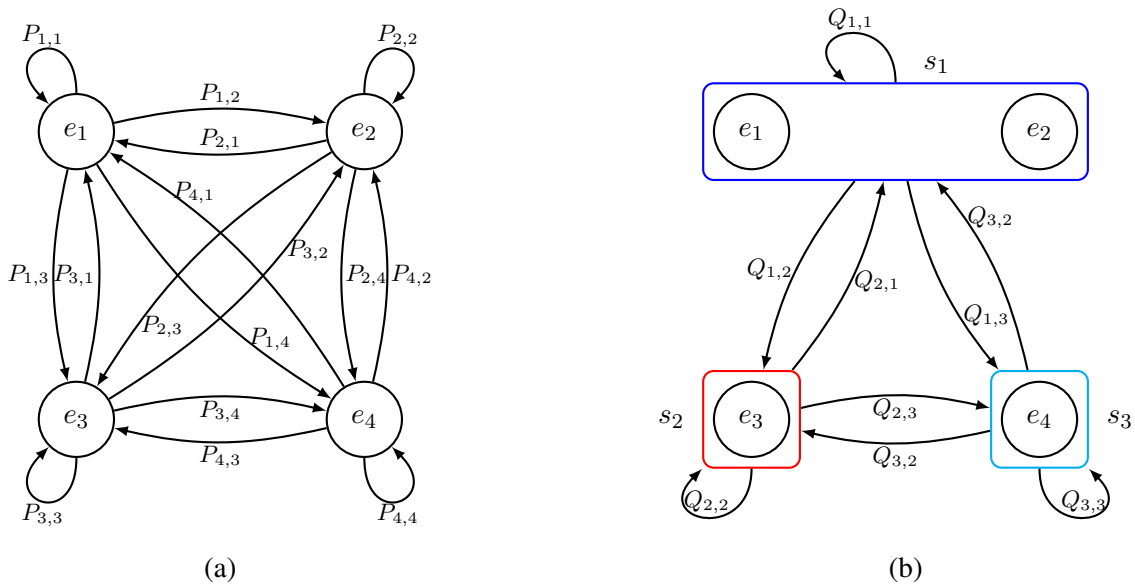

\subsection{Pseudo-Termination Sequences}\label{Sub:MarkovSG}
Consider a sequence of $ c $ helper SBs
$ (b_1,b_2,\ldots,b_c)$,
where for every $ i\in \{1,2,\ldots,c\} $, $ b_i\in \mathcal S $, and suppose they are accessed, during semi-global decoding, in a descending order (first $ b_c $, and last $ b_1 $). 
Recall that for SBs in state $ s_2 $, \eqref{Eq:s2} holds, which induces an integer $ q\geq 2 $ defined in \eqref{Eq:qmax} as the length of a sequence of consecutive SBs in state $s_2$ that suffices to produce zero outgoing DE values. 
Let $ \Delta_{c,q}(b_1,b_2,\ldots,b_c) $ denote the outgoing DE values from SB $ b_1 $ after decoding the sequence $ b_c,\ldots,b_2,b_1 $. When $ \Delta_{c,q}(b_1,b_2,\ldots,b_c) $ is all zero, we say that $b_1,b_2,\ldots,b_c$ is a pseudo-termination sequence.

\begin{definition}\label{Def:pseudo-term}
	For integers $ c,q $, the set of all pseudo-termination sequences is defined by
	\begin{align*}
	\mathcal T(c,q) = \left \{ (b_1,b_2,\ldots,b_c)\in \mathcal S^c \colon \Delta_{c,q}(b_1,b_2,\ldots,b_c)=\underline 0 \right \}.
	\end{align*}
\end{definition}

\begin{example}\label{Ex:T}
	Let $ c=3,q=2 $. Then,
	\begin{align*}
	\begin{array}{lllllllll}
	\mathcal{T}(3,2)
	&=&\{(s_1,b_2,b_3),\forall b_2,b_3\in \mathcal S\}
	&\cup&\{(s_2,s_1,b_3),\forall b_3\in \mathcal S\}	
	&\cup& \{(s_2,s_2,b_3),\forall b_3\in \mathcal S\}
	&\cup& \{(s_2,s_3,s_1)\}\\
	&\cup& \{(s_3,s_2,s_1)\}
	&\cup& \{(s_3,s_1,b_3),\forall b_3\in \mathcal S\}
	&\cup& \{(s_3,s_2,s_2)\}
	&\cup& \{(s_3,s_3,s_1)\},
	\end{array}
	\end{align*}
	with a total of $ |\mathcal{T}(3,2)|=32 $ pseudo-termination sequences, out of $ 4^3=64 $ possible sequences.
\end{example}

In what follows we answer the following question: under the SBMV channel model, what is the probability that a sequence of $ c $ helper SBs will be a pseudo termination? 
The key to answering this question is defining a third Markov chain evolving backward from $b_1$ (decoded last) to $b_c$ (decoded first).  
Clearly, if $ b_1 $ is in state $ s_1 $, then the sequence is a pseudo-termination sequence regardless of the states of $ b_2,\ldots ,b_c $. Similarly, if $ b_1 $ is in state $ s_4 $, then the sequence is \emph{not} a pseudo-termination sequence, regardless to the states of $ b_2,\ldots, b_c $. 
If $b_1$ is in state $ s_2 $, than we have to check $ b_2 $: if it is in state $ s_1 $ or it is in state $ s_2 $ and $ q=2 $, then the sequence is a pseudo-termination sequence; if $ b_1 $ is in state $ s_3 $, then we check if $b_2,\ldots,b_c$ is pseudo termination; and so forth, until we reach $ c $ SBs. 
The logic of determining pseudo-termination sequences is illustrated in Example~\ref{Ex:T}, and specified in full precision by the Markov chain presented next. 
		
Define a Markov chain with $ q+2 $ states $\tilde{\mathcal S}_q=  \{\tilde s_1,\tilde s_2^{(1)}, ,\tilde s_2^{(2)}, \ldots,\tilde s_2^{(q-1)},\tilde s_3,\tilde s_4  \}$. 
This Markov chain will be used to analyze the semi-global decoder in reverse order as described above.
The operational meanings of the states in $ \tilde{\mathcal S}_q $ are: 1) $ \tilde s_1 $ is the pseudo-termination state, which is an absorbing state because if reached, pseudo-termination is guaranteed irrespective of the subsequent SB states; 2) $ \tilde s_2^{(1)},\tilde s_2^{(2)},\ldots,\tilde s_2^{(q-1)} $ are error-reduction states, and they represent a sequence of $ q-1 $ consecutive SBs in state $ s_2 $; 3) $ \tilde s_3 $ is the reset state, as it resets the error-reduction stage; 4) $ \tilde s_4 $ is the anti-termination state, which like $\tilde{s}_1$ is also an absorbing state because if reached, then the corresponding sequence is not pseudo termination regardless of the subsequent SB states.
The transition probabilities between the states in $ \tilde{\mathcal S}_q $ are given by the following matrix
\begin{align}\label{Eq:Qq}
		Q_q = \left (
		\begin{array}{c|c|c|c|c}
		1				&	0	&	0					&	0	&	0	\\
		\hline	
		Q_{2,1}			&	0	& \begin{array}{ccc} &&\\&Q_{2,2}\cdot I_{q-2}&\\&&\\\end{array} 	&Q_{2,3}&Q_{2,4}\\
		\hline
		Q_{2,1}+Q_{2,2}	&	0	&	0					&Q_{2,3}&Q_{2,4}\\
		\hline
		Q_{3,1}			&Q_{3,2}& 	0 					&Q_{3,3}&Q_{3,4}\\
		\hline
		0				&	0	&	0					&	0	&	1
		\end{array}
		\right ).
\end{align}
where $ I_{q-2} $ is the $ (q-2)\times (q-2) $ identity matrix. The initial distribution of this Markov chain (i.e., the distribution of the first SB in the sequence) is given by
$	\underline v =  (\mu_1,\mu_2,	0,0,\ldots,0,\mu_3,\mu_4 )$, where 	$ \underline\mu=(\mu_1,\mu_2,\mu_3,\mu_4) $ is the stationary distribution of the SB-state Markov chain given in Lemma~\ref{Lemma:StatB}. Figure~\ref{Fig:Induced} illustrates the state diagram of this induced Markov chain.
		
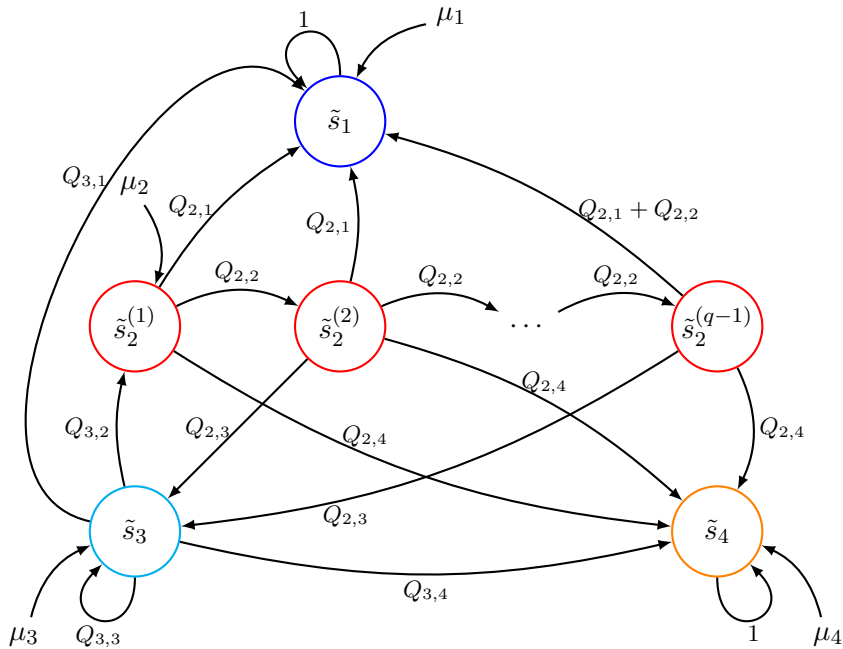
\begin{figure}
		\begin{center}
			\begin{tikzpicture}[>=latex,el/.style = {inner sep=2pt}]
			\tikzstyle{state}=[circle,draw,thick,minimum size=12mm,inner sep=1pt]
			\pgfmathsetmacro{\x}{15}
			\pgfmathsetmacro{\y}{10}
			\pgfmathsetmacro{\z}{13}
			
			\node (s1) [state,blue] {\textcolor{black}{$ \tilde s_1 $}};\node (s22) [state,below = \x mm of s1,red] {\textcolor{black}{$ \tilde s_2^{(2)} $}};
			\node (s21) [state,left = \x mm of s22,red] {\textcolor{black}{$ \tilde s_2^{(1)} $}};
			\node (dots) [right = \x mm of s22] {$ \cdots $};\node (s2qm1) [state,right = \x mm of dots,red] {\textcolor{black}{$ \tilde s_2^{(q-1)} $}};
			\node (s3) [cyan,state,below = \x mm of s21] {\textcolor{black}{$ \tilde s_3 $}};\node (s4) [orange,state,below = \x mm of s2qm1] {\textcolor{black}{$ \tilde s_4 $}};
			
			\node (n1) [above right=\y mm of s1] {$ \mu_1 $}; \node (n2) [above =\y mm of s21] {$ \mu_2 $};
			\node (n3) [below left=\y mm of s3] {$ \mu_3 $}; \node (n4) [below right=\y mm of s4] {$ \mu_4 $};
			
			\path 
			(s1) edge [loop ,min distance=\y mm,in=135,out=90,->,thick] node [above,el]{\footnotesize $ 1$} (s1)
			(s3.south) edge [loop ,min distance=\y mm,in=225,out=270,->,thick] node [below,el] {\footnotesize $ Q_{3,3} $} (s3.south west)
			(s4.south) edge [loop ,min distance=\y mm,in=315,out=270,->,thick] node [below,el] {\footnotesize $ 1 $} (s4.south east)

			(s1) edge [<-,thick,bend right=\z] node [left,el] {\footnotesize $ Q_{2,1} $} (s21)
			(s1) edge [<-,thick,bend left=\z] node [left,el] {\footnotesize $ Q_{2,1} $} (s22)
			(s1) edge [<-,thick,bend left=\z] node [right,el,pos=0.6] {\footnotesize $ Q_{2,1}+ Q_{2,2} $} (s2qm1)
			
			(s21) edge [->,thick,bend left=2*\z] node [above,el] {\footnotesize $ Q_{2,2} $} (s22)
			(s22) edge [->,thick,bend left=2*\z] node [above,el] {\footnotesize $ Q_{2,2} $} (dots)
			(dots) edge [->,thick,bend left=2*\z] node [above,el] {\footnotesize $ Q_{2,2} $} (s2qm1)
			
			(s3) edge [->,thick,bend right=\z] node [below,el] {\footnotesize $ Q_{3,4} $} (s4)
			(s3) edge [->,thick,bend left=\z] node [left,el] {\footnotesize $ Q_{3,2} $} (s21)
			(s3) edge [<-,thick] node [left,el] {\footnotesize $ Q_{2,3} $} (s22)
			(s3) edge [<-,thick,bend right=\z] node [below,el,pos=0.3] {\footnotesize $ Q_{2,3} $} (s2qm1)
			(s3) edge [->,thick,bend left=8*\z] node [left,el,pos=0.6] {\footnotesize $ Q_{3,1} $} (s1)
			
			(s21) edge [->,thick,bend right=\z] node [above,el,pos=0.4] {\footnotesize $ Q_{2,4} $} (s4)
			(s22) edge [->,thick,bend left=\z] node [above,el] {\footnotesize $ Q_{2,4} $} (s4)
			(s2qm1) edge [->,thick,bend left=2*\z] node [right,el] {\footnotesize $ Q_{2,4} $} (s4)
			(n1) edge [->,thick,bend right=2*\z] (s1) (n2) edge [->,thick,bend left=2*\z] (s21) 
			(n3) edge [->,thick,bend left=2*\z] (s3) (n4) edge [->,thick,bend right=2*\z] (s4)
			;
			\end{tikzpicture}
		\end{center}
		\caption{\label{Fig:Induced}The state diagram representing the induced Markov chain  $ \{ X_i\}_{i=1}^c   $ of the SG decoder.}
	\end{figure}
	
\begin{theorem}\label{Th:SBMV}
	Let $ c,q $ be integers such that $ q\geq 2 $, and let $ \underline B =(B_1,B_2,\ldots,B_c) $ be the SB-state Markov chain described in Section~\ref{Sub:BlockState}. Then, 
	\begin{align*}
	\Pr \left (\underline B\in \mathcal T(c,q) \right ) =  \underline v Q_q^{c-1} \underline u^T,
	\end{align*}
	where $ \underline v,\underline u $ are $ 1\times (q+2) $ vectors given by
	\begin{align*}
	{\arraycolsep=1pt\begin{array}{cclccccccr}
	\underline v &=&  (\mu_1,	&\mu_2,	&0,&0,&\ldots,&0,&\mu_3,&\mu_4 )\\
	\underline u &=&  (\; 1,		&0,		&0,&0,&\ldots,&0,&0,	&0\;)
	\end{array}},
	\end{align*}
	$ \underline\mu=(\mu_1,\mu_2,\mu_3,\mu_4) $ is the stationary distribution given in Lemma~\ref{Lemma:StatB}, and $ Q_q $ is given in \eqref{Eq:Qq}.

\end{theorem}

\begin{proof}
	Let  $ \{ X_i\}_{i=1}^c $ be a Markov chain over the state space $\tilde{\mathcal S}_q$, with $ Q_q $ as a transition matrix, and $ \underline v $ as the initial probability (see Figure~\ref{Fig:Induced} for the corresponding state diagram). 
	By marginalizing over the state of the first SB and expanding by \eqref{Eq:CahpKol},
	\begin{align*}
	\Pr \left (\underline b\in \mathcal T(c,q) \right )
	&\triangleq  \Pr\left ( X_c = \tilde s_1 \right )\\
	&=  \mu_1\Pr\left ( X_c = \tilde s_1 \big | X_1=\tilde s_1 \right )\\
	&+  \mu_2\Pr\left (  X_c = \tilde s_1 \big | X_1=\tilde s_2^{(1)} \right )\\
	&+  \mu_3\Pr\left ( X_c = \tilde s_1 \big | X_1=\tilde s_3 \right )\\
	&+  \mu_4\Pr\left ( X_c = \tilde s_1 \big | X_1=\tilde s_4 \right )\\
	&=  \underline v Q_q^{c-1} \underline u^T.
	\end{align*}
\end{proof}

\begin{corollary}\label{Coro:unif}
	Let $ c,q $ be integers such that $ q\geq 2 $. Then,
	\begin{align*}
	|\mathcal T (c,q)| =  \underline b A_q^{c-1} \underline u^T,
	\end{align*}
	where 
	\begin{align*}
	\begin{array}{llr}
	\underline b &=&\left (1,1,0,0,\ldots,0,1,1\right )\\
	\underline u &=& \left (1,0,0,0,\ldots,0,0,0\right )
	\end{array},
	\end{align*}
	and
	\begin{align*}
	A_q =\left (
	\begin{array}{c|c|c|c|c}
	4				&	0	&	0					&	0	&	0	\\
	\hline	
	1			&	0	& \begin{array}{ccc} &&\\& I_{q-2}&\\&&\\\end{array} 	&1&1\\
	\hline
	2	&	0	&	0					&1&1\\
	\hline
	1			&1& 	0 					&1&1\\
	\hline
	0				&	0	&	0					&	0	&	4
	\end{array}
	\right ).
	\end{align*}
\end{corollary}

\begin{proof}
	Follows by applying Theorem~\ref{Th:SBMV} with an i.i.d. SBMV channel. i.e, $ Q_{i,j}=\mu_i=0.25 $ for every $ i,j\in \{1,2,3,4\} $, and by the observation that for the i.i.d. case 
	\begin{align*}
	\Pr \left (\underline b\in \mathcal T(c,q) \right )= \frac{	|\mathcal T (c,q)| }{|\mathcal S^c|}=\frac{	|\mathcal T (c,q)| }{4^c}.
	\end{align*}
	
\end{proof}

\begin{example}\label{Ex:unif}
	For $ c=3 $ and $ q=2 $ we have 
	\begin{align*}
	|\mathcal T (3,2)|
	&= \begin{pmatrix}
	1&1&1&1
	\end{pmatrix}
	\begin{pmatrix}
	4&0&0&0\\
	2&0&1&1\\
	1&1&1&1\\
	0&0&0&4
	\end{pmatrix}^2 
	\begin{pmatrix}
	1\\0\\0\\0
	\end{pmatrix}\\
	&=\begin{pmatrix}
	1&1&1&1
	\end{pmatrix}
	\begin{pmatrix}
	16&0&0&0\\
	9 &1&1&5\\
	7 &1&2&6\\
	0 &0&0&16
	\end{pmatrix}
	\begin{pmatrix}
	1\\0\\0\\0
	\end{pmatrix}\\
	&= 32,
	\end{align*}
	which agrees with Example~\ref{Ex:T}.
\end{example}

\subsection{Semi-Global Decoding Performance}\label{Sub:SBMVPerf}
	
We now state lower bounds on the probability of semi-global-decoding success over the SBMV channel for two schemes: 1) one-sided: all of the $ d $ helper SBs are on one side of the target; 2) two-sided: $ d/2 $ (for even $ d $) helper SBs are on each side of the target. Semi-global-decoding success is defined as the event that the target SB is decoded successfully (in DE terms) following semi-global decoding.  

\begin{proposition}\label{Prop:1SideSBMV}
	Let $ d,q $ be integers such that $ q\geq 2 $. Then, the success probability of the one-sided semi-global decoding over the SBMV channel with a transition matrix $ P $ is lower bounded by
	\begin{subequations}
	\begin{align}\label{Eq:Suc1RProb}
	p_{\text{R}}(d) \geq \underline v Q_q^{d} \underline u^T,
	\end{align} 
	if the helper SBs are on the right and 
	\begin{align}\label{Eq:Suc1LProb}
	p_{\text{L}}(d) \geq \underline v \hat Q_q^{d} \underline u^T,
	\end{align} 
	\end{subequations}
	where $ \underline v,\underline u,Q_q $ are the same as in Theorem~\ref{Th:SBMV}, and $ \hat Q_q $ is constructed as in \eqref{Eq:Qq} with the substitution
	\begin{align*}
	\hat Q_{i,j} = \frac{\mu_j}{\mu_i}Q_{j,i},\quad 1\leq i,j\leq 4.
	\end{align*}.
\end{proposition} 

\begin{proof}
In view of Proposition~\ref{Prop:del>0Gen}, semi-global decoding with all $ d $ helper SBs on one side succeeds, if and only if the target outputs zero DE values. The proof of \eqref{Eq:Suc1RProb} follows by applying Theorem~\ref{Th:SBMV} with $ d+1 $ SBs ($ 1 $ target $+d$ helper SBs). The proof of \eqref{Eq:Suc1LProb} follows similarly combined with the observation that we should consider the reverse Markov chain, i.e., a Markov chain with $ \hat Q $ as the transition matrix (see \eqref{Eq:ReverseMarkov}).
\end{proof}

\begin{proposition}\label{Prop:2SideSBMV}
	Let $ d,q $ be integers such that $ q\geq 2 $ and $ d $ is even. Then, the success probability of the two-sided semi-global decoding over the SBMV channel with a transition matrix $ P $ is lower bounded by
	\begin{align}\label{Eq:Suc2Prob}
	p_{\text{2}}(d) \geq \underline v_2 \left (Q_q^{d/2} \otimes \hat Q_q^{d/2}\right ) \underline u_2^T,
	\end{align} 
	where $ \otimes $ is the Kronecker product, $ Q_q, \hat Q_q  $ are the same as in Proposition~\ref{Prop:1SideSBMV}, and $ \underline v_2=(v_{2,1},v_{2,2},\ldots,v_{2,(q+2)^2}),\;\underline u_2=(u_{2,1},u_{2,2},\ldots,u_{2,(q+2)^2})$ are given by
	\begin{align*}
	{\arraycolsep=1pt\begin{array}{lll}
	v_{2,j+(q+2)(i-1)}&=&\left \{
	{\arraycolsep=3pt\begin{array}{ll}
	\mu_1 & j=i=1\\
	\mu_2 & j=i=2\\
	\mu_3 & j=i=q+1\\
	\mu_4 & j=i=q+2\\
	0 & \text{otherwise} 
	\end{array}}\right. \\[12mm]
	u_{2,j+(q+2)(i-1)}&=&\left \{
	{\arraycolsep=3pt\begin{array}{ll}
	1 & i=1\text{ or } j=1\\
	0 & \text{otherwise} 
	\end{array}}\right .
	\end{array},\qquad 1\leq i,j\leq q+2}.
	\end{align*}
\end{proposition}

\begin{proof}
Consider SBs $ {\bar m-d/2},\ldots,{\bar m-1},{\bar m},{\bar m+1},\ldots,{\bar m+d/2} $, i.e., $\bar  m $ is the index of the target SB. Let $ \{B_m\}_{m=1}^M $ be the Markov chain describing the state of each SB. For every $ w\in\{0,1,\ldots,d/2\} $, let $  B_{R,w}\triangleq B_{\bar m+w} $ and  $  B_{L,w}\triangleq B_{\bar m-w} $. Then, $\{B_{L,w}\}_{w=0}^{d/2} $ and $\{B_{R,w}\}_{w=0}^{d/2} $ are two Markov chains with transition matrices $ Q $ and $ \hat Q $, respectively.
The states of the right and left helper SBs correspond to $ \{B_{R,w}\}_{w=1}^{d/2} $ and  $ \{B_{L,w}\}_{w=1}^{d/2} $, respectively, and the target-SB state is $ B_{L,0}=B_{R,0}=B_{\bar m} $.

Similar to the SG-decoder Markov chain in Figure~\ref{Fig:Induced}, we define the two-dimensional Markov chain $ \{\underline X_{w}=\left (X_{L,w},X_{R,w}\right )\}_{w=0}^{d/2} $, with $ (q+1)^2 $ states $ \tilde{ \mathcal S}_q\times\tilde{ \mathcal S}_q  $, the transition matrix 
$ Q_q\otimes \hat Q_q $, and an initial distribution 
\begin{align*}
\Pr\left (\underline X_{0}=(x_L,x_R)\right ) = \left \{
\begin{array}{ll}
\mu_i & x_L=x_R=\tilde s_i,\; i=1,3,4\\
\mu_2 & x_L=x_R=\tilde s^{(1)}_2\\
0 & \text{otherwise}
\end{array}
\right ..
\end{align*}
Success in the two-sided semi-global decoding with even $ d $ helper SBs occurs if and only if $  X_{L,d/2}=\tilde s_1 $ or  $  X_{R,d/2}=\tilde s_1 $, which in view of \eqref{Eq:Suc2Prob} and the fact that $ \left (Q_q \otimes \hat Q_q\right )^{d/2} =\left (Q_q^{d/2} \otimes \hat Q_q^{d/2}\right )  $ completes the proof.
\end{proof}

\begin{example}\label{Ex:GEC}
	Consider an SBMV channel with $ \mathcal E=\{e_1=0.33,e_2=0.42\} $, and consider two transition matrices 
	\begin{align}\label{Eq:P1P2}
	P_1=\begin{pmatrix}
	\alpha& 1-\alpha\\
	1-\alpha& \alpha
	\end{pmatrix},\quad
	P_2=\begin{pmatrix}
	a& 1-a\\
	a& 1-a
	\end{pmatrix},
	\end{align}
	for some parameters $ 0<\alpha,a<1 $.
	Note that $ P_1 $ represents a SB Gilbert-Elliot channel, while $ P_2 $ represents a SB i.i.d. channel (like in \cite{RamCass19}) where each SB's channel parameter is independent of the other SBs, and for each SB $ m $, $ \Pr(E_m=e_1) = 1-\Pr(E_m=e_2)=a $. 
	In order to compare semi-global-decoding performance over these channels, we should equalize their expected erasure rates (such that the comparison is fair). In view of \eqref{Eq:ExpEps}, we set $ \alpha$ and $a $ such that the stationary distributions of both Markov chains ($ P_1 $ and $ P_2 $) coincide. 
	From symmetry, for every $ \alpha\in[0,1] $, the stationary distribution corresponding to $ P_1 $ is the uniform distribution $ (0.5,0.5) $, thus we set $ a=0.5 $.
		
	Assume we use the $ (l=3,r=6,t=1) $ SC-LDPCL protograph from Example~\ref{Ex:361_SCLDPCL} over these channels. Recall that the thresholds induced by this protograph are given by (see Example~\ref{Ex:thresholds}) $ \epsilon^*_1=0.2,\epsilon^*_2=0.3719,\epsilon^*_3=0.4297 $. We calculate $ q(0.33)=3 $ in \eqref{Eq:q}. In addition,
	\[
	a_1=\emptyset,\,a_2=\{1\},\,a_3=\{2\},a_4=\emptyset,\qquad 
	\underline \mu = (0,0.5,0.5,0),
	\]
	and
	\[
	Q_1=\begin{pmatrix}
	1	&	0	&	0	&	0\\
	0	&\alpha	&1-\alpha&	0\\
	0	&1-\alpha&\alpha	&	0\\
	0	&	0	&	0	&	1\\
	\end{pmatrix},\quad
	Q_2=\begin{pmatrix}
	1	&	0	&	0	&	0\\
	0	&	0.5	&0.5	&	0\\
	0	& 	0.5	&0.5	&	0\\
	0	&	0	&	0	&	1\\
	\end{pmatrix},
	\]
	corresponding to $ P_1,P_2 $, respectively.
	
	Figure~\ref{Fig:GEC} plots the decoding success lower bound of \eqref{Eq:Suc1RProb} and \eqref{Eq:Suc2Prob} as a function of $ d $ for the two transition matrices $ P_1 ,P_2$ in \eqref{Eq:P1P2}. For the non-i.i.d. channel $ P_1 $, we calculated the success-probability lower bound for $ \alpha=0.9 $ (blue) and $ \alpha=0.1 $ (green). The former is a model in which the channel parameter tends to stay in the same state between two neighbor SBs, while in the latter the channel model tends to change states. As seen in the plots, the more-realistic channel model, with a positive correlation between SBs ($ \alpha=0.9 $, blue), shows better performance than the i.i.d. model (black) and the $ \alpha=0.1 $ channel model (green) for all $ d $. This can be explained by the observation that in the $ \alpha=0.9 $ model, for every $ d $, there is a higher chance to see a pseudo-termination sequence ($ q $ consecutive SBs in state $ s_2 $). It is also observed that in this setting the one-sided decoder performs better than the two-sided, with generally a small advantage except in low d values where it is more significant.
	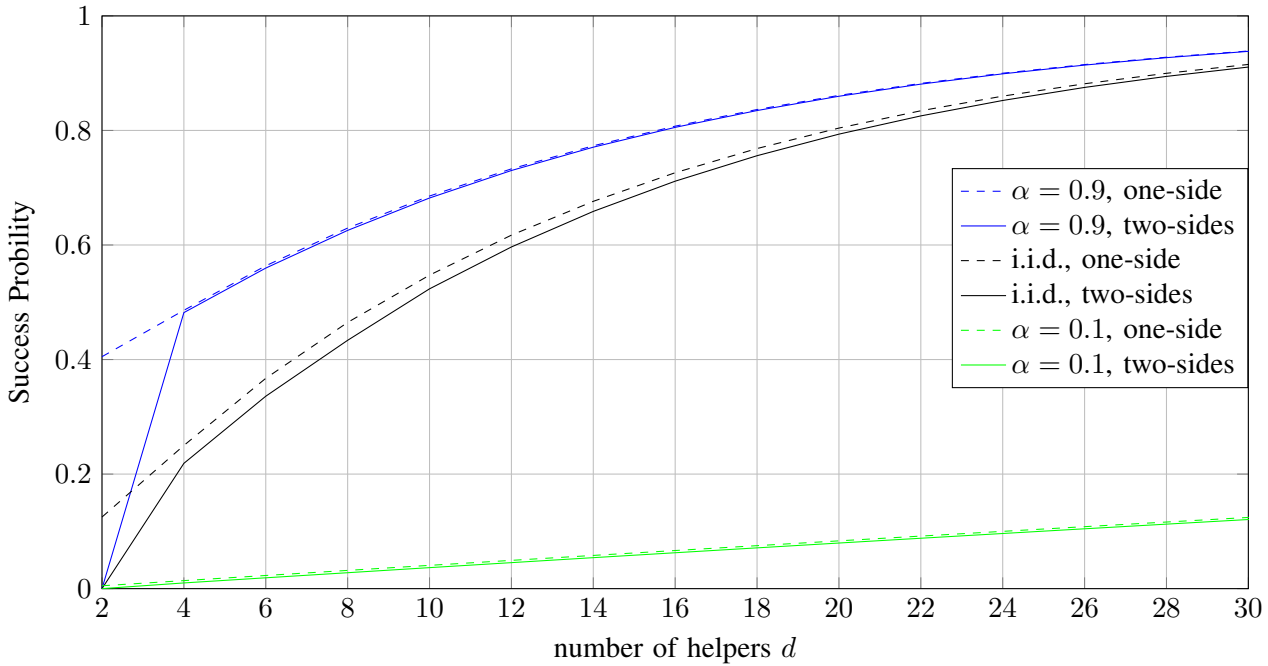
\begin{figure}
		\begin{tikzpicture}

\begin{axis}[%
width=6in,
height=3in,
at={(0,0)},
scale only axis,
xmin=2,
xmax=30,
xlabel={number of helpers $ d $},
ymin=0,
ymax=1,
ylabel={Success Probility},
xmajorgrids,
ymajorgrids,
legend style={
	legend cell align=left, 
	align=left, 
	at={(1,0.35)}, 
	anchor=south east, 
}
]
\addplot [color=blue, dashed]
  table[row sep=crcr]{%
2	0.405\\
4	0.486\\
6	0.5637195\\
8	0.62956359\\
10	0.68547202965\\
12	0.732942419688\\
14	0.773248302669195\\
16	0.80747098740681\\
18	0.836528585556621\\
20	0.861200642021762\\
22	0.882149048255504\\
24	0.899935799204049\\
26	0.915038070268278\\
28	0.927861018762764\\
30	0.938748653304149\\
32	0.947993062728216\\
34	0.955842252125136\\
36	0.962506796214705\\
38	0.968165488555501\\
40	0.972970138142544\\
42	0.977049642074549\\
44	0.980513443550543\\
46	0.983454467965541\\
48	0.985951615873571\\
50	0.988071879698237\\
52	0.989872140977006\\
54	0.991400696354942\\
56	0.992698553266588\\
58	0.993800530066001\\
60	0.994736190118778\\
62	0.995530634915463\\
64	0.996205177483682\\
66	0.996777914164995\\
68	0.997264210095861\\
70	0.997677111416997\\
72	0.998027695269698\\
74	0.998325366968681\\
76	0.998578112323872\\
78	0.998792711880327\\
80	0.99897492282381\\
82	0.999129633432134\\
84	0.999260994215797\\
86	0.999372529266117\\
88	0.99946723079806\\
90	0.999547639424106\\
92	0.999615912312727\\
94	0.999673881059986\\
96	0.999723100826817\\
98	0.999764892060221\\
100	0.999800375917661\\
};
\addlegendentry{$\alpha=0.9$, one-side}

\addplot [color=blue]
  table[row sep=crcr]{%
2	0\\
4	0.48195\\
6	0.5596695\\
8	0.62551359\\
10	0.6820515576\\
12	0.7300430845875\\
14	0.770786304014475\\
16	0.80538053219321\\
18	0.834753636763244\\
20	0.859693579063274\\
22	0.880869439974694\\
24	0.898849316844532\\
26	0.914115566144686\\
28	0.927077744325936\\
30	0.938083595200119\\
32	0.947428379060604\\
34	0.955362793800314\\
36	0.962099700528347\\
38	0.967819834117417\\
40	0.972676651882261\\
42	0.976800450467309\\
44	0.980301861389735\\
46	0.983274819014806\\
48	0.985799080587152\\
50	0.987942365924247\\
52	0.989762174174918\\
54	0.99130732638209\\
56	0.992619275232991\\
58	0.993733217134243\\
60	0.99467903644611\\
62	0.995482107207426\\
64	0.996163973859533\\
66	0.996742929231405\\
68	0.997234505291888\\
70	0.997651889834775\\
72	0.998006280275331\\
74	0.998307184049793\\
76	0.998562673676838\\
78	0.998779603323679\\
80	0.998963792686759\\
82	0.999120183120089\\
84	0.9992529701998\\
86	0.999365716281275\\
88	0.999461446068508\\
90	0.999542727759576\\
92	0.99961174194517\\
94	0.999670340108552\\
96	0.999720094296364\\
98	0.999762339292827\\
100	0.999798208428767\\
};
\addlegendentry{$\alpha=0.9$, two-sides}

\addplot [color=black, dashed]
  table[row sep=crcr]{%
2	0.125\\
4	0.25\\
6	0.3671875\\
8	0.46484375\\
10	0.54736328125\\
12	0.6171875\\
14	0.676239013671875\\
16	0.726181030273438\\
18	0.76841926574707\\
20	0.804141998291016\\
22	0.834354281425476\\
24	0.859906136989594\\
26	0.881516464054585\\
28	0.899793267250061\\
30	0.91525076283142\\
32	0.928323846077546\\
34	0.939380327035906\\
36	0.948731278826017\\
38	0.956639789654218\\
40	0.963328364777226\\
42	0.968985186667965\\
44	0.973769409512897\\
46	0.977815637001065\\
48	0.981237709387042\\
50	0.984131906376487\\
52	0.986579655947197\\
54	0.988649825318101\\
56	0.990400658522408\\
58	0.991881415098362\\
60	0.993133756002017\\
62	0.994192915734573\\
64	0.995088693661968\\
66	0.995846292417419\\
68	0.996487026975289\\
70	0.997028924346024\\
72	0.997487230764496\\
74	0.997874840641487\\
76	0.998202659346803\\
78	0.998479910030891\\
80	0.998714393117366\\
82	0.998912705767249\\
84	0.999080427489505\\
86	0.999222277120041\\
88	0.999342245585738\\
90	0.999443708188829\\
92	0.999529519570731\\
94	0.999602094027127\\
96	0.999663473433967\\
98	0.999715384695463\\
100	0.999759288330394\\
};
\addlegendentry{i.i.d., one-side}

\addplot [color=black]
  table[row sep=crcr]{%
2	0\\
4	0.21875\\
6	0.3359375\\
8	0.43359375\\
10	0.5234375\\
12	0.5965576171875\\
14	0.658660888671875\\
16	0.711410522460938\\
18	0.755905151367188\\
20	0.793555736541748\\
22	0.825404524803162\\
24	0.852335870265961\\
26	0.875113964080811\\
28	0.894378511235118\\
30	0.910671218764037\\
32	0.924450728809461\\
34	0.936104665277526\\
36	0.945960905803076\\
38	0.954296763773527\\
40	0.961346764931477\\
42	0.967309260528054\\
44	0.972352005150441\\
46	0.97661687585002\\
48	0.980223864402941\\
50	0.983274453123922\\
52	0.985854470053301\\
54	0.988036503741439\\
56	0.9898819455145\\
58	0.991442716714296\\
60	0.992762729501335\\
62	0.99387912231931\\
64	0.994823304781977\\
66	0.995621841379898\\
68	0.996297198870036\\
70	0.996868378375989\\
72	0.99735144998293\\
74	0.997760004868041\\
76	0.998105537687089\\
78	0.998397769975553\\
80	0.998644923663171\\
82	0.998853952397211\\
84	0.999030737182724\\
86	0.999180251844106\\
88	0.999306702963206\\
90	0.999413648231141\\
92	0.99950409654362\\
94	0.999580592655962\\
96	0.999645288779559\\
98	0.999700005134161\\
100	0.999746281159593\\
};
\addlegendentry{i.i.d., two-sides}

\addplot [color=green, dashed]
  table[row sep=crcr]{%
2	0.005\\
4	0.014\\
6	0.0229595\\
8	0.03183719\\
10	0.04063368965\\
12	0.049349933528\\
14	0.057986777447195\\
16	0.0665450224121543\\
18	0.0750254325572287\\
20	0.0834287464735762\\
22	0.0917556843506048\\
24	0.100006952486052\\
26	0.108183246141926\\
28	0.11628525136054\\
30	0.124313646126722\\
32	0.132269101118884\\
34	0.140152280201469\\
36	0.147963840754657\\
38	0.155704433901613\\
40	0.163374704671125\\
42	0.170975292119484\\
44	0.178506829426534\\
46	0.18596994397534\\
48	0.193365257421361\\
50	0.200693385754865\\
52	0.207954939358924\\
54	0.21515052306445\\
56	0.222280736203215\\
58	0.229346172659433\\
60	0.236347420920264\\
62	0.243285064125489\\
64	0.250159680116488\\
66	0.256971841484634\\
68	0.263722115619139\\
70	0.270411064754424\\
72	0.277039246017005\\
74	0.283607211471939\\
76	0.290115508168833\\
78	0.296564678187422\\
80	0.302955258682729\\
82	0.309287781929812\\
84	0.315562775368103\\
86	0.321780761645339\\
88	0.327942258661095\\
90	0.334047779609922\\
92	0.340097833024094\\
94	0.346092922815959\\
96	0.352033548319917\\
98	0.357920204334003\\
100	0.363753381161105\\
};
\addlegendentry{$\alpha=0.1$, one-side}

\addplot [color=green]
  table[row sep=crcr]{%
2	0\\
4	0.00995\\
6	0.0189095\\
8	0.02778719\\
10	0.0366562616\\
12	0.0453797265875\\
14	0.054075536123675\\
16	0.0626512417881215\\
18	0.0711812929557757\\
20	0.0796082728727249\\
22	0.0879788046809135\\
24	0.0962573733489448\\
26	0.104473265438987\\
28	0.112604607878743\\
30	0.120669887483258\\
32	0.128655725205424\\
34	0.136573877972471\\
36	0.144416253420574\\
38	0.1521904128697\\
40	0.159891552632824\\
42	0.167524616688367\\
44	0.175086845497959\\
46	0.182581548679204\\
48	0.190007235711041\\
50	0.197366194125651\\
52	0.204657719685941\\
54	0.211883459247537\\
56	0.219043194255752\\
58	0.226138168514063\\
60	0.233168461992592\\
62	0.240135063514875\\
64	0.247038235076883\\
66	0.253878802817915\\
68	0.260657138269572\\
70	0.267373962442077\\
72	0.274029711321979\\
74	0.280625036706736\\
76	0.287160411027828\\
78	0.293636439307564\\
80	0.300053613043469\\
82	0.306412504524057\\
84	0.312713613554323\\
86	0.318957488499618\\
88	0.325144630836053\\
90	0.331275570557374\\
92	0.337350806740779\\
94	0.343370854528929\\
96	0.34933620812734\\
98	0.355247370081958\\
100	0.36110482824757\\
};
\addlegendentry{$\alpha=0.1$, two-sides}

\end{axis}

\end{tikzpicture}%
		\caption{\label{Fig:GEC}Plots corresponding to the one-sided and two-sided semi-global-decoding success probabilities in \eqref{Eq:Suc1LProb} and \eqref{Eq:Suc2Prob} with $ q=3 $ for the SB Gilbert-Elliot channel.}
	\end{figure}
\end{example}

\begin{example}\label{Ex:4States}
	Consider a generalization of the SB Gilbert-Elliot channel from Example~\ref{Ex:GEC}, where $ \mathcal E=\{e_1=0.175,e_2=0.35,e_3=0.42,e_4=0.47\} $, and 
	\begin{align}\label{Eq:P1}
	P_1=\begin{pmatrix}
	0		&	0.5			& 0.5				&	0\\
	\beta	&\alpha			& 1-\alpha-2\beta	&	\beta\\
	\beta	&1-\alpha-2\beta&\alpha				& 	\beta\\
	0		&	0.5			& 0.5				&	0
	\end{pmatrix}
	\end{align}
	for some parameters $ 0<\alpha,\beta<1 $ such that $ \alpha+2\beta\leq 1 $. Like in Example~\ref{Ex:GEC}, we use the same $ (l=3,r=6,t=1) $ SC-LDPCL protograph, whose thresholds imply that for every $ i\in\{1,2,3,4\} $, $a_i=\{i\}$; $ q=3 $ remains for the value $e_2=0.35$, as before.
	The states $ s_1 $ and $ s_4 $ (corresponding to $ a_1 $ and $ a_4 $, respectively) are added to the channel model as extreme states, where $ s_1 $ is a very good state (local decoding/termination), and $ s_4 $ is a very bad state (anti-termination). To reflect the fact that these extreme states are likely rare, we set 
	\[
	\beta = 0.01.
	\]
	It can be verified that for $ \beta=0.01$ and every $ 0\leq \alpha\leq 0.98 $, the stationary distribution of $ P_1 $ is given by $ \tfrac1{2\beta+1}\cdot(\beta,0.5,0.5,\beta)=(0.0098,0.4902,0.4902,0.0098)$. Hence, for the i.i.d. channel model we set 
	\begin{align}\label{Eq:P2}
	P_2=\begin{pmatrix}
	0.0098&0.4902&0.4902&0.0098\\
	0.0098&0.4902&0.4902&0.0098\\
	0.0098&0.4902&0.4902&0.0098\\
	0.0098&0.4902&0.4902&0.0098
	\end{pmatrix}.
	\end{align}
	
	Figure~\ref{Fig:4States} plots the decoding success lower bound of \eqref{Eq:Suc1RProb} and \eqref{Eq:Suc2Prob} as a function of $ d $ for  $ P_1$ and $P_2$ in \eqref{Eq:P1} and \eqref{Eq:P2}, respectively. 
	Similarly to Figure~\ref{Fig:GEC}, the more-realistic channel model, with a positive correlation between SBs ($ \alpha=0.9 $, blue), shows better performance than the i.i.d. model (black) and the $ \alpha=0.1 $ channel model (green) for all $ d$.
	In contrast to Figure~\ref{Fig:GEC}, here the two-sided decoder outperforms the one-sided when $ d $ is large. A possible reason is that the addition of the states $ \tilde{s}_1,\tilde{s}_4  $ introduces a benefit for the two attempts at pseudo-termination offered by the two-sided decoder. 
	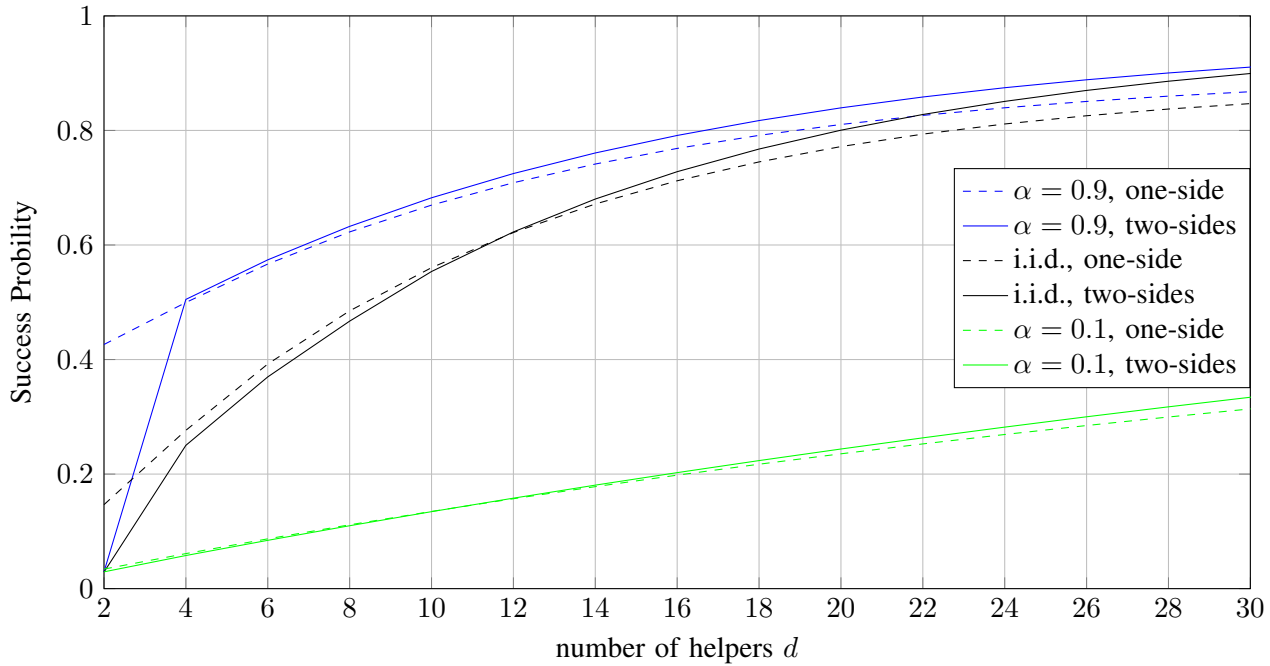
\begin{figure}
		\begin{tikzpicture}

\begin{axis}[%
width=6in,
height=3in,
at={(0,0)},
scale only axis,
xmin=2,
xmax=30,
xlabel={number of helpers $ d $},
ymin=0,
ymax=1,
ylabel={Success Probility},
xmajorgrids,
ymajorgrids,
legend style={
	legend cell align=left, 
	align=left, 
	at={(1,0.35)}, 
	anchor=south east, 
}
]
\addplot [color=blue, dashed]
  table[row sep=crcr]{%
2	0.426274509803921\\
4	0.499632274509803\\
6	0.566793994086274\\
8	0.62282604301881\\
10	0.669625228479611\\
12	0.708712471752033\\
14	0.74135861158571\\
16	0.768625063650278\\
18	0.791398335221836\\
20	0.810418847932612\\
22	0.826305010912928\\
24	0.839573326971212\\
26	0.850655185519475\\
28	0.859910888228413\\
30	0.867641363424712\\
32	0.874097950097801\\
34	0.879490569620991\\
36	0.883994550872123\\
38	0.887756330657125\\
40	0.890898214772937\\
42	0.893522354505027\\
44	0.895714067846355\\
46	0.897544613419753\\
48	0.899073507291548\\
50	0.900350458002407\\
52	0.901416982728581\\
54	0.902307757119375\\
56	0.903051742697789\\
58	0.903673128479203\\
60	0.904192117422735\\
62	0.904625583284949\\
64	0.904987619232047\\
66	0.905289996047424\\
68	0.905542544832181\\
70	0.90575347664123\\
72	0.905929649447254\\
74	0.906076791112243\\
76	0.906199685616022\\
78	0.906302328596583\\
80	0.906388057259251\\
82	0.906459658878405\\
84	0.906519461419422\\
86	0.906569409227227\\
88	0.906611126242291\\
90	0.9066459687994\\
92	0.906675069725821\\
94	0.906699375172632\\
96	0.906719675376692\\
98	0.906736630353412\\
100	0.906750791355673\\
};
\addlegendentry{$\alpha=0.9$, one-side}

\addplot [color=blue]
  table[row sep=crcr]{%
2	0.0293137254901961\\
4	0.50501956862745\\
6	0.574056502925489\\
8	0.632502797790189\\
10	0.682414241997006\\
12	0.724743306443847\\
14	0.760687286124973\\
16	0.791249554533367\\
18	0.817270291895932\\
20	0.839455277623416\\
22	0.858397667038957\\
24	0.87459622222846\\
26	0.8884705786043\\
28	0.900374020332802\\
30	0.91060417316684\\
32	0.91941195598577\\
34	0.927009076197722\\
36	0.933574307324618\\
38	0.93925874795896\\
40	0.944190228585995\\
42	0.948477005446165\\
44	0.95221085778575\\
46	0.955469685768976\\
48	0.958319690385027\\
50	0.960817203362587\\
52	0.963010223971729\\
54	0.964939710288304\\
56	0.966640664718684\\
58	0.968143047081694\\
60	0.969472543109835\\
62	0.970651211688339\\
64	0.971698030351569\\
66	0.972629355379582\\
68	0.973459310181014\\
70	0.974200113426401\\
72	0.974862356537366\\
74	0.975455238582027\\
76	0.975986765325806\\
78	0.976463918097867\\
80	0.976892797221881\\
82	0.97727874399672\\
84	0.977626444573594\\
86	0.977940018540913\\
88	0.978223094579684\\
90	0.978478875176419\\
92	0.97871019206548\\
94	0.978919553808576\\
96	0.979109186697458\\
98	0.979281069979885\\
100	0.979436966252747\\
};
\addlegendentry{$\alpha=0.9$, two-sides}

\addplot [color=black, dashed]
  table[row sep=crcr]{%
2	0.146629124544858\\
4	0.276420349593485\\
6	0.392172207342438\\
8	0.485007589230803\\
10	0.560414789111791\\
12	0.621740425617048\\
14	0.671590727796619\\
16	0.712113847668071\\
18	0.745055255110135\\
20	0.771833413778511\\
22	0.793601449582003\\
24	0.811296744175491\\
26	0.825681296476039\\
28	0.837374537502465\\
30	0.846880004400028\\
32	0.854607023858224\\
34	0.860888338637444\\
36	0.865994436217387\\
38	0.870145196593567\\
40	0.873519360802286\\
42	0.876262227909897\\
44	0.87849191191596\\
46	0.88030442800655\\
48	0.881777827183167\\
50	0.882975557313633\\
52	0.883949195339358\\
54	0.884740668293928\\
56	0.885384058775056\\
58	0.885907072617559\\
60	0.886332231968642\\
62	0.886677845141992\\
64	0.886958795014759\\
66	0.887187179917578\\
68	0.887372834615769\\
70	0.88752375381636\\
72	0.88764643643807\\
74	0.887746165469305\\
76	0.887827235465467\\
78	0.887893137482124\\
80	0.887946709407683\\
82	0.887990258169218\\
84	0.888025659073918\\
86	0.888054436564008\\
88	0.888077829862652\\
90	0.888096846337672\\
92	0.888112304881065\\
94	0.888124871172324\\
96	0.888135086344091\\
98	0.888143390284535\\
100	0.888150140579926\\
};
\addlegendentry{i.i.d., one-side}

\addplot [color=black]
  table[row sep=crcr]{%
2	0.0289330649599325\\
4	0.25020672472462\\
6	0.369676261835534\\
8	0.467060446212284\\
10	0.553693694842037\\
12	0.62281695903099\\
14	0.680233974313184\\
16	0.727996118295151\\
18	0.767546455629149\\
20	0.800449465712742\\
22	0.827865577024302\\
24	0.850748002502905\\
26	0.869890259288517\\
28	0.885938246305614\\
30	0.899422765249232\\
32	0.910780381240209\\
34	0.920370034485095\\
36	0.928487376083498\\
38	0.935376206817402\\
40	0.94123782153696\\
42	0.946238664638477\\
44	0.950516572394401\\
46	0.954185875805846\\
48	0.957341574038244\\
50	0.960062748672721\\
52	0.962415359176626\\
54	0.964454533341559\\
56	0.966226445371931\\
58	0.967769857202011\\
60	0.969117384608871\\
62	0.970296538314153\\
64	0.971330581006145\\
66	0.972239233669233\\
68	0.973039258465071\\
70	0.973744940405959\\
72	0.974368485983996\\
74	0.974920353597209\\
76	0.97540952790559\\
78	0.975843748041705\\
80	0.976229697799323\\
82	0.976573164453851\\
84	0.976879171668719\\
86	0.977152090962194\\
88	0.977395735408756\\
90	0.977613438594898\\
92	0.977808121314135\\
94	0.977982348048069\\
96	0.978138374921694\\
98	0.978278190527224\\
100	0.978403550769601\\
};
\addlegendentry{i.i.d., two-sides}

\addplot [color=green, dashed]
  table[row sep=crcr]{%
2	0.0341176470588235\\
4	0.0611616862745098\\
6	0.0869294136941176\\
8	0.111442628725478\\
10	0.134762314155686\\
12	0.156946767031424\\
14	0.178051364926043\\
16	0.198128738548086\\
18	0.217228926264788\\
20	0.235399512732435\\
22	0.252685755754747\\
24	0.269130703967499\\
26	0.284775307005084\\
28	0.299658519236687\\
30	0.31381739781586\\
32	0.327287195577438\\
34	0.340101449186057\\
36	0.352292062859016\\
38	0.363889387933539\\
40	0.374922298513356\\
42	0.385418263404878\\
44	0.395403414535188\\
46	0.404902612030006\\
48	0.413939506118332\\
50	0.422536596020664\\
52	0.430715285969093\\
54	0.438495938499716\\
56	0.44589792515066\\
58	0.452939674692288\\
60	0.459638719009879\\
62	0.466011736753138\\
64	0.47207459486129\\
66	0.477842388067193\\
68	0.483329476478834\\
70	0.488549521331812\\
72	0.493515519001813\\
74	0.498239833361774\\
76	0.502734226564297\\
78	0.507009888325952\\
80	0.51107746378638\\
82	0.514947080011562\\
84	0.518628371207231\\
86	0.522130502705215\\
88	0.525462193782411\\
90	0.528631739369223\\
92	0.531647030701496\\
94	0.534515574967371\\
96	0.537244513997976\\
98	0.539840642048484\\
100	0.542310422713806\\
};
\addlegendentry{$\alpha=0.1$, one-side}

\addplot [color=green]
  table[row sep=crcr]{%
2	0.0293137254901961\\
4	0.0578038823529412\\
6	0.0843194362588235\\
8	0.109813572183915\\
10	0.134394225669291\\
12	0.15797693092698\\
14	0.180704955713044\\
16	0.202536234231215\\
18	0.22356766857528\\
20	0.243787247138646\\
22	0.26326319403767\\
24	0.281999933878527\\
26	0.300048017798203\\
28	0.317420434686644\\
30	0.334156538822393\\
32	0.350273402292324\\
34	0.365803022537918\\
36	0.380764033256872\\
38	0.395183380166141\\
40	0.409079866555177\\
42	0.422476794542083\\
44	0.435392393950394\\
46	0.447847204565683\\
48	0.459858510492346\\
50	0.471444658741292\\
52	0.482621825402086\\
54	0.493406547914383\\
56	0.503813849143576\\
58	0.51385872706803\\
60	0.523555069722281\\
62	0.532916535641344\\
64	0.541955929359104\\
66	0.550685725290136\\
68	0.55911771128083\\
70	0.567263303388027\\
72	0.575133345254062\\
74	0.582738299917241\\
76	0.590088139567335\\
78	0.597192464757525\\
80	0.604060446378932\\
82	0.610700901769392\\
84	0.617122266662228\\
86	0.623332645494471\\
88	0.629339800374715\\
90	0.635151185765279\\
92	0.640773946937302\\
94	0.646214945029846\\
96	0.651480760626826\\
98	0.65657771275179\\
100	0.661511865049646\\
};
\addlegendentry{$\alpha=0.1$, two-sides}

\end{axis}

\end{tikzpicture}%
		\caption{\label{Fig:4States}Plots corresponding to the one-sided and two-sided semi-global-decoding success probabilities in \eqref{Eq:Suc1LProb} and \eqref{Eq:Suc2Prob} for the SB generalized Gilbert-Elliot channel.}
	\end{figure}
\end{example}

\subsection{Code Design}\label{Sub:code}

We now perform a threshold analysis and performance evaluation of semi-global decoding over the SB Markov-varying channel for a family of SC-LDPCL protographs sharing the same code rate and node degrees, however, differ in edge spreading. 
We focus on protographs in which the SBs are symmetric (see Definition~\ref{Def:SymSB}), which have the convenient property in which left and right helper SBs have identical structure. We first establish a connection between symmetric SBs and their degree profile.

\begin{definition}\label{Def:symdeg}
	A degree profile  $ \udc \in \mathbb N^{a},\, \udv \in \mathbb N^{b}$ is said to be \emph{symmetric} if by possibly reordering the elements of $  \udc $ and $  \udv $ we get vectors $ \tilde{\underline d}_C$ and $ \tilde{\underline d}_V $, respectively, such that
	\begin{align*}
	\tilde{d}_{C,i} + \tilde{d}_{C,a+1-i} = b,\quad \forall 1\leq i \leq a,\\
	\tilde{d}_{V,j} + \tilde{d}_{V,b+1-j} = a,\quad \forall 1\leq j \leq b.
	\end{align*}
\end{definition}

\begin{example}
	The protomatrix in Example~\ref{Ex:361_SCLDPCL} has a symmetric degree profile.
\end{example}
\begin{example}
	The following protomatrix has a symmetric degree profile:
	\begin{align*}
	A = \begin{pmatrix}
	0&1&0&1&1&1&0&0\\
	0&0&0&0&1&1&1&0\\
	0&1&0&0&1&1&1&1
	\end{pmatrix}.
	\end{align*}
	Indeed, $a=3,b=8$, $ \udc(A)=(4,3,5)$ and $\udv(A) = (0,2,0,1,3,3,2,1)$ and after sorting  $ \udc(A)$ and  $ \udv(A)$ in ascending and descending orders (Definition~\ref{Def:symdeg} allows arbitrary reordering), respectively, we get $ \tilde{\underline d}_C(A) = (3,4,5),\,\tilde{\underline d}_V(A) =(3,3,2,2,1,1,0,0)$. 
\end{example}

\begin{lemma}\label{Lemma:symdegsum}
	For every symmetric degree profile $ \udc \in \mathbb N^{a},\, \udv \in \mathbb N^{b}$ we have
	\[
	\|\udc\|_1 = \frac{ab}{2} = \|\udv\|_1.
	\]
\end{lemma}

\begin{proof}
	Let $ \tilde{\underline d}_C=(\tilde{d}_{C,1},\tilde{d}_{C,2},\ldots,\tilde{d}_{C,a}) $ be the sorted version of $ \udc $. Then,
	\begin{align*}
	2\|\udv\|_1
	&= 2\|\udc\|_1 \\
	&= 2\|\tilde{\underline d}_C\|_1\\
	&= \sum_{i=1}^a \tilde{d}_{C,i}+\tilde{d}_{C,a+1-i}\\
	&= ab .
	\end{align*}
\end{proof}

\begin{corollary}\label{Coro:nec_symdeg}
	If $ a $ and $ b $ are both odd, then no symmetric degree profile $ \udc \in \mathbb N^{a},\, \udv \in \mathbb N^{b}$ exists.
\end{corollary}

We now connect the definition of symmetric SBs in Definition~\ref{Def:SymSB} with symmetric degree profiles from Definition~\ref{Def:symdeg}. Recall that symmetric SBs simplify the analysis since left and right helper SBs have an identical structure. This leads to a simpler description of the semi-global-decoding process (e.g., unified left and right SB thresholds). The following proposition shows that if we are interested in symmetric SBs we only need to consider symmetric degree profiles.
\begin{proposition}\label{Prop:symSB}
	If a SB $\left (B_\mathrm{left}\;;\;B_\mathrm{loc}\;;\;B_\mathrm{right}\right )\in\{0,1\}^{(t+l)\times r}$ is symmetric then the degree-profile of $ B_\mathrm{left} $ is symmetric. In addition, if $ t\leq 2 $, then the converse holds as well.
\end{proposition}

\begin{proof}
	See Appendix~\ref{App:symSB}.	
\end{proof}

We now proceed to performance comparison between different codes over the SB Markov-varying channel. We focus on protographs with a coupling parameter $  t\leq 2$ and symmetric SBs.
Let $ l=4$ and $r=6 $ be the VN and CN degrees, respectively, of the base matrix $ B $. 
We consider $ t\in\{0,1,2\} $, where $ t=0 $ corresponds to an isolated SB (i.e., not coupled to its neighbors), and $ t\in\{1,2\} $ corresponds to a SB in a proper SC-LDPC code with SB locality.

In the case of $ t=0 $, the protograph consists of isolated SBs, each being a $ (l,r) $ code. Since there are no coupling checks, all of the SB's thresholds coincide
$\epsilon^*_3=\epsilon^*_2=\epsilon^*_1=0.5061$.

For $ t=1 $, we have only a single edge-spreading rule that induces symmetric SBs
	\begin{align}\label{Eq:edgespread1_t=1}
	\left (
	\begin{array}{c}
	\phantom{
		\begin{array}{cc}
		1&1
		\end{array}
	}
	B_\mathrm{left} 
	\phantom{
		\begin{array}{cc}
		1&1
		\end{array}
	}
	\\
	\hline
	\phantom{
		\begin{array}{c}
		1\\1\\1
		\end{array}
	}
	B_\mathrm{loc} 
	\phantom{
		\begin{array}{c}
		1\\1\\1
		\end{array}
	} \\
	\hline
	B_\mathrm{right}
	\end{array}\right )=
	\begin{pmatrix}
	1&1&1&0&0&0\\
	\hline
	1&1&1&1&1&1\\
	1&1&1&1&1&1\\
	1&1&1&1&1&1\\
	\hline
	0&0&0&1&1&1
	\end{pmatrix}.
	\end{align}
All other assignments to $ B_{\mathrm{left}} $ with $ t=1 $ and $ r=6 $ violate the necessary condition for symmetric SBs in Proposition~\ref{Prop:symSB}. The thresholds for the SB in \eqref{Eq:edgespread1_t=1} are 
$\epsilon^*_1=0.4294,\; \epsilon^*_2=0.4788,\;\epsilon^*_3=0.5474$.

If we use $ t=2 $ coupling check nodes, then more edge-spreading rules are possible. Let $ (d_{C,1},d_{C,2})=\udc(B_{\mathrm{left}})$ be the check degree profile of $ B_{\mathrm{left}} $. 
In view of Proposition~\ref{Prop:symSB}, a SB is symmetric if and only if $d_{C,1}+d_{C,2}=r=6  $. Thus, we set $ d_{C,1}\in \{1,2,3\} $ and $  d_{C,2} = 6-d_{C,1}$. We can assume w.l.o.g. that the first row in $ B_{\mathrm{left}} $ is left aligned, meaning it has $d_{C,1}  $ consecutive ones followed by $  6-d_{C,1} $ zeros.
After setting the first row of $ B_{\mathrm{left}} $, we set the overlap between the first an second rows of $ B_{\mathrm{left}} $. Let $ j\in\{1,2,\ldots,r-(d_{C,2}-1)\} $. The second row of $ B_{\mathrm{left}} $ starts with $ j-1 $ zeros, followed by $d_{C,2}  $ ones, and additional $ r-(d_{C,2}-j+1) $ zeros.
In all of the $ t=2 $ designs we have $B_{\mathrm{loc}}=1^{2\times6}$, while $ B_{\mathrm{left}}  $ is given by Table~\ref{Tbl:Bleft}.
All other assignments to $	B_{\mathrm{left}} $  with $ r=6,t=2 $ (i.e., not in Table~\ref{Tbl:Bleft}) either induce non-symmetric SBs or are permuted versions of those in Table~\ref{Tbl:Bleft} and thus are equivalent. 
The local threshold of all $ t=2 $ SBs is $ \epsilon^*_1=0.2 $ (the threshold of the $(l=2,r=6)  $ LDPC ensemble); the other thresholds, alongside with their $q$ value (see \eqref{Eq:q}) for erasure parameter $0.435$, are given in Table~\ref{Tbl:Bleft}. Note that for $ \udc=(3,3) $ and $ j=1 $ (line 8 in Table~\ref{Tbl:Bleft}) we have $ 0.435>\epsilon^*_2 $, and thus $ q(0.435) =\infty$. For all other designs  $ 0.435<\epsilon^*_2 $ and $ q(0.435) <\infty$. Also shown are the global thresholds $ \epsilon_G^* $ (i.e., the threshold of the coupled protograph \cite{RamCass18a}) for protographs with $ M=50 $ SBs.

\begin{table}
	\caption{\label{Tbl:Bleft}All possible edge connections for the coupling check nodes in $l=4,r=6$ (rate $ 0.33 $) SC-LDPCL protographs with symmetric SBs.}
	\begin{center}
		\[\begin{array}{c|c|c|c|c|c|c|c|c|c}
		\#&t	&\udc	&	j	&		B_{\mathrm{left}} 									&\epsilon^*_1&\epsilon^*_2&\epsilon^*_3&\epsilon^*_G&q(0.435)	\\
				\hline
		1&0		&		&		&															&0.5061		&0.5061		&0.5061		&0.5061		&1	\\[3mm]
		2&1		&3		&	1	&	\begin{pmatrix}1&1&1&0&0&0\end{pmatrix}					&0.4294		&0.4788		&0.5474		&0.5564		&2	\\[3mm]
		3&2		&(1,5)	&	1	&	\begin{pmatrix}1&0&0&0&0&0\\1&1&1&1&1&0	\end{pmatrix}	&0.2		&0.4368		&0.5320		&0.5836		&10	\\[3mm]
		4&2		&(1,5)	&	2	&	\begin{pmatrix}1&0&0&0&0&0\\0&1&1&1&1&1	\end{pmatrix}	&0.2		&0.4423		&0.5175		&0.5655		&5	\\[3mm]
		5&2		&(2,4)	&	1	&	\begin{pmatrix}1&1&0&0&0&0\\1&1&1&1&0&0	\end{pmatrix}	&0.2		&0.4386		&0.5918		&0.6114		&11	\\[3mm]
		6&2		&(2,4)	&	2	&	\begin{pmatrix}1&1&0&0&0&0\\0&1&1&1&1&0	\end{pmatrix}	&0.2		&0.4442		&0.5722		&0.5966		&7	\\[3mm]
		7&2		&(2,4)	&	3	&	\begin{pmatrix}1&1&0&0&0&0\\0&0&1&1&1&1	\end{pmatrix}	&0.2		&0.4488		&0.5594		&0.5863		&5	\\[3mm]
		8&2		&(3,3)	&	1	&	\begin{pmatrix}1&1&1&0&0&0\\1&1&1&0&0&0	\end{pmatrix}	&0.2		&0.4345		&0.4991		&0.6338		&\infty \\[3mm]
		9&2		&(3,3)	&	2	&	\begin{pmatrix}1&1&1&0&0&0\\0&1&1&1&0&0	\end{pmatrix}	&0.2		&0.4411		&0.5974		&0.6104		&10	\\[3mm]
		10&2	&(3,3)	&	3	&	\begin{pmatrix}1&1&1&0&0&0\\0&0&1&1&1&0	\end{pmatrix}	&0.2		&0.4463		&0.5802		&0.5982		&7	\\[3mm]
		11&2	&(3,3)	&	4	&	\begin{pmatrix}1&1&1&0&0&0\\0&0&0&1&1&1	\end{pmatrix}	&0.2		&0.4507		&0.5688		&0.5900		&6	\\[3mm]
		\end{array}.\]
	\end{center}
\end{table}

In what follows, we evaluate the semi-global performance of the above codes over the SBMV channel. We consider two SB Gilbert-Elliot channels, both with the transition matrix $ P_1 $ in \eqref{Eq:P1P2} with $ \alpha=0.9 $. For the first one $ \mathcal E_1 = \{0.435,0.54\} $ and for the second $   \mathcal E_2 = \{0.435,0.57\}$. 
Note that in some of the possible designs in Table~\ref{Tbl:Bleft}, we have that the erasure rate in the "bad" state ($ 0.54 $ and $ 0.57 $ for $ \mathcal E_1 $ and $ \mathcal E_2 $, receptively) is greater than $\epsilon^*_3 $. In view of Section~\ref{Sub:SBMVPerf}, this means that in these particular cases, "bad" SBs are anti-termination SBs (i.e, in state $ s_4 $, see Section~\ref{Sub:BlockState}). 
For these cases, \eqref{Eq:Suc1RProb}, \eqref{Eq:Suc1LProb} simplify for one-sided decoding to
\begin{subequations}
\begin{align*}
p_R(d) = p_L(d) = \left \{\begin{array}{ll}
0.5\cdot\alpha^{q-1} & d\geq q-1\\
0			& \text{otherwise}
\end{array}\right .,
\end{align*}
and for the two-sided mode \eqref{Eq:Suc2Prob} simplifies to
\begin{align*}
p_2(d)= \left \{\begin{array}{ll}
0.5\cdot(2\alpha^{q-1}-\alpha^{2(q-1)}) & \tfrac{d}{2}\geq q-1\\
0			& \text{otherwise}
\end{array}\right ..
\end{align*}
\end{subequations}
In view of Section~\ref{Sub:SBMVPerf}, for a given SBMV channel, the performance of semi-global decoding is determined by: 1) the code thresholds, and in particular the induced $ Q $ matrix in \eqref{Eq:Q}; 2) the parameter $ q(e) $ which is the minimal number of SBs with erasure rate $ e<\epsilon^*_2 $ needed for pseudo-termination. Hence, the performance of some of the possibilities in Table~\ref{Tbl:Bleft} coincide. 
Further, some of the designs have properties (thresholds and $ q $) that are dominated by others, e.g., the design in line 3 is dominated by the design in line 7 (worse threshold and $ q $).

In Figure~\ref{Fig:CodeDesign} we eliminate these multiplicities and dominated cases and plot the performance of three designs corresponding to lines 1,2, and 6 in Table~\ref{Tbl:Bleft}. 
The success probabilities of the two-sided semi-global decoding mode over the SB Gilbert-Elliot channel (with $ \mathcal E_1 $ and $ \mathcal E_2 $) are plotted for these three design options. The plots show the trade-offs between the different designs. For example, the uncoupled protograph ($ t=0 $) has the highest success probability for no helper SBs ($ d=0 $), however, since the SBs are uncoupled, then adding helper SBs does not improve the performance.
In addition, on one hand, the $ t=1 $ design shows high success probabilities for the $ \mathcal E_1 $ channel, and on the other hand the performance is poor for the $ \mathcal E_2 $ channel. This degradation is a result of a relatively small $ \epsilon_3^* $ threshold ($ 0.5474 $, see line 2 in Table~\ref{Tbl:Bleft}). 
Furthermore, the $ t=2 $ design with $ (d_{C,1},d_{C,2})=(2,4),\,j=2 $ (line 6 in Table~\ref{Tbl:Bleft}) shows the best performance for the $ \mathcal{E}_2 $ channel (for $ d\geq 18 $), while it requires high values for $ d $ (see $ q=7 $ in Table~\ref{Tbl:Bleft}).
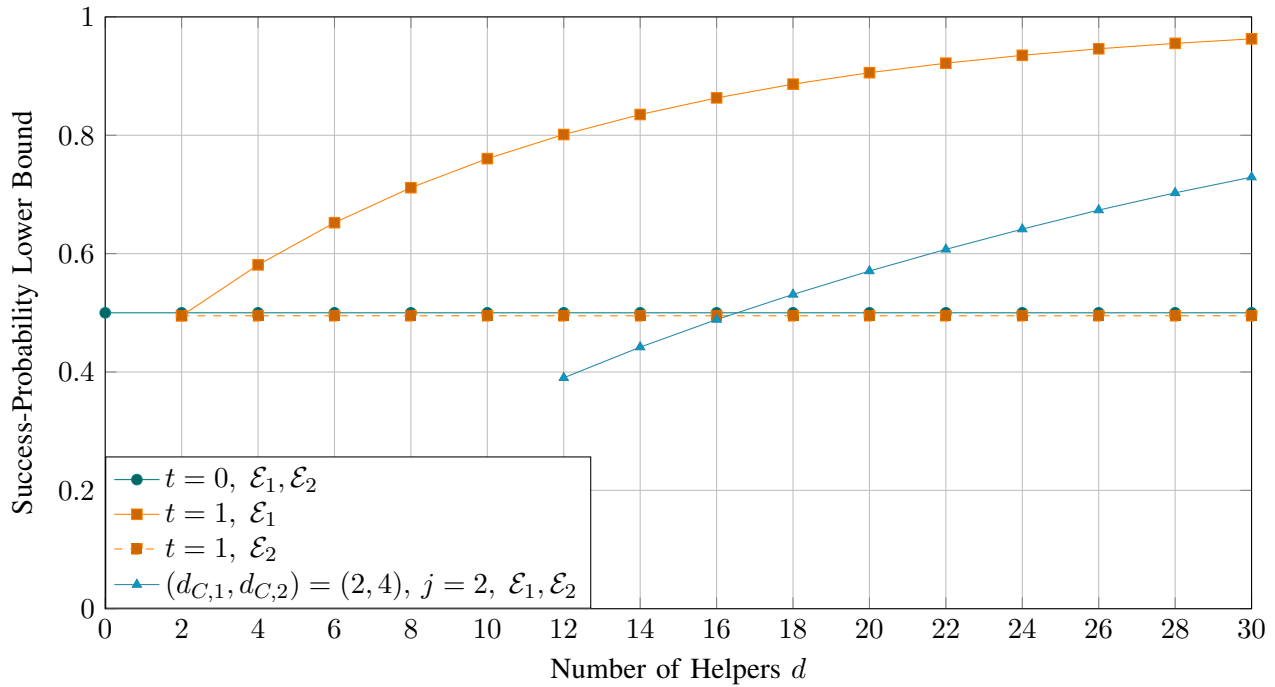
\begin{figure}
	\begin{center}
		\begin{tikzpicture}

\pgfplotscreateplotcyclelist{Mylist}{%
	teal,every mark/.append style={fill=teal!80!black},mark=*\\
	orange,every mark/.append style={fill=orange!80!black},mark=square*\\
	orange,every mark/.append style={fill=orange!80!black},mark=square*,dashed\\
	cyan!60!black,every mark/.append style={fill=cyan!80!black},mark=triangle*\\
	cyan!60!black,every mark/.append style={fill=cyan!80!black},mark=triangle*,dashed\\
	red!70!white,mark=star\\
	lime!80!black,every mark/.append style={fill=lime},mark=diamond*\\
	red,densely dashed,every mark/.append style={solid,fill=red!80!black},mark=*\\%
	yellow!60!black,densely dashed,
	every mark/.append style={solid,fill=yellow!80!black},mark=square*\\%
	black,every mark/.append style={solid,fill=gray},mark=otimes*\\%
	blue,densely dashed,mark=star,every mark/.append style=solid\\%
	red,densely dashed,every mark/.append style={solid,fill=red!80!black},mark=diamond*\\%
}
\begin{axis}[%
width=6in,
height=3.1in,
at={(0,0)},
scale only axis,
xmin=0,
xmax=30,
xlabel={Number of Helpers $ d $},
ymin=0,
ymax=1,
ylabel={Success-Probability Lower Bound},
xmajorgrids,
ymajorgrids,
cycle list name=Mylist,
legend style={
	legend cell align=left, 
	align=left,
	at={(0,0)}, 
	anchor=south west
}
]
\addplot
table[row sep=crcr]{%
	0   0.5\\
	2	0.5\\
	4	0.5\\
	6	0.5\\
	8	0.5\\
	10	0.5\\
	12	0.5\\
	14	0.5\\
	16	0.5\\
	18	0.5\\
	20	0.5\\
	22	0.5\\
	24	0.5\\
	26	0.5\\
	28	0.5\\
	30	0.5\\
};
\addlegendentry{$t=0 ,\;\mathcal{E}_1,\mathcal E_2$}

\addplot
  table[row sep=crcr]{%
2	0.495\\
4	0.58095\\
6	0.652239\\
8	0.711400275\\
10	0.76049700435\\
12	0.801241373583\\
14	0.835054290373455\\
16	0.86311493687261\\
18	0.886401892175229\\
20	0.905727259011752\\
22	0.921764984790537\\
24	0.935074364650245\\
26	0.946119546161224\\
28	0.955285715877351\\
30	0.962892532223585\\
};
\addlegendentry{$ t=1,\; \mathcal{E}_1$}
\addplot 
table[row sep=crcr]{%
	2	0.495\\
	4	0.495\\
	6	0.495\\
	8	0.495\\
	10	0.495\\
	12	0.495\\
	14	0.495\\
	16	0.495\\
	18	0.495\\
	20	0.495\\
	22	0.495\\
	24	0.495\\
	26	0.495\\
	28	0.495\\
	30	0.495\\
};
\addlegendentry{$ t=1, \;\mathcal E_2$}

\addplot 
table[row sep=crcr]{%
	12	0.3902\\
	14	0.442\\
	16	0.4886\\
	18	0.531\\
	20	0.5706\\
	22	0.6072\\
	24	0.6414\\
	26	0.6736\\
	28	0.7027\\
	30	0.7292\\
};
\addlegendentry{$ (d_{C,1},d_{C,2})=(2,4),\,j=2,\; \mathcal{E}_1,\mathcal{E}_2$}



\end{axis}

\end{tikzpicture}%
		\caption{\label{Fig:CodeDesign}The success probabilities of the two-sided semi-global decoding mode over the SB Gilbert-Elliot channel for the designs in lines 1,2, and 6 in Table~\ref{Tbl:Bleft}.}
	\end{center}
\end{figure}

\begin{example}\label{Ex:l4r16}
	Table~\ref{Tbl:l4r16} lists the thresholds of different designs for high-rate SC-LDPCL protographs with $l=4 $ and $ r=16 $ (rate $ 0.75 $) (all designs induce symmetric SBs). The global thresholds are calculated for protographs with $ M=8 $ SBs.
	\begin{table}
		\caption{\label{Tbl:l4r16}Thresholds of SC-LDPCL designs with $l=4,r=16$.}
		\begin{center}
			\[\begin{array}{c|c|c|c|c|c|c|c|c}
			\#&t	&\udc	&	j	&\epsilon^*_1&\epsilon^*_2&\epsilon^*_3&\epsilon^*_G&q(0.16)	\\
			\hline
			1&0		&		&		&0.1931		&0.1931		&0.1931		&0.1931		&1	\\[3mm]
			2&1		&8		&	1	&0.1568		&0.1794		&0.2036		&0.2119		&2	\\[3mm]
			3&2		&(5,11)	&	1	&0.0667		&0.1601		&0.2082		&0.2288		&33	\\[3mm]
			4&2		&(6,10)	&	2	&0.0667		&0.1609		&0.2121		&0.2298		&15	\\[3mm]
			5&2		&(7,9)	&	3	&0.0667		&0.1613		&0.2142		&0.2304		&13	\\[3mm]
			6&2		&(8,8)	&	4	&0.0667		&0.1615		&0.2149		&0.2306		&12	\\[3mm]
			7&2		&(6,10)	&	1	&0.0667		&0.1598		&0.1999		&0.2326		&\infty	\\[3mm]
			8&2		&(7,9)	&	2	&0.0667		&0.1604		&0.1999		&0.2330		&25 \\[3mm]
			9&2		&(8,8)	&	3	&0.0667		&0.1605		&0.1999		&0.2331		&21	\\[3mm]
			10&2	&(7,9)	&	1	&0.0667		&0.1592		&0.1667		&0.2368		&\infty		\\[3mm]
			11&2	&(8,8)	&	2	&0.0667		&0.1594		&0.1667		&0.2368		&\infty		\\[3mm]
			12&2	&(8,8)	&	1	&0.0667		&0.1428		&0.1429		&0.2431		&\infty		\\[3mm]
			\end{array}.\]
		\end{center}
	\end{table}
	
\end{example}

\section{Summary}\label{Sec:Summ}

This paper analyzes and characterizes semi-global decoding of SC-LDPC protographs with SB locality. 
We perform a two-step analysis of the decoder: 1) from the perspective of a single SB decoded in the process, 2) the full decoding process with $ d $ helper SBs. In the former, we define SB thresholds and prove theoretical results about the incoming and outgoing DE values to and from SBs. In the latter, we use the single-SB results to prove characterization results about semi-global thresholds over memoryless channels, and to derive lower bounds on the semi-global performance over the SB Markov-varying channel.

The most interesting line of future work is to find relations between the semi-global thresholds studied here and the classical global threshold. These may lead to a systematic way to design protographs with improved global thresholds. We find encouragement for this direction in a variation of the SG decoder that empirically approaches the performance of the global decoder. This mode progresses similarly to the SG decoder, but instead of stopping at the target (Phase 1 in Figure~\ref{Fig:BFSG}), it proceeds outward toward the termination SBs (Phase 2), and then returns to the target again (Phase 3).

\begin{figure}
	\begin{center}
		\begin{tikzpicture}[>=latex]\label{Tikz:BFSG decoding}
		
		\tikzstyle{SB}=[rectangle,very thick,draw, rounded corners, minimum width=1.2cm,minimum height=0.5cm,fill=white]
		\tikzstyle{W}=[rectangle,opacity=0.2,fill=blue!50!white,draw,dashed,blue,thick, minimum width=1.2cm,minimum height=8mm]
		\pgfmathsetmacro{\x}{1}
		\pgfmathsetmacro{\y}{1.1}
		
		\begin{scope}[yshift=-12*\y mm]
		\node (p1) {Phase 1:};
		\node (SB1) [SB,right=\x mm of p1]  {\footnotesize $1$};
		
		\node (dots1) [right=0 mm of SB1 ]  {$ \cdots$};
		\node (SB2) [SB,right=0 mm of dots1 ]  {\footnotesize $m-2$};
		\node (SBc1) [SB,right=0 mm of SB2 ]  {\footnotesize $m$$-1$};
		\node (SBc2) [SB,right=0 mm of SBc1,fill=red!50!white  ]  {\footnotesize $m$};
		\node (SBc3) [SB,right=0 mm of SBc2 ]  {\footnotesize $m$$+1$};
		\node (SBMm1) [SB,right=0 mm of SBc3 ]  {\footnotesize $m+2$};
		\node (dots2) [right=0 mm of SBMm1 ]  {$ \cdots$};
		\node (SBM) [SB,right=0 mm of dots2 ]  {\footnotesize $M$};
		\node (W1) [W, minimum width=7mm] at(dots1) {};\node (W2) [W, minimum width=7mm] at(dots2) {};
		\node (x1) at(W1.north east) {};\node (x2) [right =5*\x mm of x1] {};
		\node (xM) at(W2.north west) {};\node (xm1) [left =5*\x mm of xM] {};
		\draw [thick, ->,>=latex,blue] (x1)--(x2);\draw [thick, ->,>=latex,blue] (xM)--(xm1);
		\end{scope}
		
		\begin{scope}[yshift=-2*12*\y mm]
		\node (p2) {Phase 2:};
		\node (SB1) [SB,right=\x mm of p2]  {\footnotesize $1$};
		
		\node (dots1) [right=0 mm of SB1 ]  {$ \cdots$};
		\node (SB2) [SB,right=0 mm of dots1 ]  {\footnotesize $m-2$};
		\node (SBc1) [SB,right=0 mm of SB2 ]  {\footnotesize $m$$-1$};
		\node (SBc2) [SB,right=0 mm of SBc1,fill=red!50!white ]  {\footnotesize $m$};
		\node (SBc3) [SB,right=0 mm of SBc2 ]  {\footnotesize $m$$+1$};
		\node (SBMm1) [SB,right=0 mm of SBc3 ]  {\footnotesize $m+2$};
		\node (dots2) [right=0 mm of SBMm1 ]  {$ \cdots$};
		\node (SBM) [SB,right=0 mm of dots2 ]  {\footnotesize $M$};
		\node (W1) [W] at(SBc2) {};
		\node (x1) at(W1.north east) {};\node (x2) [right =5*\x mm of x1] {};
		\node (x3) at(W1.north west) {};\node (x4) [left =5*\x mm of x3] {};
		\draw [thick, ->,>=latex,blue] (x1)--(x2);\draw [thick, ->,>=latex,blue] (x3)--(x4);
		\end{scope}
		
		\begin{scope}[yshift=-3*12*\y mm]
		\node (p3) {Phase 3:};
		\node (SB1) [SB,right=\x mm of p3]  {\footnotesize $1$};
		
		\node (dots1) [right=0 mm of SB1 ]  {$ \cdots$};
		\node (SB2) [SB,right=0 mm of dots1 ]  {\footnotesize $m-2$};
		\node (SBc1) [SB,right=0 mm of SB2 ]  {\footnotesize $m$$-1$};
		\node (SBc2) [SB,right=0 mm of SBc1,fill=red!50!white  ]  {\footnotesize $m$};
		\node (SBc3) [SB,right=0 mm of SBc2 ]  {\footnotesize $m$$+1$};
		\node (SBMm1) [SB,right=0 mm of SBc3 ]  {\footnotesize $m+2$};
		\node (dots2) [right=0 mm of SBMm1 ]  {$ \cdots$};
		\node (SBM) [SB,right=0 mm of dots2 ]  {\footnotesize $M$};
		\node (W1) [W, minimum width=7mm] at(dots1) {};\node (W2) [W, minimum width=7mm] at(dots2) {};
		\node (x1) at(W1.north east) {};\node (x2) [right =5*\x mm of x1] {};
		\node (xM) at(W2.north west) {};\node (xm1) [left =5*\x mm of xM] {};
		\draw [thick, ->,>=latex,blue] (x1)--(x2);\draw [thick, ->,>=latex,blue] (xM)--(xm1);
		\end{scope}
		
		\end{tikzpicture}
		
	\end{center}
	\caption{\label{Fig:BFSG} The three phases in the modified semi-global decoder.}
\end{figure}
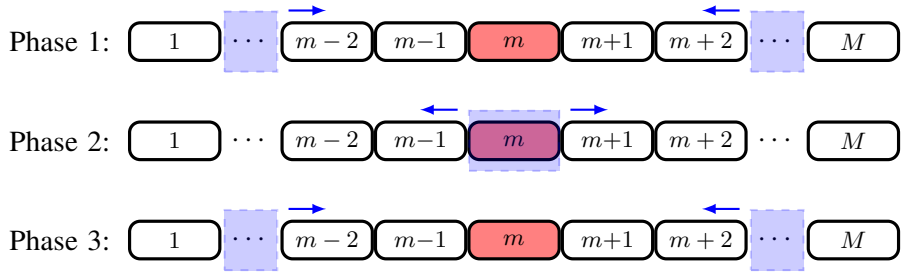

\section{Acknowledgments} \label{Sec:Ack}
This work was supported in part by the Israel Science Foundation and in part by the US-Israel Binational Science Foundation.

\appendices

\section{Proof of Lemma~\ref{Prop:symSB}}\label{App:symSB}
We can assume w.l.o.g. that $ \udc\left (B_\mathrm{left}\right ) $ and $ \udv\left (B_\mathrm{left}\right ) $ are sorted in ascending order (if not, apply a sorting permutation).
Consider the row (CN) and column (VN) permutations $ \pi_C $ and $ \pi_V $, respectively, defined by
\begin{align*}
\begin{array}{llll}
\pi_C(i) &= &t+1-i,&1\leq i\leq t,\\
\pi_V(j) &= &r+1-j,&1\leq j\leq r.
\end{array}
\end{align*}
Since we assumed that  $ \udc\left (B_\mathrm{left}\right ) $ and $ \udv\left (B_\mathrm{left}\right ) $ are sorted in ascending order, then the row and column permutations that transform $ B_\mathrm{left} $ into $ B_\mathrm{left} $ are given by $ \pi_C $ and $  \pi_V$, respectively. Thus, in view of \eqref{Eq:dBU=dBD1}, for every $ i\in \{1,\ldots,t\} $ we have
\begin{align*}
r-\left [ \udc\left (B_\mathrm{left}\right )\right ]_{i}
&=\left [ \udc\left (B_\mathrm{right}\right )\right ]_{i}\\
&=\left [ \pi_C\left (\udc\left (B_\mathrm{left}\right )\right )\right ]_i \\
&=\left [ \udc\left (B_\mathrm{left}\right )\right ]_{\pi_C\left (i \right )}\\
&=\left [ \udc\left (B_\mathrm{left}\right )\right ]_{t+1-i}
\end{align*}
Similarly, it can be shown that for every $ j\in\{1,\ldots,r\} $, $ t-\left [ \udv\left (B_\mathrm{left}\right )\right ]_{j}=\left [ \udv\left (B_\mathrm{left}\right )\right ]_{r+1-j} $, which in view of Definition~\ref{Def:symdeg} implies that the degree-profile of $ B_\mathrm{left} $ is symmetric.

For the converse, assume that $ t\leq 2 $ and that the degree-profile of $ B_\mathrm{left} $ is symmetric. If $ t=1 $, then from Corollary~\ref{Coro:nec_symdeg}, $ r $ is even. We get that $  B_\mathrm{left} $ is a single-row matrix with $ r/2 $ ones and $ r/2 $ zeros. Since $  B_\mathrm{right}=1- B_\mathrm{left}$, then $  B_\mathrm{left} $  also has $ r/2 $ ones and $ r/2 $ zeros. The permutation that swaps between ones and zeros will transform $ B_\mathrm{left} $ into $ B_\mathrm{right} $. In view of Definition~\ref{Def:SymSB}, the SB is symmetric.

Now assume $ t=2 $. Let $ \udv\left(B_\mathrm{left}\right )=\left (d_1,\ldots,d_r\right )$ be the VN degree profile of $ B_\mathrm{left} $. For $ i\in\{0,1,2\} $, let $ \mathcal J_i =\{1\leq j \leq r\colon  d_j = i\}$. $ \mathcal J_2 $ and $ \mathcal J_0 $ are the column indices in which $ B_\mathrm{left} $ has only ones and zeros, respectively.
Consider the row-swap permutation $ \pi_C(1) =2,\pi_C(2) =1$ and any column (VN) permutation $ \pi_V $ such that $\pi_V(\mathcal J_2) = \mathcal J_0$.
Since $ \underline d $ is symmetric, then $ |I_2|=|I_0| $ and such a column permutation exists. It can be verified that these permutations transform $ B_\mathrm{left} $ into $ B_\mathrm{right} $, and this the SB is symmetric.

\bibliographystyle{IEEEtran}
\bibliography{SG_Dec_Refs}
\end{document}